\documentclass[aps,prx,onecolumn,superscriptaddress,amsmath,nofootinbib,amssymb,11pt,floatfix]{revtex4-2_custom}

\numberwithin{equation}{section}
\usepackage{url}
\usepackage{graphicx}
\usepackage{subfigure}
\usepackage{float,caption,subcaption,wrapfig} 
\usepackage{mathtools}
\usepackage{bbm}
\usepackage{ulem}   
\usepackage[usenames,dvipsnames]{xcolor}
\usepackage{braket}
\usepackage{bbold}
\usepackage{comment}
\usepackage{amsthm} 
\usepackage{amssymb}
\usepackage{bm}
\usepackage[colorlinks, linkcolor=blue, citecolor=blue]{hyperref}
\usepackage{multirow}
\usepackage{makecell}

\usepackage{tikz}
\usetikzlibrary{cd}

\renewcommand{\thesection}{\arabic{section}}
\renewcommand{\thesubsubsection}{\alph{subsubsection}}

\newcommand{\Tr}{{\rm Tr\,}}

\renewcommand{\phi}{\varphi}
\renewcommand{\epsilon}{\varepsilon}

\newcommand{\uv}[1]{\mathbf{\hat{#1}}}
\renewcommand{\sc}[1]{\mathcal{#1}}
\newcommand{\bb}[1]{\mathbb{#1}}

\def\mod{\text{ mod }}

\def\U{\mathrm{U}(1)}
\def\pf{\mathsf{Pf}}
\def\cA{\mathcal{A}}

\def\bZ{\mathbb{Z}}
\def\cT{\mathcal{T}}

\def\tmfd{{M^{(3)}}}
\def\fmfd{{M^{(4)}}}

\newtheorem{theorem}{Theorem}[section]

\newtheorem{lemma}[theorem]{Lemma}
\newtheorem{corollary}[theorem]{Corollary}

\theoremstyle{definition}
\newtheorem{definition}[theorem]{Definition}

\newtheorem{remark}[theorem]{Remark}

\newtheorem{example}[theorem]{Example}

\usepackage{tcolorbox}
\usepackage{lipsum}
\tcbuselibrary{skins,breakable}
\usetikzlibrary{shadings,shadows}

    {\endtcolorbox}

\begin{document}
\title{Ordering the topological order in the fractional quantum Hall effect}

\author{Meng Cheng}
\affiliation{Department of Physics, Yale University, New Haven, Connecticut 06511, USA}
\author{Seth Musser}
\affiliation{Condensed Matter Theory Center and Joint Quantum Institute, Department of Physics, University of Maryland, College Park, Maryland 20742, USA}
\author{Amir Raz}
\affiliation{University of Texas, Austin, Physics Department, Austin, Texas 78712, USA}
\author{Nathan Seiberg}
\affiliation{School of Natural Sciences, Institute for Advanced Study, Princeton, NJ 08540, USA}
\author{T. Senthil}
\affiliation{Department of Physics, Massachusetts Institute of Technology, Cambridge, Massachusetts 02139, USA}
\begin{abstract}\noindent
We discuss the possible topological order/topological quantum field theory of different quantum Hall systems.  Given the value of the Hall conductivity, we constrain the global symmetry of the low-energy theory and its anomaly.  Specifically, the one-form global symmetry and its anomaly are presented as the organizing principle of these systems. This information is powerful enough to lead to a unique minimal topological order (or a small number of minimal topological orders).   Almost all of the known experimentally discovered topological orders are these minimal theories. Since this work is interdisciplinary, we made a special effort to relate to researchers with different backgrounds by providing translations between different perspectives.
\end{abstract}  
\maketitle

\newpage

\tableofcontents

\section{Introduction}
 
The discovery of the fractional quantum Hall effect (FQHE) of electrons moving in two space dimensions revolutionized our understanding of the natural world. Apart from the precisely quantized Hall conductivity, FQHE systems exhibit several novel phenomena: they support quasiparticle excitations with fractional electric charge and fractional statistics; further, at their spatial boundary, there are dissipationless one-way current-carrying edge modes. These features are consequences of a non-trivial many-body ground state with subtle long-range quantum correlations, though there is no classical long-range order. 

The most familiar experimental realizations of the quantum Hall effects are found in high-mobility two-dimensional electron gases subject to a uniform magnetic field. The standard platforms are semiconductors such as GaAs or Si, or monolayer and bilayer graphene \cite{girvin2019modern}.  In recent years, FQHEs have also been observed~\cite{cai_signatures_2023,park_observation_2023,zeng_thermodynamic_2023, lu_fractional_2024} in two-dimensional Van der Waals moir\'e materials under zero magnetic field (known as the fractional quantum anomalous Hall or FQAH effect). There have also been observations of field-induced FQHE in related two-dimensional materials where a strong periodic potential is present~\cite{spanton2018observation,xie2021fractional,aronson2024displacement}. 

Universal aspects of the FQHE phenomenon are well understood. They can be captured by continuum topological quantum field theories (TQFT) enriched by a global $\U$ symmetry (corresponding to conservation of electric charge). Equivalently, they can be understood as topological orders enriched by $\U$ symmetry as a special case of symmetry enriched topological orders (SETs). A mature microscopic theory of the FQHE in the classic setting of two-dimensional electron gases in a magnetic field also exists. Microscopic theories of FQAH and related lattice systems are under development (see  \cite{Zhang2018_nearly_flat,Ledwith2019_FCITBG,Repellin2019_TBGChern,Wilhelm2020_TBGFCI,Abouelkomsan2019_MoireFCI,Wu2018_TMD_topo,Yu2019_TMD_topo,Devakul2021_TMD_topo,li2021spontaneous,reddy2023toward,yu2024fractional,dong2024theory,zhou2024fractional,dong2024anomalous,yu2024moir,huang2025fractional} for a sampling of recent papers).

\begin{figure*}
    \centering
    \subfigure[Classification of SETs]{
        \includegraphics[width=0.45\textwidth]{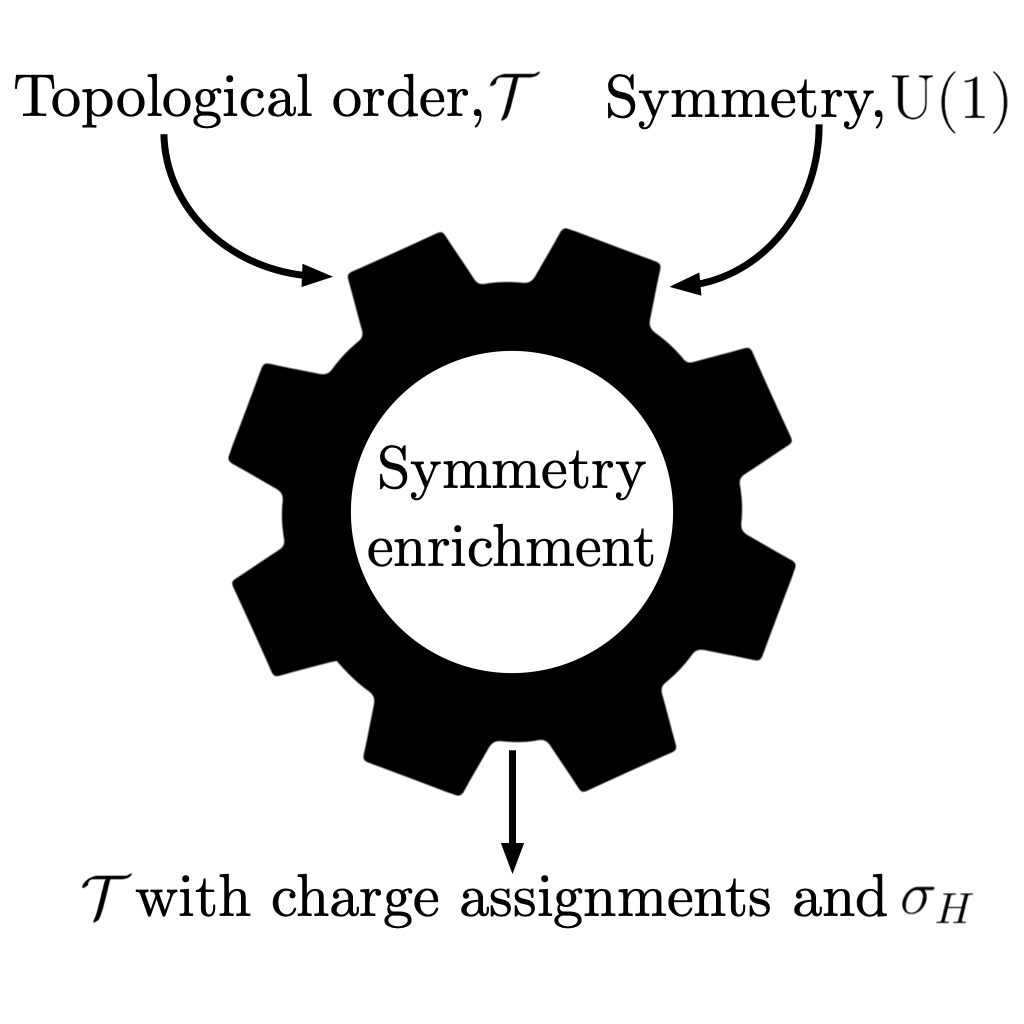}
        \label{fig:set}
            }
    ~ 
    \subfigure[Our approach]{
        \includegraphics[width=0.45\textwidth]{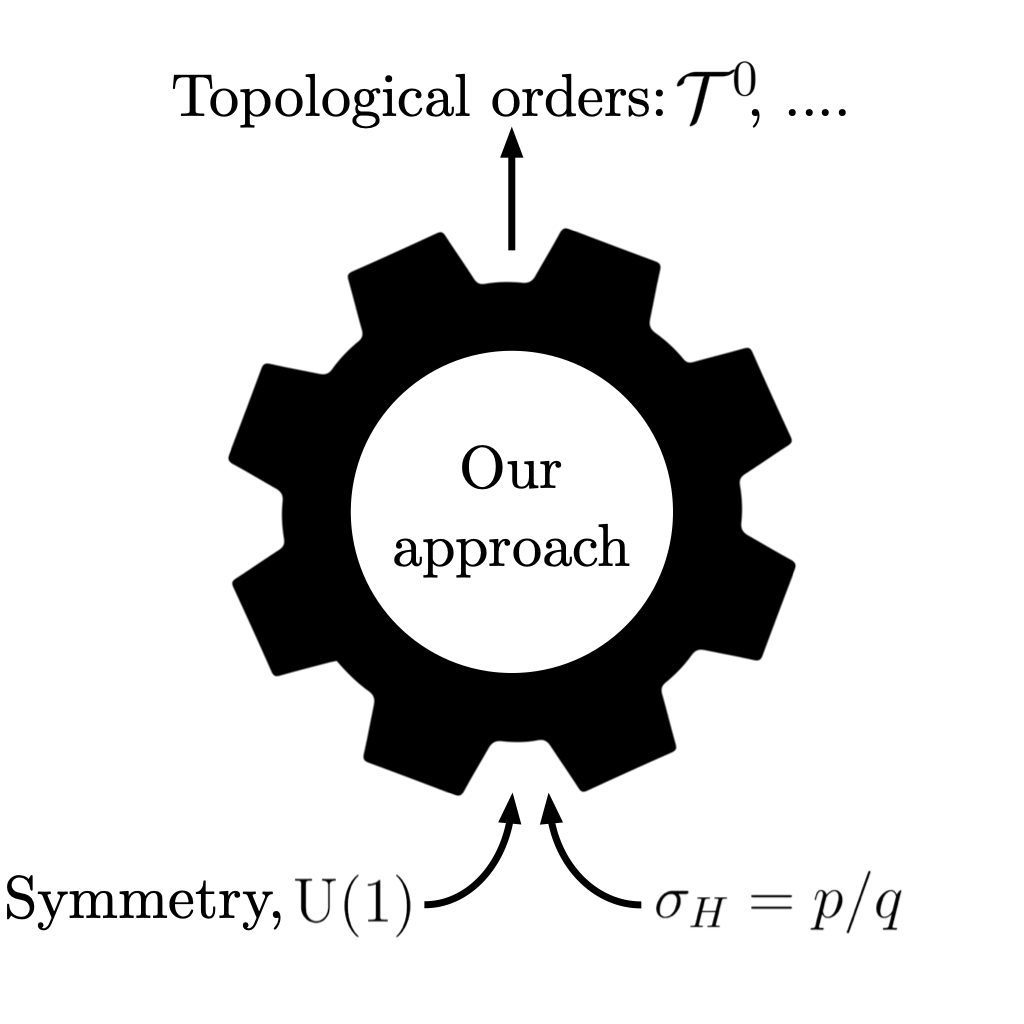}\
        \label{fig:consist}
            }
    \caption{A cartoon illustrating the difference between our approach and the usual classification of SETs. (a) Illustration of the usual approach taken to classify SETs. Here, the topological order is already known, and this information, along with the symmetry, is put into the classification. The classification then provides the possible consistent ways of assigning symmetry charges to the anyons. If the symmetry is $\U$, each symmetry assignment will come with a possibly distinct value of $\sigma_H$. In contrast, our approach, illustrated in (b), assumes the $\U$ symmetry and the value of $\sigma_H$ and returns a list of consistent topological orders. This is relevant experimentally as in this setting the topological order of the system is not already known.}
	\label{fig:contribution_cartoon}
\end{figure*}

In this context, it is interesting to ask how much of the structure of the universal low-energy theory of an FQHE system can be constrained given only the value of the quantized Hall conductivity $\sigma_H$.\footnote{ Several authors have explored general constraints on the TQFT given measurable data like $\sigma_H$ and $c_-$, see e.g., \cite{levin2009fractional,Kong:2020bmb,Dehghani:2020jls,Randal-Williams:2020hls,  Cheng:2022nds}.} This is often the first (and perhaps only) piece of experimental information available (for example, in the recent discoveries of the FQAH). What can we say about the low-energy TQFT given the Hall conductivity? We note that this is the opposite direction from the usual classification of SETs. There, we already know the topological order, and we merely want to understand how symmetry labels can be assigned to anyons. In this approach, the fractional Hall conductivity is a consequence of the symmetry assignment. Our contribution, illustrated in Fig.~\ref{fig:contribution_cartoon}, is the ability to go in the reverse direction. That is, given the value of the fractional Hall conductivity, can we produce a list of consistent topological orders and give recipes, either experimental or numerical, to distinguish between these possibilities? Answering this question will be our main focus. We will build on a variety of ideas going back to the very first theoretical papers to  organize the possible TQFTs with a given Hall conductance. 

Along the way, we will identify the {\it minimal} topological order (mTO)\footnote{See  \cite{musser_fractionalization_2024,jian_minimal_2024} for recent work on minimal topological order in other contexts. } --- defined to be the one with the smallest number of distinct anyons --- for any given value of the Hall conductivity. We will show that there is either a unique or a very small number of mTO theories at a given $\sigma_H$. Specifically for $\sigma_H = \frac{p}{q}$ with $p,q$ coprime integers, and $q$ odd, there is a unique minimal theory with $q$ distinct Abelian anyons. For $q$ even, there are four distinct minimal theories with $3q$ anyons ($q$ of which are non-Abelian and the rest are Abelian) and are variants of the familiar Pfaffian quantum Hall states. We discuss these mTOs in more detail in the summary of our results in Sec.~\ref{sec:summary}. Interestingly, almost all the observed FQHE states are described by these mTO theories. Thus, knowledge of $\sigma_H$ and the criterion of minimality is often enough to fully determine the low-energy TQFT  with no reference to a detailed microscopic theory. This may be useful as more FQHE systems are discovered, especially in lattice settings, with or without a weak magnetic field, where the microscopic physics is potentially rather different from the traditional FQHE systems. 

Our analysis will build on the old idea that the mere observation of a quantized fractional $\sigma_H$ guarantees the presence of fractionally charged Abelian anyon excitations. (Note that the quantization of $\sigma_H$ requires that current-carrying excitations are gapped in the bulk.) We will further restrict, as is standard, to situations in which there is a gap to all excitations so that the low-energy physics is described by a topological theory. We will infer several consequences of the presence of the anyons implied by the value of $\sigma_H$, and show that they lead to powerful constraints on the TQFT and, eventually, to a full ordering. 

One of the beauties of this topic is that many of the results can be derived independently from different perspectives. The topological order that characterizes FQHE states can be described in multiple equivalent ways. One of them is through a Chern-Simons theory \cite{Witten:1988hf}. A different perspective is to just focus on the data (such a braiding, fusion, etc.) needed to characterize the anyon excitations of the topological order. In its general form, this leads to a powerful (but abstract) description of the topological order in terms of a 
(unitary) modular tensor category (MTC) enriched by the global $\U$ symmetry \cite{barkeshli_symmetry_2019, zhang2025hierarchy}. A related perspective involves keeping track of the full set of gapped anyons and the action of global symmetries on them, as occurs in any specific microscopic setting.  Yet another perspective focuses on the generalized symmetries of the IR TQFT and their 't Hooft anomalies.  Finally, through the ``bulk-edge'' correspondence that relates Chern-Simons theories in 2+1d dimensions to rational conformal field theories in 1+1d dimensions \cite{Witten:1988hf,Moore:1988qv,Moore:1989vd}, we may fruitfully view quantum Hall phenomena through the lens of the conformal field theory (see, e.g. \cite{Frohlich:2000qs, Cappelli:2010jv}).

A particular focus will be on a point of view where the existence of Abelian anyon excitations is related to a (generalized) symmetry of the TQFT with an 't Hooft anomaly \cite{Gaiotto:2014kfa}. Indeed, it is well appreciated that a focus on symmetries and their anomalies leads to powerful nonperturbative constraints on the possible low-energy behaviors of physical systems. We note that in the underlying microscopic setting (the UV theory), we only have an ordinary global (also known as zero-form)  $\U$  symmetry and possibly some other ordinary symmetries (like translation). In contrast, the IR TQFT, at scales below the gap to anyon excitations,  has generalized global symmetries with certain anomalies.  These new IR symmetries can be standard emergent symmetries or, as we will discuss below, be a result of symmetry transmutation \cite{Transmutaion}. The presence of these symmetries in the ground states of the TQFT is associated with the existence of excited anyonic quasiparticle states. Our analysis will make use of the generalized global symmetries and anomalies of the IR theory, which are necessarily present in FQH systems.

Since the paper addresses questions of interest to theorists from diverse backgrounds, we will attempt to communicate our results from more than one of the various perspectives mentioned above. Naturally, we also include review material on topics that may be familiar to some, but not all, readers. Good basic introductions to quantum Hall physics and topological phases of matter can be found, for example, in \cite{Moore:1989vd,kitaev_anyons_2006,girvin2019modern,Witten:2015aoa}.

\subsection{Summary of main results}
\label{sec:summary}

The goal of this paper is to characterize the possible topological orders that are consistent with a given Hall conductivity. Throughout the course of our discussion, we will make frequent reference to various bosonic, spin, and spin$^{c}$ TQFTs. To make it clear what we mean when we refer to these, we have included Table~\ref{tab:defns} as a helpful reference.

\begin{table*}
    \centering
    \begin{tabular}{|c||c|c|}
        \hline
           Microscopic charge-$1$ particles&  Nature of local  particles & Type of TQFT \\
         \hline \hline
         Bosons& Integer-charged bosons & Non-spin or bosonic \\
        Bosons and fermions& Integer-charge bosons and fermions& Spin or Fermionic  \\
         Fermions& Fermions (bosons) have odd (even) charges & Spin$^c$ or electronic\\
         \hline
    \end{tabular}
    \caption{Types of physical systems and their low-energy TQFTs considered in this paper. The physical systems are characterized by specifying the nature of the fundamental charge carriers, which are taken to have charge one. The nature of composite particles formed by binding an integer number of these fundamental charge carriers is indicated in the second column. The last row describes electronic solids. The distinction between these physical systems is captured in the low-energy TQFT through the possible space-time manifolds in which they can be placed. We limit ourselves to orientable spacetimes. In spin (or Fermionic) TQFTs, the spacetime has to be a spin manifold, and the physical answers can depend on the choice of spin structure. A spin$^c$ TQFT is a spin TQFT that permits a particular $\U$ enrichment that corresponds to the structure of the allowed local particles as defined in the third row. It can be placed on nonspin manifolds, provided there is a specific background gauge field.  }
    \label{tab:defns}
\end{table*}

We first review the crucial observation that a fractional $\sigma_H$ implies the existence of a special fractionally charged Abelian anyon, which we call the vison. The vison, denoted by $v$, has charge $Q(v) = p/q = \sigma_H$, topological spin $h(v) = p/2q = \sigma_H/2$, and it determines the fractional $U(1)$ charges of all other anyons via braiding [i.e., $Q(x) \mod 1 = B(v,x)$ for any anyon $x$.] The existence of the vison can be argued in many ways: using the flux threading argument (see Sec. \ref{flux threading}), by showing that the response theory has a one-form symmetry with a particular anomaly  (see Sec. \ref{from response to}), by matching the anomalies of the magnetic translations in the UV theory (see Sec. \ref{sec:UVIR}), or through formal arguments on symmetry enrichment of MTCs (see Appendix \ref{charge-assignment-classification}.) We present many different approaches on the vison's existence to help build physical intuition, and to highlight a variety of complementary perspectives.

As the vison is an Abelian anyon, it generates a (one-form) symmetry group, which we denote ${\cal V}^{s,r} =\{1,v,v^2,\cdots , v^{s-1}\}$. Here $s = q\ell$ is the order of the group, and $r=p\ell$ specifies the anomaly. Note that the order of the group $s$ must be a multiple of $q$ to ensure that $v^s =1$, a transparent boson. 

Using this information, we can determine the minimal electronic or spin$^c$ topological orders given $\sigma_H = p/q$. These minimal orders all have the fewest number of anyons in ${\cal V}^{s,r}$. This means that the order of the vison is $s=q$. As mentioned earlier, if $q$ is odd, then there is a unique minimal theory, ${\cal V}^{q,p}$, which has $q$ anyons that are all powers of the Abelian vison. For $q$ even, there are four minimal theories with $3q$ anyons, all variants of the non-Abelian Pfaffian quantum Hall states. We emphasize that this result was not known for $q$ even and has already been useful in numerical studies of the $3/4$ FQHE \cite{huang_non-Abelian_2024}. We have included Table~\ref{tab:mTOs} as a helpful summary of our results. In addition to the experimentally relevant electronic systems we also considered minimal orders in the settings of: 3D topological insulator surface topological orders, bilayer electronic FQHE, and bosonic FQHE.

\begin{table*}
    \centering
    \begin{tabular}{|c||c|c|c|c|c|}
        \hline
         Symmetry and constraints& Min. Abelian sector& Min. TQFTs&  $\sc{D}^2_{{\cal T}}$&  $N_{{\cal T}}$ \\
         \hline \hline
         Spin$^c$, $\sigma_H = p/q$& & & & \\
        $q$ odd& $\mathcal{V}^{q,p}$  & $\mathcal{V}^{q,p}$ & $q$& $q$\\
         $q$ even& $\mathcal{V}^{2q,2p}$  & $\mathsf{Pf}^f_{q,p,n}$& $4q$& $4q$ ($3q$) for $n$ even (odd)\\
        \hline
         Spin$^c$ , 3D topological insulator surface TO& $\mathcal{V}^{4,2} $ &  T-Pfaffian& 8 & 6\\
         \hline
         Spin$^c$, $\U^{(1)}\times \U^{(2)}\times \bb{Z}_2^{{\cal CT}}$, $\sigma_{H}^{ab} = \frac{\delta_{ab}}{2}$& $\mathcal{V}^{4,2} \boxtimes \mathcal{V}^{4,2}$& $D(\bb{Z}_4)$& $16$& $16$\\
         \hline
        Bosons,  $\sigma_H = p/q$& & &  & \\
         $pq$ even&  $\mathcal{V}^{q,p}$  & $\mathcal{V}^{q,p}$&  $q$&  $q$\\
         $pq$ odd& $\mathcal{V}^{2q,2p} = \mathcal{A}^{2,2} \boxtimes \mathcal{A}^{q,p+q}$  & $\mathsf{Pf}^b_{q,p,n}$&  $4q$& $4q$ ($3q$) for $n$ even (odd)\\
         \hline
    \end{tabular}
    \caption{A summary of our results. This table gives the list of possible minimal TQFTs by quantum dimension $\sc{D}^2_{{\cal T}}$ and anyon count $N_{{\cal T}}$ for each FQH setup. See Appendix~\ref{app:anyontheory} for a description of the notation. In all the examples, we find that the minimal TQFTs by anyon count are a subset of the minimal TQFTs by quantum dimension; if there are non-Abelian options, as in the case of the odd $n$ theories above, they are precisely these options.}
    \label{tab:mTOs}
\end{table*}

We are further able to 1) develop a systematic algorithm to compile a list of all the possible topological orders that are consistent with a measured Hall conductivity $\sigma_H$, and (perhaps) other topological data, and 2) order the allowed topological orders for a given $\sigma_H$. We give more details in Sec.~\ref{sec:IR}.

\subsection{Outline} 

In  Sec.~\ref{QHP}, we provide an introduction to the physics of the quantum Hall effect (QHE), and in  Sec.~\ref{sec:tqft-basics}, we review some known facts about topological quantum field theory (TQFT).  The purpose of these sections is to introduce known techniques and terminology, which we will use later in the paper. In particular, we review the crucial result that a fractionally quantized Hall conductivity implies that the low-energy theory admits fractionally charged Abelian anyon excitations.

In  Sec.~\ref{from response to}, we relate the flux threading argument (reviewed in  Sec.~\ref{flux threading}), the fact that the response theory of the FQH system is not globally well defined, and the existence of a one-form global symmetry with a particular anomaly (reviewed in  Sec.~\ref{sec:symm_anom_gauge}). The anomalous one-form global symmetry is a reformulation of the statement that the system includes Abelian anyon excitations with particular properties.  This reformulation allows us to use the powerful machinery of TQFT.

In  Sec.~\ref{sec:UV}, we rephrase the known physics of charged particles in a homogeneous magnetic field in modern terms.  This allows us to identify the one-form global symmetry of the IR TQFT, which underlies the response theory, as originating from the magnetic translation group of the UV system. Another perspective relating the UV  translation to the IR one-form global symmetry is presented in  Sec.~\ref{conntranf}. 

Section~\ref{sec:symmetrylines} describes properties of the one-form global symmetry.  We review known material and present various extensions of it, which will be needed later.

In  Sec.~\ref{sec:IR}, we use the information about the global symmetry and its anomaly to organize the possible TQFTs with a given Hall conductance.  Given the value of $\sigma_H$, we present an algorithm leading to a minimal theory.

We summarize our results, compare them with the literature, and present further thoughts in  Sec.~\ref{sec: summary}.

Several appendices address more technical details, some of which are known, and others are new.  Appendix~\ref{app:anyontheory} reviews some known examples of TQFTs and introduces our notation.  Appendix~\ref{app:generator} describes how the one-form global symmetry that leads to $\sigma_H$ can be embedded in a larger one-form global symmetry of the system.  In Appendix~\ref{app:spinc-classification}, we review and extend the different ways a TQFT can be coupled to a background $\U$ gauge field (the electromagnetic field).  Appendix~\ref{app:gauging-quantum-dim} reviews how the gauging procedure, which is central in our discussion, affects the quantum dimension of the TQFT.  And in Appendix~\ref{app:mme}, we present a construction of a minimal modular extension of some of the premodular TQFTs that we discuss often.

\section{Quantum Hall preliminaries}
\label{QHP}

In this section, we review some background material on quantum Hall physics that we will need later in our discussion.  In the next section, we review background material on TQFTs.

\subsection{Quantum Hall effect}

We review\footnote{For a good review of quantum Hall physics, see, e.g.,   \cite{girvin2019modern}} the basic phenomenology of the quantum Hall effects (integer, fractional, spin, etc.). Since every quantum Hall effect has a signature in the Hall conductivity $\sigma_H$, we begin by defining it. In practice, $\sigma_H$ can be defined either via electrical transport or thermodynamic quantities.  Both of these are routinely used in experiments and give the same answer.

We begin by giving the transport definition of $\sigma_H$. Consider a 2+1d system of electrons at zero temperature. The conductivity tensor is defined as
\begin{equation}
\label{qhresp}
    J_i = \sum_j \sigma_{ij}E_j\,,
\end{equation}
where $i,j = x,y$ are indices for the spatial coordinates. 
In a system showing the Hall effect, the longitudinal conductivity $\sigma_{xx}=\sigma_{yy}=0$,  i.e., the electrical transport is dissipationless.   The conductivity tensor takes the following form:
\begin{equation}
\label{Halldefn} 
    \sigma_{ij} = \frac{\sigma_H}{2\pi} \begin{pmatrix}
        0 & -1\\
        1 & 0
    \end{pmatrix}_{ij}.
\end{equation}
The electrical current is thus perpendicular to the electric field. 

Now we give the thermodynamic definition. Consider a many-body state with a gap to excited states carrying electric current. If we change the charge density by a small amount $\delta \rho$, this will generically close the gap. However, if we also change the magnetic field by a small amount $\delta B$ then it is possible to maintain this gap even as we tune the density. The Hall conductivity is related to the coefficient of proportionality between these two changes. More precisely, $\sigma_H$ is defined by 
\begin{equation}
\label{streda}
    \frac{\sigma_H}{2\pi}  = \frac{d\rho}{dB}\,,
\end{equation}
where the derivative is taken while preserving the many-body gap. Equation~\eqref{streda} is known as the Streda formula.

Though the transport and thermodynamic definitions na\"{i}vely appear distinct, they are equivalent. At low energies and long wavelengths, \eqref{qhresp} will hold as a local equation, i.e., we can write
\begin{equation} 
J_i (\vec x, t) = \sum_j \sigma_{ij} E_j (\vec x, t)\, .
\end{equation} 
Taking the spatial divergence on both sides, and using the continuity/conservation equation $\partial_i J_i = - \partial_t \rho$, and the Maxwell equation $\sum_{ij} \epsilon_{ij} \partial_i E_j = -\partial_t B$, we get 
\begin{equation}\label{rhodB}
\partial_t \rho = \frac{\sigma_H}{2\pi}  \partial_t B\,.
\end{equation} 
Here, we have assumed that the magnetic field changes slowly over time, such that the local version of the relation between the current and the electric field holds.  Integrating both sides over a time interval $\delta t$, we arrive at the Streda formula.

Everything we have said up to this point applies to \textit{any} system with a gap to excited states carrying current. Henceforth, we specialize to the various quantum Hall systems.
In these systems the Hall conductivity $\sigma_H$ is quantized to rational numbers in units where the electron charge $e = 1$, and $\hbar = 1$.\footnote{It will be a consequence of our later discussion that any system with a gap to all excitations \textit{must} have $\sigma_H$ a rational number. So we have not lost much generality by specializing to quantum Hall systems.} Throughout the text we will write $\sigma_H=p/q$, where $p$ and $q$ are coprime integers. 

The observation of very precisely quantized $\sigma_H$ in real materials tells us that the quantum Hall effect does not require translation or other spatial symmetries. This is because these real materials have impurities (and/or, in recent realizations,  a strong periodic potential) that explicitly break translation symmetry.
Indeed, the very existence of plateaus of quantized $\sigma_H$ as the charge density is changed {\it requires} the presence of impurities (see, e.g.,  \cite{girvin2019modern}).

Nevertheless, in some cases it is useful to consider an idealized system with no impurities and, in an infinite system, continuous translation symmetry. To get a quantum Hall effect in such a system, we need a background magnetic field $B$ such that there is a rational Landau level filling fraction $\nu = 2\pi \rho/B$. This, of course, is the classic setting for quantum Hall phenomena.

In this paper, we will obtain constraints on the low-energy, i.e., infrared (IR), physics of systems with these two starting points; either $\sigma_H = p/q$, or translation symmetry with $\nu = p/q$. If $\sigma_H$ is rational, very general arguments, reviewed below, demonstrate that the IR theory must support fractionally charged anyon excitations, have associated topological order, and be described by a topological quantum field theory (TQFT). We will determine the minimal TQFT consistent with the given $\sigma_H$. On the other hand, when translation is assumed and a rational $\nu$ is given, we will analyze the global symmetries and anomalies of this microscopic, i.e., ultraviolet (UV), system. We consider the structure of IR TQFTs that preserve the UV symmetries and can match the anomalies. We show that the resulting theories are identical to the ones obtained by the consideration of TQFTs consistent with $\sigma_H = p/q$. 

That these two different starting points give the same results is not surprising. With the assumed charge $\U$ and translation symmetries, and the further assumption of a gapped ground state, $\sigma_H$ is exactly determined by the Landau level filling factor through $\sigma_H = \nu$.  This is readily seen from the aforementioned Streda formula. In the translationally invariant Landau level, the many-body ground state is gapped at some fixed Landau level fillings $\nu = p/q$. Now consider changing the external magnetic field by $\delta B$. To maintain the Landau level filling $p/q$, the density must change by $\delta \rho = \nu \delta B/2\pi$. The Streda formula then tells us that $\sigma_H = \nu$. This conclusion relies crucially on translation invariance. If impurities or even a periodic potential are present, $\sigma_H$ and $\nu$ do not have any direct relation. Indeed, the persistence of the quantized Hall effect in a range of $\nu$ in real samples is a clear-cut manifestation of the absence of such a relationship. As such, we will often discuss $\sigma_H = p/q$ instead of $\nu$, as the relation $\nu = \sigma_H$ means that the translation invariant case is a subset of the discussion with fixed $\sigma_H$.

Lastly, we note that the discussion here can be easily generalized to systems with multiple independent $\U$ symmetries, such as multilayer QH systems or 2+1d materials with multiple flavor (such as electron spin and/or ``valley") degrees of freedom. We will use multilayer systems as an example. In this case, the Hall conductivity is a symmetric matrix $\sigma_H^{ab}$, defined as
\begin{equation}\label{bilaryJE}
J_i^{(a)}=\sum_{j,b}\sigma_{ij}^{ab}E_j^{(b)}\qquad,\qquad \sigma_{ij}^{ab}=\frac{\sigma_H^{ab}}{2\pi}\begin{pmatrix}
        0 & -1\\
        1 & 0
    \end{pmatrix}_{ij}.
\end{equation}
Here $a,b$ are the layer indices. The off-diagonal components are often called the ``drag" Hall conductivity. It is also straightforward to derive the multi-layer generalization of the Streda formula:
\begin{equation}\label{bilaryrhoB} \delta\rho^{(a)}=\sum_b\frac{\sigma^{ab}_H}{2\pi}\delta B^{(b)}.
\end{equation}

In practically all FQH systems, both layers see the same magnetic field, and so we have
\begin{equation}
    \nu^{(a)}=\sum_b \sigma_H^{ab},
\end{equation}
where $\nu^{(a)}$ is the filling factor of the $a$th layer.

\subsection{Flux threading argument}\label{flux threading}

The mere observation of a quantized, fractional $\sigma_H$ leads to a profound conclusion: the FQH system must support excitations that are fractionally charged anyons. The corresponding ground states must then be topologically ordered. We review the arguments below; they are a summary of a line of thought initiated by Laughlin \cite{Laughlin1} and developed by many authors subsequently. 
That the fractional quantum Hall state supports anyon excitations was first discussed based on microscopic wavefunctions by Halperin \cite{HalperinHierarchy} and a bit later in  \cite{arovas1984fractional}. It is also implied by the Chern-Simons effective field theory description of these states \cite{zhang1989effective,read1989order,lopez1991fractional,sachdev2023quantum,wen2004quantum}. The explanation based on flux threading, reviewed below, trickled in over many years. For some representative papers, see \cite{tao1984gauge,kivelson1985consequences,einarsson1990fractional,wen1990ground,oshikawa_fractionalization_2006}.

Consider the system on an annulus $R_1\leq r\leq R_2,\ 0\leq \phi<2\pi$ where $r$ and $\phi$ are the polar coordinates. We assume that $R_2\gg R_1$. Imagine threading a magnetic flux\footnote{Later in the paper, we will refer to this as a holonomy and reserve the use of the term `magnetic flux' to situations where the flux penetrates the space occupied by the charged particles. The only exception is that we will use the term ``flux threading argument'' when we refer to this argument.  We also want to point out that this process of continuously changing the holonomy around a nontrivial cycle in space and following how the spectrum evolves is known as spectral flow and has been used extensively in various branches of mathematics and physics.} $\Phi(t)$ through the inner hole. This leads to the vector potential $\vec{A}(t)=\Phi(t)\uv{e}_\varphi/(2\pi r)$ on the annulus, where $\uv{e}_\varphi$ is the unit vector along the transverse direction.  We assume that the flux threading is slow and the time evolution of the bulk system is adiabatic. The flux threading induces an electric field $\vec{E}=-\uv{e}_\varphi(d\Phi/dt)/(2\pi r)$, driving a Hall current in the inward radial direction $I_r=\sigma_H (d\Phi/dt)/2\pi$. If the total flux threaded is $\Delta\Phi$, we find that an amount of charge $\sigma_H\Delta\Phi/2\pi$ is accumulated near the inner hole. This setup is depicted in Fig.~\ref{fig:current}.

\begin{figure*}
    \centering
    \subfigure[Flux-threading]{
        \includegraphics[width=0.45\textwidth]{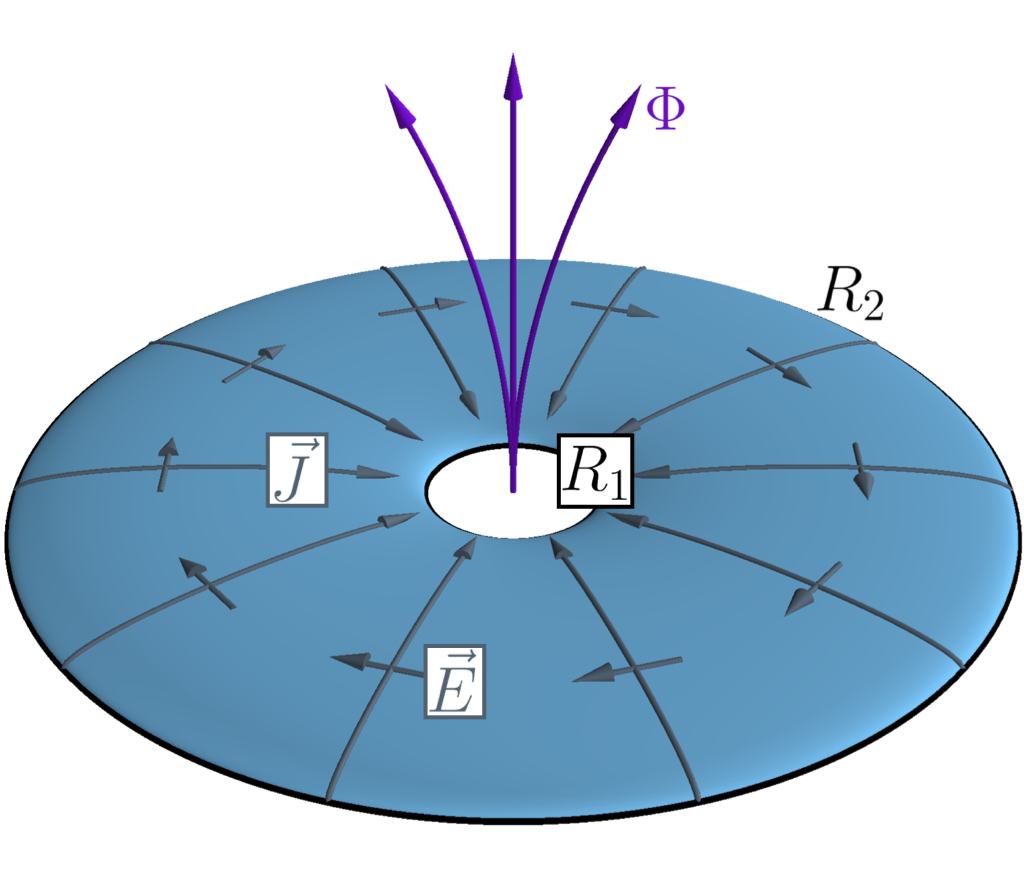}
        \label{fig:current}
            }
    ~ 
    \subfigure[Braiding with $v$]{
        \includegraphics[width=0.45\textwidth]{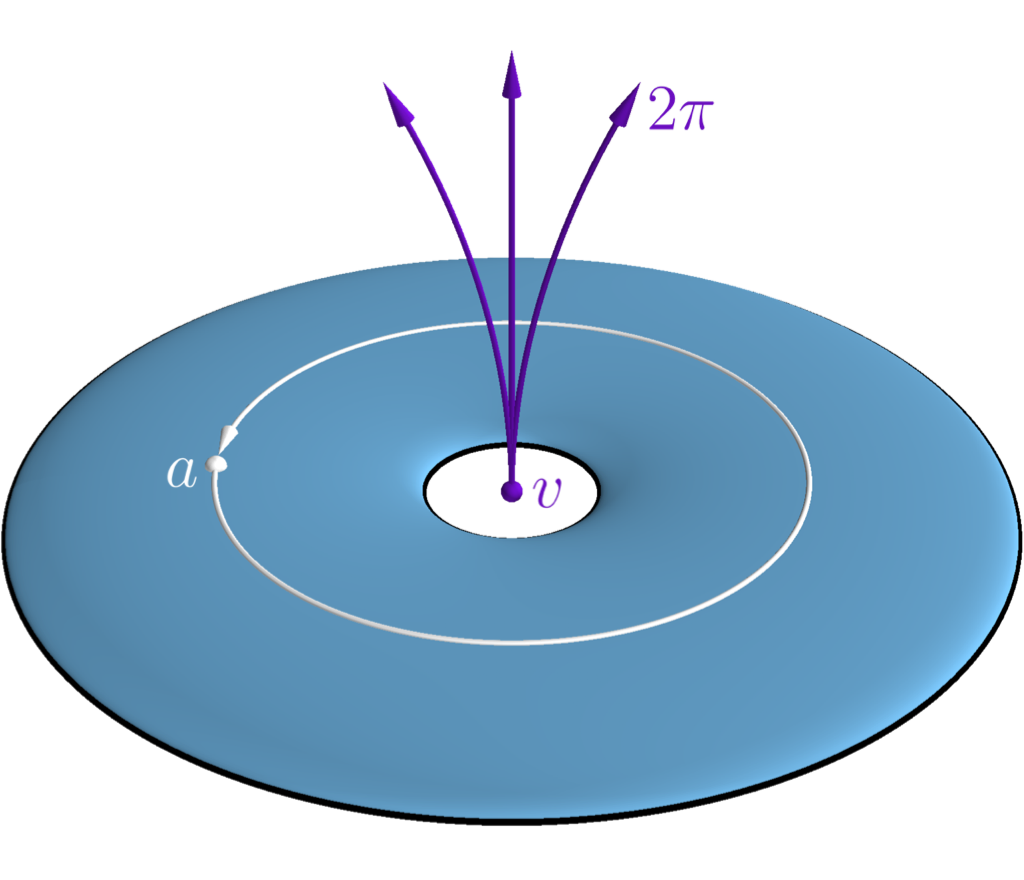}\
        \label{fig:braid}
            }
    \caption{Illustration of the flux-threading argument. (a) The setup of the argument; an annulus with inner radius $R_1$ and outer radius $R_2$. Flux $\Phi$ is threaded through the hole of the annulus, resulting in a radial electric field, $\vec{E}$. This electric field in turn induces a current flow from the outer edge to the inner edge, $\vec{J}$. When $2\pi$ flux is threaded through the hole, the system can be thought of as having returned to its ground state, albeit having accumulated a charge of $\sigma_H$ in the inner hole. (b) Display of the braiding of an anyon $a$ around this flux, which will detect its charge. The flux itself can be thought of as an anyon, the vison $v$, whose charge is $\sigma_H$ and whose topological spin is $\sigma_H/2$.}
	\label{fig:annulus}
\end{figure*}

If we set $\Delta\Phi=2\pi$ to be a single flux quantum, then after flux threading, the system returns to the initial configuration, possibly up to a large gauge transformation. Therefore, the final state of the adiabatic flux threading must be an eigenstate of the Hamiltonian. Because of adiabaticity, away from the two edges, the system must remain in the local ground state.  We can thus equivalently view the final state as nucleating a quasiparticle and its anti-quasiparticle and separating them to the two boundaries.  

The argument in the previous paragraph shows that this quasiparticle excitation carries electric charge $\sigma_H$.  We denote the anyon type of this quasiparticle by $v$ and will refer to it as the ``vison." It is displayed in Fig.~\ref{fig:braid}. Note that the flux threading argument determines the charge of $v$, not just its charge modulo one. This is because the Hall conductivity is some fixed number $\sigma_H$.

Coming back to the flux-threading argument, let us also consider adiabatically moving another quasiparticle of type $a$ counterclockwise along a nontrivial cycle of the annulus, displayed in Fig.~\ref{fig:braid}. This will result in a braiding phase $B(v,a)$. The same phase can be understood as the Aharonov-Bohm phase of a charge $Q(a)$  around the flux tube, so we find
\begin{equation}
\label{visonbraid}
    B(v,a)=Q(a) \mod 1\,.
\end{equation}
In other words, the vison determines the fractional part of the electric charge of all quasiparticles. In particular, setting $a=v$ we find $B(v,v)=\sigma_H \mod 1$. When $\sigma_H$ is fractional, we necessarily have a non-zero fractional part to $B(v,v)$. Thus $v$ has anyon statistics. This establishes that fractional $\sigma_H$ {\it requires} the existence of fractionally charged anyons. 

In addition, the topological spin $h(v)$ of the vison $v$ must satisfy $h(v)=B(v,v)/2 \mod {1/2}$. In fact, one can further show that  $h(v)=\sigma_H/2$ mod 1. Here we present a heuristic calculation assuming continuous rotation symmetry~\cite{Goldhaber:1988iw, Greiter:2022iph}. To determine the topological spin of the vison $h(v)$, we compute the change of the orbital angular momentum $L_z$ during the flux-threading process:
\begin{equation}
\begin{split}
    \frac{dL_z}{dt}&=\int_{R_1\leq r\leq R_2} d^2\vec{r}\, \rho(\vec{r}) \big[\vec{r}\times \vec{E}(\vec{r},t)]_z \\
    &= \frac{1}{2\pi}\frac{d\Phi}{dt}\int_{R_1\leq r\leq R_2} d^2\vec{r}\,\rho(\vec{r})\\
    &=\frac{1}{2\pi}\frac{d\Phi}{dt}\cdot \frac{\sigma_H}{2\pi}\Phi = \frac{\sigma_H}{8\pi^2}\frac{d}{dt}(\Phi^2).
\end{split}
\end{equation}
Integrating from $\Phi=0$ to $\Phi=2\pi$, we find $L_z=\sigma_H/2$, which we interpret as the topological spin of the vison being $h(v)=\sigma_H/2$. A derivation of this relationship can be readily given for systems whose IR physics is described by Abelian TQFTs using the $K$-matrix Chern-Simons theory reviewed in  Sec.~\ref{sec:tqft-basics}. This applies irrespective of whether the microscopic UV system has rotational invariance or not. A rigorous derivation for lattice systems can be found in \cite{Kapustin:2020bwt}.

We pause here to note that we could always bind local charged particles to $v$, e.g., the electron, and this new particle would have the same braiding with all other anyons. Strictly speaking, though, the new particle is not the one that is nucleated by adiabatic flux insertion. This is because its charge differs from $\sigma_H$. However, we will often focus on the anyon content and braiding of the theory, which only depends on the fractional part of the Hall conductivity. In other words, we will consider a variety of states, all with the same $\sigma_H \mod 1$. In that case, we can say that the flux threading argument determines the charge and topological spin of the vison up to the identifications discussed in  Sec.~\ref{sec:tqft-basics}. We will further elaborate on this comment below, e.g., after \eqref{ruspq}.
 
Let us briefly restate what we learned above about the structure of the TQFT that describes fractional quantum Hall states.  The flux-threading argument identifies a quasiparticle whose
corresponding anyon type, the vison $v$, has
\begin{equation}
\label{vison-h-Q}
{\big(h(v),Q(v)\big)\sim\left( \frac12\sigma_H, \sigma_H\right)\,,}
\end{equation}
modulo the identifications discussed in  Sec.~\ref{sec:tqft-basics}. The vison further has Abelian braiding with all the other anyons given by \eqref{visonbraid}. In turn, this implies \cite{barkeshli_symmetry_2019} that $v$ itself is an Abelian anyon. Thus, we reach the important conclusion that any TQFT that describes a fractional quantum Hall state must have a non-empty set of Abelian anyons $\cal A$ that includes the vison $v$ (and its powers obtained by fusion). 

Using the modern language of one-form symmetries, reviewed in Sec.~\ref{sec:tqft-basics}, we can restate what we learned about the TQFT of fractional quantum Hall states. A fractional $\sigma_H$ implies that the 
TQFT has a one-form symmetry with a particular anomaly structure and a particular coupling to the background electromagnetic gauge field. In Sec.~\ref{from response to}, we provide more detail when we discuss the response theory. 

Let us briefly comment on the FQHE for bosonic systems, where in the microscopic system the $\U$ charge is carried by bosons.  Though not yet observed in experiment, it is interesting to consider such bosonic systems. The flux-threading arguments above lead to the same conclusion as before. Any bosonic system with a fractional $\sigma_H$ will necessarily have a non-empty set of fractionally charged Abelian anyons. Note that for bosonic systems, `integer' QHE states with no anyon excitations must have even integer $\sigma_H$ \cite{lu_theory_2012,senthil2013integer}. This is because with odd $\sigma_H$, the vison excitation nucleated by flux threading has topological spin $ h(v) =\sigma_H/2 = 1/2 \pmod{1}$ and hence is a fermion, which cannot be created locally in a bosonic system. Thus, an odd integer $\sigma_H$ also implies a TQFT with non-trivial one-form symmetry in the bosonic QHE, and so we will consider states with the same $\sigma_H$ mod 2 for bosonic theories. 

The flux-threading argument in this section can be easily generalized to the multi-layer case. The $\U$ symmetry of the $a$th layer leads to a vison quasiparticle denoted as $v_a$. Its charge under the $\U$ symmetry of the $b$th layer is given by
\begin{equation}
    Q_b(v_a)=\sigma^{ab}_H.
\end{equation}
This implies that the braiding phase between $v_a$ and $v_b$ is given by
\begin{equation}
    B(v_a,v_b)=\sigma^{ab}_H\mod 1.
\end{equation}
We can also find the spin of $v_a$,
\begin{equation}
    h(v_a)=\frac12\sigma^{aa}_H.
\end{equation}
However, the visons in each layer are not necessarily distinct from each other. For example, in the $\nu=2/3$ bilayer, the Halperin (112) state has $v_1=v_2$.

\section{TQFT Preliminaries}
\label{sec:tqft-basics}

As discussed in the prior section, the quantization of the Hall conductivity and the concomitant vanishing longitudinal conductivity imply that the current-carrying excitations are gapped. We will further assume that there is a gap to all excitations, including electrically neutral ones. Then the low-energy physics can be described in terms of a  continuum topological quantum field theory (TQFT). The TQFT is further enriched by a global $\U$ symmetry corresponding to the conservation of electric charge~\cite{barkeshli_symmetry_2019}.    

\subsection{Introduction to TQFTs}

The TQFT supports a number of distinct types of gapped anyon excitations (defects).
Let us quickly review the universal data of anyons that will be frequently used in our analysis in this work. For a more in-depth review, see e.g., \cite{kitaev_anyons_2006}. An anyon type $a$ represents the equivalence class of quasiparticle excitations up to attachment of local excitations.  A special type, denoted by the identity $1$, represents local bosonic excitations. For fermionic systems, we will use $c$ to denote local fermions. In many cases, we will consider $c$ to be a nontrivial anyon, which is often useful for bookkeeping purposes.  Occasionally, it will prove useful to identify $c$ with the vacuum, though we will always be careful to state when we do this.
 
A basic physical operation is the fusion of two anyons. We denote the fusion of anyons $a$ and $b$ as $a\times b$.  In general, $a\times b$ is a linear combination of anyons.  If $a\times b$ is a single anyon, for all $b$, we refer to $a$ as an Abelian anyon.   It is clear that the set of all Abelian anyons $\cal A$ forms an Abelian group whose multiplication is given by fusion.  The $n$th power of an Abelian anyon $a^n$ is then defined in the obvious way.  

Each anyon type $a$ is associated with a topological spin $h(a)$. The exchange statistics of $a$ is given by $e^{2\pi i h(a)}$, so $h(a)$ is defined modulo one. In addition, there is also a phase produced by a full braid between $a$ and $b$, i.e., moving $a$ counterclockwise around $b$. In all cases of interest to us, either $a$ or $b$ is Abelian. Then the amplitude is a phase factor denoted by $e^{2\pi iB(a,b)}$.\footnote{The terminology in the original discussions of this topic, in the context of RCFT \cite{Moore:1988qv,Moore:1989vd}, followed the braid group terminology where the term braiding refers to half a twist.  This operation is commonly referred to as $R$ in the more modern TQFT discussions.} Further,
\begin{equation} \label{eqn:braid_defn}
     B(a,b) = h(a\times b) - h(a) - h(b),
\end{equation}
so the topological spin of all anyons determines the braiding with Abelian anyons. We require that $h$ is such that $B$ is linear when one of its entries is restricted to Abelian anyons, i.e., that $B(a\times d,b) = B(a,b)+B(d,b)$ for $a,d\in {\cal A}$. Additionally we require that $h$ is homogeneous \cite{Belov:2005ze, Stirling:2008bq},   i.e., if $a$ is an Abelian anyon then 
\begin{equation}
     h(a^n) = n^2h(a).
\end{equation} 
Note that these two relations mean $B(a,a)=2h(a)$, as it should. Now $h$ defines a quadratic form over $\cA$, with $B$ being the associated symmetric bilinear form. 

Let us return to the full set of anyons, not necessarily in $\cal A$.  In a bosonic system, we require that $B$ be non-degenerate, i.e., the only anyon that braids trivially with all others is the identity $1$. (This requirement leads to the fact that the theory is modular; hence, the M in MTC.) A fermionic theory is not modular.  It includes a unique anyon $c$, besides $1$, that braids trivially with all the other anyons.  It is the physical fermion, which satisfies $c^2=1$ and $h(c)=\frac12$. Often we identify $c$ with $1$, since both are local particles.

Let us demonstrate the general discussion with the special case of Abelian theories.  They are 
described by a $K$-matrix Chern-Simons Lagrangian\footnote{There is a procedure for mapping the algebraic data of the TQFT to a $K$ matrix, which is given in Appendix~C of \cite{ma_fractonic_2022}. We will not review it here, except to mention that it is possible to go back and forth.} 
\begin{equation} \label{eqn:K_mat}
    \begin{split}
    &\mathcal{L} = {1\over 4\pi}\sum_{ij}K^{ij} b_i\wedge db_j\\
    &K^{ij}\in \bb{Z}\,,
    \end{split}
\end{equation}
where the matrix $K$ is symmetric and invertible.\footnote{We hope that we do not confuse the reader by denoting some anyons by $b$ and also using $b$ to denote some dynamical $\U$ gauge fields.} The restriction that $K^{ij}$ are integers guarantees that the theory is consistent as a spin theory.  For a bosonic theory, we also need the restriction
\begin{equation}
    K^{ii}\in 2\bb{Z} \qquad,\qquad {\rm bosonic\ theories}\,.
\end{equation}

The anyons in these theories are $a_\xi = e^{i \sum_i \xi^i\int b_i }$, with $\xi \in \mathbb{Z}^n$ a vector of integers. Fusion of anyons is given by $a_\xi \times a_\eta = a_{\xi + \eta}$. The anyons have topological spin given by
\begin{equation} \label{eqn:spin_Abelian}
    h(a_\xi) =\left({1\over 2}\sum_{ij} \xi^i (K^{-1})_{ij} \xi^j\right) \mod 1\,.
\end{equation}
From  \eqref{eqn:braid_defn} we see that the braiding is given by
\begin{equation} \label{eqn:braid_Abelian}
B(a_\xi, a_\eta) = h(a_{\xi+\eta}) -h(a_\xi)-h(a_\eta)= \sum_{ij}\xi^i (K^{-1})_{ij} \eta^j \mod{1}\,.
\end{equation}

Before discussing the coupling of a general TQFT to a global $\U$ symmetry, let us demonstrate it in this simple Abelian theory.  The $\U$ symmetry is coupled to a background gauge field $A$ through the Chern-Simons Lagrangian (see, e.g., \cite{wen1992classification} and \cite{Frohlich:1995qm}) 
\begin{equation}\label{Abelian CS T}\begin{split}
    &{1\over 4\pi}\sum_{ij} K^{ij} b_i \wedge db_j -{1\over 2\pi}\sum_i t^i A \wedge db_i\\
    &t^i\in \bb{Z}\,.
    \end{split}
\end{equation}
In this case each anyon type $a_\xi$ carries electric charge $Q(a_\xi) = B(a_\xi, a_t)$.

Let us turn to the more general case.  The topological spin $h(a)$ and the $\U$ charge $Q(a)$ are subject to identifications because of attachment of local excitations. We consider three cases.
\begin{enumerate}
\item  If the microscopic theory is made out of bosons with integer $\U$ charges, the TQFT is bosonic and 
\begin{equation}
(h(a),Q(a))\sim (h(a)+1, Q(a))\sim (h(a),Q(a)+1)\,.
\end{equation}
Physically, the ambiguity in the charge arises from attaching a transparent boson with integer charge. 

The different ways of assigning $\U$ charges in a bosonic TQFT were classified in  \cite{barkeshli_symmetry_2019}. There, it was found that charges modulo one are equivalent to braiding with a special Abelian anyon, and thus the different charge assignments correspond to choices of this Abelian anyon. This anyon is precisely the vison we saw in the previous section.

In the simple case of an Abelian theory \eqref{Abelian CS T}, the fusion of Abelian anyons forms an Abelian group  ${\cal A} = \bb{Z}^n/K\bb{Z}^n$.

\item In a spin theory where $c$ is identified with the vacuum,\footnote{A spin TQFT is a TQFT that is well defined only on spin manifolds. And then, it depends on the choice of spin structure on that manifold \cite{Dijkgraaf:1989pz,Gaiotto:2015zta,Bhardwaj:2016clt}. Such a system does not have to be enriched by any global symmetry.}
\begin{equation}
\left(h(a),Q(a)\right)\sim \left(h(a)+\frac{1}{2}, Q(a)\right)\sim \left(h(a),Q(a)+1\right)\,.
\end{equation}
We prove in Appendix~\ref{app:spinc-classification} that the different ways of assigning $\U$ charge in this theory are essentially identical to the bosonic theory. They also correspond to different choices of Abelian vison, modulo fusion with $c$.

Going to the Abelian theory \eqref{Abelian CS T}, 
the group of Abelian anyons is ${\cal A} = \bb{Z}^n/K\bb{Z}^n$.

\item If the system is fermionic and all the local fermions/bosons carry odd/even charges, then the theory is referred to as a spin$^{c}$ TQFT, and this condition is often referred to as the spin/charge  relation. Such a theory describes the electronic systems that arise naturally in condensed matter physics. 
This is the setting for all known observations of the FQHE, so it will be the bulk of our focus in this paper.

If $c$ is not identified with the vacuum, this means that $Q(a)$ is now defined modulo two since odd-charged excitations cannot be identified with the vacuum. Thus
\begin{equation}
\left(h(a), Q(a)\right) \sim \left(h(a)+1, Q(a)\right)\sim\left(h(a), Q(a)+2\right)\,.
\end{equation}
If $c$ is identified with the vacuum, then there is an additional equivalence relation which relates the tuple $(h(a), Q(a))$ to $(h(a\times c), Q(a\times c))$,
\begin{equation}\label{eq:fermi_id}
\left(h(a), Q(a)\right) \sim \left(h(a) + \frac{1}{2}, Q(a)+1\right).
\end{equation}

This identification is often complicated to work with because it means that we can no longer think of $h$ and $Q$ as independent. Instead it is helpful to define another function
\begin{equation} \label{eqn:qdef}
\hat{h}(a) = h(a) + \frac{1}{2}Q(a)\sim \hat{h}(a)+1\,.
\end{equation}
This identification is valid even when $c$ is not identified with the vacuum. We can choose to work instead with the tuple $(\hat{h},Q)$.  Now, if we identify $c$ with the vacuum,  $(\hat{h},Q) \sim (\hat{h}+1,Q)\sim (\hat{h},Q+1)$.

This $\hat{h}$ is not just a convenient tool, but plays an important role in the physics of spin$^c$ theories.  There are several ways to think about  $\hat h(a)$. Geometrically, we have twisted the theory by relating the spin connection on space to the $\U$ gauge field $A$ such that, after this twist, the theory is effectively bosonic. The function $\hat h(a)$ of \eqref{eqn:qdef} is then the topological spin after the twist.  Alternatively, in terms of the corresponding RCFT, the twisting by $A$ corresponds to a spectral flow transformation. Finally, in the 2+1d TQFT we can introduce a $\pi$ flux of $\U$, $\bm{\gamma}$. Note that this is not an anyon in the TQFT. Then $\hat h(a)$ can be written as $h(\bm{\gamma}\times a)-h(\bm{\gamma})$.

Just as $\hat{h}$, the $\pi$-flux defect is not just a convenient tool, but plays an important role in the physics. Classifying the different ways of assigning $\U$ charges in a spin$^c$ TQFT can be done by extending the theory to a bosonic TQFT. Here, the $\pi$-flux line appears not as a defect line, but as an actual anyon in the theory. Once the theory has been extended to a bosonic TQFT, the techniques of  \cite{barkeshli_symmetry_2019} can be employed to classify charge assignments. We give much more detail in Appendix~\ref{app:spinc-classification}.

Now consider the Abelian theory \eqref{Abelian CS T}. Here a spin$^c$ theory is a fermionic theory where the requirement $K^{ii}\in \bb{Z}$ is strengthened to
\begin{equation}\label{tplusK}
 K^{ii}+ t^i \in 2\bb{Z}\,.
\end{equation}
With these restrictions, along with the earlier restriction that $K^{ij},t^i\in \bb{Z}$, the group of Abelian anyons is ${\cal A} = \bb{Z}^n/K\bb{Z}^n$.
\end{enumerate}

Note that $B(a,b)$, which is defined modulo 1, is not affected by these identifications.

\subsection{Global symmetries, anomalies, and gauging}
\label{sec:symm_anom_gauge}

We now rephrase the discussion using the language of higher-form symmetries \cite{Gaiotto:2014kfa}, within the framework of TQFT.  The advantage of doing so is that the approach streamlines the matching of symmetries and anomalies between UV and IR theories, and the reasoning can be easily extended to nontopological theories with such symmetries.

We again take the IR limit to be described by a TQFT. A 2+1d TQFT has no local, point-like observables.\footnote{Here, we follow the condensed matter definition of a TQFT.  Field theorists and mathematicians allow TQFTs to have point-like observables.} Instead, the only observables are line operators/defects.  Following the standard terminology, a line stretched in time is referred to as a defect in the TQFT.  It represents the worldline of an anyon. A line stretched in space at a fixed time is an operator in the TQFT.  We can think of it as creating an anyon-anti-anyon pair, dragging one of them around a closed curve, and then annihilating them.

Among these lines, those that correspond to Abelian anyons generate the one-form symmetry group of the TQFT \cite{Gaiotto:2014kfa}.\footnote{A TQFT may also have nontrivial zero-form symmetries that permute the lines \cite{barkeshli_symmetry_2019}, e.g., charge conjugation.}  As before, we denote this set of Abelian lines as $\cal A$.  Further, if the anyons have non-trivial self-statistics, this one-form symmetry is itself anomalous \cite{Gaiotto:2014kfa,Hsin:2018vcg}. 

Let us demonstrate this in the simple case of $\U_s$.  This is a special case of \eqref{eqn:K_mat} with a single $\U$ gauge field $b$ with the action
\begin{equation}\label{Uonelss}
{s\over 4\pi} \int_\tmfd b \wedge db\,.
\end{equation}
For simplicity, we focus on even $s$, in which case the theory is bosonic. We can also couple it to a background $\U$ gauge field $A$ and write it as
\begin{equation}\label{bAcoupl}
{s\over 4\pi} \int_\tmfd b \wedge db-{1\over 2\pi}\int_\tmfd A\wedge db\,.
\end{equation}
(For simplicity, we set $t=1$.) 

The action \eqref{Uonelss} is invariant under the symmetry transformation
\begin{equation}\label{onefomrac}
b\to b+ \lambda\,,
\end{equation}
with $\lambda$ a flat $\U$ gauge field, i.e., $d\lambda=0$, with $e^{is\oint \lambda}=1$ around every one-cycle.  Such $\lambda$ is a $\mathbb{Z}_s$ gauge field and the global symmetry transformation \eqref{onefomrac} is a $\mathbb{Z}_s^{(1)}$ one-form global symmetry.\footnote{Throughout this paper, we use the convention that $\mathbb{Z}_s^{(p)}$ denotes a $\mathbb{Z}_s$ $p$-form symmetry.}

Given a global symmetry, it is common to couple it to a background gauge field.  In this case, the background gauge field $\cal B$ is a $\mathbb{Z}_s$ two-form field.  This means that for every two-cycle $\cal C$ in our spacetime,  $\int_{\cal C} {\cal B}$ is an integer modulo $s$.  The effect of nonzero $\cal B$ is to twist the dynamical gauge field $b$.  Instead of having $\int_{\cal C} {db\over 2\pi} \in \mathbb{Z}$, we have
\begin{equation}\label{bBqua}
\int_{\cal C} {db\over 2\pi} \mod 1 = {1\over s} \int_{\cal C} {\cal B}\mod 1\,.
\end{equation}
In other words, for a nonzero background field $\cal B$, the flux of $b$ has a fractional value.

A standard way to define the action based on \eqref{Uonelss} is to view our spacetime as the boundary of an auxiliary four manifold $\fmfd$ and to write \eqref{Uonelss} as ${s\over 4\pi} \int_\fmfd db\wedge db$ \cite{Witten:1988hf}. Indeed, we will frequently want to take advantage of this useful tool of viewing our 2+1d spacetime as the boundary, $\tmfd=\partial \fmfd$, of a 3+1d bulk, $\fmfd$. We emphasize that the 3+1d manifold $\fmfd$ is not physical. It is merely a mathematical device.
 
Regardless of our choice of $\fmfd$, with nonzero $\cal B$ the definition of the action depends on $b$ in the bulk $\fmfd$. Importantly, this is not as bad as it could have been because the dependence on $b$ in the bulk depends only on $\cal B$.  This means that we can extend $b$ to the bulk in an arbitrary way, as long as we keep \eqref{bBqua}.  Then, the dependence on the bulk is only through $\cal B$ and is given by a quadratic expression in $\cal B$, which we will discuss below.  The significance of this statement is that even for nonzero $\cal B$, we can still view $b$ as a 2+1d field and the theory remains 2+1-dimensional.

What we have seen here is that the theory has a global symmetry acting as \eqref{onefomrac} and when we couple it to background gauge field $\cal B$, we need to specify additional information, which is captured by an extension of $\cal B$ to one dimension higher.  This is the standard setting of an 't Hooft anomaly.  We conclude that the one-form $\mathbb{Z}_s^{(1)}$ global symmetry has an anomaly, which is described by \cite{Kapustin:2014gua,Gaiotto:2014kfa,Benini:2018reh,Hsin:2018vcg}
\begin{equation}
    \frac{2\pi }{2s} \int_{\fmfd} {\cal P}({\cal B})\,.
    \label{PB}
\end{equation}
Here, ${\cal P}({\cal B})$ is the Pontryagin square of $\cal B$.\footnote{The Pontryagin square ${\cal P}({\cal B})$ is a quadratic operation on cohomology \cite{whitehead1949on}.  (For reviews for physicists, see e.g., Appendix~A in \cite{Kapustin:2013qsa} and Appendix~C in \cite{Benini:2018reh}.)  In our case it is a map $H^2(M,\bZ_s)\to H^4(M,\bZ_{2s})$ for even $s$.  As such, for ${\cal B}_{1,2} \in H^2(M, \bZ_s)$
\begin{equation}
    {\cal P}({\cal B}_1 + {\cal B}_2) = {\cal P}({\cal B}_1) + {\cal P}({\cal B}_2) + 2 {\cal B}_1 \cup {\cal B}_2
\end{equation}
and therefore 
\begin{equation}
    {\cal P}(n {\cal B}) = n^2 {\cal P}({\cal B})\,.
\end{equation}
For odd $s$, we can take it to be simply the cup product $ {\cal P}({\cal B})={\cal B}\cup {\cal B}\mod s$. The novelty in the Pontryagin square is that for even $s$, it is well defined modulo $2s$.  This generalizes the obvious fact that for ordinary numbers, if $x$ is defined modulo $s$, then $x^2$ is defined modulo $s$ for odd $s$ and modulo $2s$ for even $s$.}

This discussion of $\cal B$ is related to the coupling to a background $\U$ gauge field $A$ as in \eqref{bAcoupl}.  To see that, note that in the absence of sources, \eqref{bAcoupl} leads to the equation of motion of $b$ 
\begin{equation}\label{sdbdA}
sdb=dA\,
\end{equation}
and hence, for every closed two-cycle $\cal C$
\begin{equation}
\int_{\cal C} {db\over 2\pi} = {1\over s}\int_{\cal C} {dA\over 2\pi} \,.
\end{equation}
Hence, we identify $\int {dA\over 2\pi}=\int {\cal B}$, where $\cal B$ is the background field of the $\mathbb{Z}_s^{(1)}$ one-form global symmetry.  For values of $\cal B$ satisfying this relation, the anomaly \eqref{PB} is
\begin{equation}\label{dAdAt}
{1\over 4\pi s}\int_\fmfd dA\wedge dA\,.
\end{equation}
We should emphasize, though, that while \eqref{PB} is a nontrivial anomaly theory, this is not the case with \eqref{dAdAt}.  The latter is a well-defined expression in a manifold with a boundary.

Let us try to gauge the one-form symmetry.  The anomaly prevents us from doing it.  However, if there are $m\in 2\bZ$ and $n\in \bZ$ such that $s=mn^2$, then, we can gauge $\mathbb{Z}_n^{(1)}\subset \mathbb{Z}_s^{(1)}$, because the anomaly for this subgroup vanishes.  Denote the gauge field for this subgroup as $\frak b$, in terms of which ${\cal B}={s\over n}\frak b$.\footnote{We use uppercase characters for classical background fields and lowercase characters for dynamical fields.}  Clearly, $\frak b$ is a $\mathbb{Z}_n$ gauge field.  Then, for fixed $\frak b$, \eqref{bBqua} becomes 
\begin{equation}\label{bBquat}
\int_{\cal C} {db\over 2\pi} \mod 1 = {1\over n} \int_{\cal C} \frak b\mod 1\,.
\end{equation}
This means that the fractional part of the flux of $b$ is a multiple of $1\over n$.  Then, when $\frak b$ is a dynamical field and we sum over its values, it means that the dynamical field $b$ can be replaced with a dynamical $\U$ gauge field $\hat{b}=nb$ with standard fluxes.  This turns the Lagrangian \eqref{Uonelss} into
\begin{equation}\label{Uonels}
{m\over 4\pi}\int_{\tmfd} \hat{b} \wedge d\hat{b}\,.
\end{equation} 
It is a bosonic $\U_m$ Chern-Simons theory, and it has a $\bZ_m^{(1)}$ one-form symmetry.

For example, the $\U_8$ TQFT has a $\mathbb{Z}_8^{(1)}$ one-form global symmetry.  A $\mathbb{Z}_2^{(1)}$ subgroup of it is anomaly free and can be gauged.  In the notation above, this corresponds to $s=8$ with $m=n=2$.  This gauging leads to the $\U_2$ TQFT.

We now generalize the gauging procedure and continue to focus on bosonic theories. A subset of the anyons in the TQFT $\cal T$ is the Abelian anyons ${\cal A}\subset {\cal T}$. They are the symmetry lines of the one-form global symmetry of $\cal T$.  We consider an anomaly free subset of them ${\cal A}_0\subset {\cal A}$. This means that for all $x,y \in {\cal A}_0$,
\begin{equation}
    h(x)=0\mod 1\qquad,\qquad B(x,y)=0\mod 1\,.
\end{equation}

The anyons in ${\cal A}_0$ generate the one-form symmetry group $\otimes_I \mathbb{Z}_{n_I}^{(1)}$, which is a subgroup of the total one-form symmetry generated by $\cal A$.

Gauging ${\cal A}_0$ leads to another TQFT denoted by ${\cal T}/{\cal A}_0$. To find it, we couple $\cal T$ to background gauge fields for ${\cal A}_0$.  If the symmetry associated with ${\cal A}_0$ is $\otimes_I \mathbb{Z}_{n_I}^{(1)}$, we have background fields ${\frak b}_I$.  Denote the partition function on 2+1d spacetime manifold $\tmfd$ with background field ${\frak b}_I$ as ${\cal Z}_{\cal T}[\tmfd, {\frak b}_I]$. Then, the gauging corresponds to  making ${\frak b}_I$ dynamical by summing over all possible configurations of ${\frak b}_I$,
 \begin{equation}\label{ZsumB}
 \begin{aligned}
     &{\cal Z}_{{\cal T}/{\cal A}_0}[\tmfd]={1\over {\cal G}}\sum_{{\frak b}_I}{\cal Z}_{\cal T}[\tmfd, {\frak b}_I]\\
     &{\cal G} =\prod_I{|H^1(\tmfd,\bZ_{n_I})|\over |H^0(\tmfd,\bZ_{n_I})|}\,,
    \end{aligned}
 \end{equation}
where $\cal G$ is the volume of the gauge group (see, e.g., Appendix~B in \cite{Gaiotto:2014kfa}).
Since ${\cal A}_0$ is non-anomalous (meaning that the theory coupled to ${\frak b}_I$ is gauge invariant), the procedure is well defined and leads to a new TQFT ${\cal T}/{\cal A}_0$, as illustrated above by the example of the $\U_{mn^2}$ theory. 

This gauging procedure can also be understood using the ``Poincar\'e dual" perspective. Instead of thinking about gauge fields, one can equivalently think in terms of the topological defects.  In this case, the defects in spacetime are nothing but the line operators of ${\cal A}_0$.  Coupling to the background field means inserting defect lines in the partition function. For example, in \eqref{bBqua} if $\int_{\cal C}{\cal B}=k$ is non-zero mod $s$, then in the Poincar\'e dual picture a Wilson line of charge $k$ of the Chern-Simons theory intersects the surface ${\cal C}$.  The gauged theory is then obtained by summing over all possible insertions of ${\cal A}_0$ lines in the theory. 

Now, suppose $\tmfd$ takes the form of $\Sigma\times S^1$ where $\Sigma$ is a closed surface, and $S^1$ is the Euclidean time circle. For simplicity, let us take $\Sigma=S^2$ so $\cal T$ has a unique state. Quantizing the theory, one obtains the state of ${\cal T}/{\cal A}_0$ as a coherent superposition of all possible insertion of (point-like) defects in ${\cal A}_0$ on the state of ${\cal T}$,\footnote{If the TQFT is embedded in a nontopological UV  theory, these defects can be the IR avatars of massive particles, or anyons.} i.e., a proliferation of ${\cal A}_0$ defects.

This picture is also convenient for understanding the effect of gauging algebraically. The spectrum of lines in $\cT/\cA_0$ can be obtained through the following three-step procedure:\footnote{\label{gaugingvscon}This procedure was first discussed in the context of RCFT as an extension of the chiral algebra \cite{Moore:1988ss}.  It was later related to a quotient of the gauge group of the corresponding Chern-Simons gauge theory \cite{Moore:1989yh}.  [In the example leading to \eqref{Uonels}, we used a $\mathbb{Z}_n$ quotient of $\U$.]   From a more modern perspective, this is gauging a one-form symmetry \cite{Gaiotto:2014kfa}. The extension of the chiral algebra also leads to examples [e.g., ${\rm SU}(2)_{10}\to {\rm Spin}(5)_1$] that can be thought of as gauging a non-invertible symmetry.  The term condensation for this procedure was introduced in \cite{Bais:2008ni} and is used often.  However, it has led to some confusion. This issue will be summarized in \cite{gaugingcondensation}, where the similarities and the distinctions between gauging and condensation will be clarified.}
\begin{enumerate}
\item First, all the lines charged under ${\cal A}_0$ are projected out. Only the lines with trivial braiding with ${\cal A}_0$ survive. This is expected from the general definition of gauging, which projects on gauge-invariant operators. We can also see it from the defect line proliferation picture.  Consider inserting a line $x\in {\cal T}$ with nontrivial braiding with ${\cal A}_0$.  Then, by summing over all possible ${\cal A}_0$ loops linked with the $x$ line, the partition function vanishes.
\item Next, we split the lines braiding trivially with ${\cal A}_0$ into orbits under fusing with $\cA$. Since all possible insertions of ${\cal A}_0$ are summed over, lines in the same orbit must be identified. 
\item There is an additional subtlety when a line $x$ is invariant when fusing with nontrivial elements of ${\cal A}_0$. Denote the stabilizer subgroup $\cA_x$ as $\{a\in {\cal A}_0\,|\, a\times x=x\}$. Then in ${\cal T}/{\cal A}_0$, there are multiple descendants of the line $x$, which are associated with lines $a\in\cA_x$ ending on the lines of the anyon $x$.  All of these lines have the same topological spin as $x$. The computation of other observables, such as the S matrix, is more involved. More details can be found in, e.g., \cite{Eliens:2013epa}.
\end{enumerate}

The theory ${\cal T}/{\cal A}_0$ has a dual $\otimes_I \mathbb{Z}_{n_I}^{(0)}$ zero-form symmetry whose ``charges'' are given by integrals over two-cycles $\cal C$, $\int_{\cal C}{\frak b}_I \mod n_I$. In some contexts (e.g., in the study of orbifolds), this dual symmetry is referred to as a ``quantum symmetry.''  These charges can be coupled to background one-form gauge fields $\hat A_I$ by generalizing  \eqref{ZsumB} to   \begin{equation}
    {\cal Z}_{{\cal T}/\cA_0}[\tmfd,\hat{A}_I]={1\over {\cal G}}\sum_{{\frak b}_I} {\cal Z}_{\cal T}[\tmfd,{\frak b}_I]\exp \left(2\pi i\sum_I {1\over n_I}\int_{\tmfd} \hat{A}_I\cup {\frak b}_I\right).
    \label{gauged-theory-with-dualA}
\end{equation}

It is straightforward to check, using   \eqref{gauged-theory-with-dualA}, that the procedure can be inverted: starting with ${\cal T}/\cA_0$, and then gauging the dual zero-form symmetry, i.e., summing over $\hat{A}$, leads to the original theory. We will use this procedure in  Sec.~\ref{sec:IR}.\footnote{More precisely, there is an ambiguity in the gauging of the dual zero-form symmetry, associated with adding a counterterm that depends on $\hat A_I$ before the gauging.  In this case, this counterterm is a Dijkgraaf-Witten term.} 
    
It is worth emphasizing that the prescription does not rely on the theory $\cal T$ being topological. The same procedure can be applied to a non-topological theory with non-anomalous one-form symmetry.

Since gauging corresponds to a quotient, we will adopt the notation $\cT/\cA_0$ throughout the paper.  Most commonly, we will denote $\cA_0$ by its generating anyons. In some cases, $\cT$ has a Chern-Simons gauge theory description, in which case we will write in the denominator the corresponding subgroup (in the center of the gauge group). Whenever possible, we will give both expressions.

Let us now demonstrate the three-step procedure through a couple of examples.

First, we revisit the example of gauging a $\bZ_2^{(1)}$ one-form symmetry in $\U_8$, analyzed above. Denote the Wilson lines in $\U_8$ by $a_j$, where $j$ is a mod 8 integer, and $a_0$ is the identity $1$. In this case ${\cal A}_0=\{1,a_4\}$. Algebraically, in step 1, the gauging removes all $a_j$ with odd $j$. In the second step, we identify $a_2\sim a_6, a_0\sim a_4$, after which there are only two anyon types $1$ and $a_2$. Since the theory is Abelian, there is no need for the last step. Using the notation introduced above, this example can be summarized as
\begin{equation}
    \frac{\U_8}{(a_4)}=\frac{\U_8}{\bZ_2}=\U_2\,.
\end{equation}
Next, we consider a non-Abelian example: gauging the one-form symmetry generated by the isospin-2 line in ${\rm SU}(2)_4$. It leads to ${{\rm SU}(2)_4\over \mathbb{Z}_2}={\rm SO}(3)_4$.   Denote the isospin-$j$ anyon by $\Phi_j$ ($\Phi_0$ is the identity). The gauging removes all half-integer isospin anyons, and we are left with $\Phi_0$ (identified with $\Phi_2$) and $\Phi_1$. Since $\Phi_2\times \Phi_1=\Phi_1$, $\Phi_1$ splits into two Abelian anyons, both with $h=1/3$. This case is atypical because the resulting theory SO(3)$_4$ has another simple Chern-Simons presentation as ${\rm SU}(3)_1$ \cite{Moore:1988ss,Moore:1989yh}. The inverse procedure is to gauge the $\bZ_2^{(0)}$ zero-form charge-conjugation symmetry in ${\rm SU}(3)_1$ to recover SU(2)$_4$.

\subsection{TQFTs vs gapped phases}\label{sec:tqft-gappedphase}

In this work, the main question we set out to address is how the Hall conductivity $\sigma_H$ constrains the anyon content of the topological phase. In this case, the answer depends only on $\sigma_H$ modulo the values of invertible theories, as adding invertible theories does not change the anyons, even though it changes the phase. For example, if the system is spin{}$^c$, the answer should depend only on $\sigma_H$ mod 1, since shifting $\sigma_H$ by an integer corresponds to stacking integer quantum Hall states.  

Let us elaborate on this point in the field theory context.
In general, a quantum field theory depends on various  parameters:
\begin{itemize}
    \item The traditional parameters affect correlations of local point operators at separated points. Examples include fermion masses, $\theta$ parameters, Yukawa coupling constants, etc.
    \item More subtle parameters do not affect correlation functions of point operators, but they affect correlation functions of extended operators at separated locations.  The level $k$ in Chern-Simons theory is a well-known example of this type, because this theory does not have point operators. Additional examples include certain discrete theta parameters in 3+1d gauge theory \cite{Aharony:2013hda}.   
    \item Other parameters are choices of contact terms.  These affect correlation functions only when operators touch each other. (Some of these contact terms can be changed by a redefinition of the operators.)  One way such coupling constants arise is when we couple the theory to background fields, and then these terms are nonlinear expressions in the background fields.
    \end{itemize}

The parameters in an added invertible theory, such as $\frac{k}{4\pi}AdA$ with $k\in \bZ$, are of the last kind.  To see that, note that $\frac{k}{4\pi}AdA$ affects the two-point function of the current only at coincident points. (See the related discussion in \cite{Closset:2012vp}.) Such parameters can still have important physical consequences as they may lead to different boundary states.

By the definition adopted in this paper, a TQFT does not have any nontrivial local point operators. So all the parameters are of the last two types. Furthermore, theories that differ only by their contact terms have the same correlation functions of extended operators as long as they do not touch each other.  Indeed, such theories have the same modular tensor category data.  Therefore, for the purpose of this paper, we will identify them.

However, if the TQFTs are effective low-energy descriptions of gapped states, adding nontrivial invertible theories can change the phase of matter and lead to experimentally observable effects, such as electric or thermal Hall conductivity and gapless edge modes. These effects originate from short-distance physics, and as such, they affect only contact terms. (From this perspective, boundary states are like contact terms.)

\section{From a response theory to a one-form symmetry}\label{from response to}

In this section, we will provide another perspective on the conclusions of Sec.~\ref{flux threading}.  We will discuss the response theory of the FQHE and will relate the flux threading argument in Sec.~\ref{flux threading} to the fact that the response theory is not globally well-defined.  This will allow us to constrain the TQFT describing the system.

The response of the quantum Hall state to external electric and magnetic fields given by Eqs.~\eqref{qhresp}--\eqref{streda} can be compactly summarized by writing an effective theory for a background electromagnetic gauge field $A$,
\begin{equation}\label{AdAresp}
{\sigma_H\over 4\pi } \int_\tmfd  A\wedge dA\,.
\end{equation}
That this is the correct response theory can be seen by taking derivatives with respect to the spatial components $A_i$ to reproduce Eqs.~\eqref{qhresp} and \eqref{Halldefn}, or with respect to the time component $A_0$ to reproduce  \eqref{streda}. 

The theory given by  \eqref{AdAresp} represents the response of the quantum Hall system to a long-wavelength, low-frequency background gauge field $A$ with small magnitude. It should be viewed as the leading term in an expansion in derivatives and in $A$. 

The action \eqref{AdAresp} is well defined for topologically trivial gauge fields.  However, for generic $\sigma_H$, it is not meaningful when $A$ carries nonzero fluxes.  Following \cite{Witten:1988hf}, we can try to define it by viewing our 2+1d (arbitrary oriented) spacetime $\tmfd$ as the boundary of a 3+1d bulk $\fmfd$, and define the action \eqref{AdAresp} as
\begin{equation}\label{fourda}
    {\sigma_H\over 4\pi }\int_\fmfd dA\wedge dA\,.
\end{equation}
Unlike \eqref{AdAresp}, this expression is meaningful even for topologically nontrivial $A$.  However, it might depend on the choice of $\fmfd$ and on the way the background gauge field $A$ is extended into $\fmfd$.  

How can a 2+1d theory on $\tmfd$ depend on an extension of $A$ to the bulk $\fmfd$?

Before we answer this question, we observe that the theory \eqref{fourda} is independent of the bulk for the following values of $\sigma_H$: for bosonic theories it must be an even integer, for fermionic/spin theories it can be any integer, and for spin$^c$ theories it can also be any integer.\footnote{More precisely, this statement is true only if we add an appropriate term depending on the metric, which represents the thermal Hall conductivity.}  These values of $\sigma_H$ correspond exactly to the allowed `integer' QHE states for bosons and fermions. Thus, it seems that the response theory of \eqref{AdAresp} is well defined only for the integer states. 

We are, however, interested in the response theory \eqref{AdAresp} with a fractional coefficient
\begin{equation}
    \sigma_H={p\over q}\qquad, \quad \gcd(p,q)=1\,.
\end{equation}
Then \eqref{AdAresp} is not a meaningful 2+1d theory on $\tmfd$.  This happened because we have integrated out the modes of a nontrivial TQFT. We refer to this TQFT as ${\cal T}$.\footnote{Here we assume that our FQH state is gapped.}  

The rest of our discussion will focus on what we can learn about ${\cal T}$ given that we know it must produce \eqref{AdAresp} when it is integrated out. For simplicity, in the discussion below, we will focus on bosonic systems, and then we will review the analogous conclusions for fermions.

We pause to note that the discussion of flux threading in  Sec.~\ref{flux threading} is a manifestation of the fact that \eqref{AdAresp} is globally ill defined.  Consider a one-parameter family of background gauge fields depending on a parameter, Lorentzian time, or Euclidean time, shifting $\oint A_\varphi d\varphi$ by $2\pi$.  This is a closed path in the space of $A$'s, but the integral of \eqref{AdAresp} is not single valued under it. Instead, it is shifted by ${\sigma_H\over 2\pi}\int d\varphi dr dt E_r$, signaling the transfer of charge $\sigma_H$ from one side of the annulus to the other.  

Our main point is that the TQFT $\cal T$ must have a global one-form symmetry.  As a first sign of such a symmetry, we note that in addition to the ordinary $\U$ gauge symmetry of $A$, for a closed manifold $\fmfd$, the action \eqref{fourda} also has a $\U$ one-form gauge symmetry  \cite{Kapustin:2014gua}
\begin{equation}\label{qoneform}
    A\to A+q\lambda\,,
\end{equation}
with $\lambda$ a properly normalized $\U$ gauge field. [Here, we assume for simplicity that $pq$ is even.  Otherwise, we should replace $q$ in \eqref{qoneform} with $2q$.] This means that the only information in $A$ is the same as the information in a gauge field for a $\mathbb{Z}_q^{(1)}$ one-form symmetry.  The corresponding $\mathbb{Z}_q^{(1)}$ gauge invariance is violated when $\fmfd$ has a boundary.  Then, depending on the boundary conditions, this $\mathbb{Z}_q^{(1)}$ can act as a global symmetry.  Even though this is not a symmetry of our problem, we will soon see that a related one-form symmetry must be present.

Motivated by this, we assume that the TQFT ${\cal T}$, has a $\mathbb{Z}_s^{(1)}$ global one-form symmetry for some $s$.  (We will discuss the most general case in Appendix~\ref{app:generator}.)  As discussed in  Sec.~\ref{sec:symm_anom_gauge}, in this case, it is natural to couple $\cal T$ to a background $\mathbb{Z}_s$-valued gauge field $\cal B$.  In general, the one-form symmetry has an 't Hooft anomaly, which can be characterized by an anomaly theory in one higher dimension.
In our case, it is
\begin{equation}\label{BB}
  {2\pi r\over 2s}\int_{\fmfd} {\cal P}({\cal B})\qquad , \qquad rs\in 2\mathbb{Z}\,,
\end{equation}
with ${\cal P}({\cal B})$ the Pontryagin square operation.  (See Sec.~\ref{sec:symm_anom_gauge}.) As above, our 2+1d spacetime $\tmfd$ is the boundary of the 3+1d manifold $\fmfd$.\footnote{Note that $\fmfd$ does not have to be physical; it is merely a mathematical device.  Having said that, we can view the bulk $\fmfd$ as an SPT phase of a one-form $\mathbb{Z}_s^{(1)}$ symmetry, classified by an integer $r \mod 2s$ with even $rs$. The corresponding boundary 2+1d theory has 't Hooft anomaly for the $\bZ_s^{(1)}$ one-form symmetry that is also labeled by the same integer $r \mod 2s$.}   We emphasize that $r$ and $s$ are not necessarily coprime.

So far, we have not yet related the expression \eqref{fourda} for $A$ and the expression \eqref{BB} for $\cal B$.  Before doing that, we remind the reader that as we said after \eqref{dAdAt}, while \eqref{fourda} is well defined even when $\fmfd$ has a boundary $\tmfd=\partial \fmfd$, this is not the case with \eqref{BB}.  The ``problem'' with \eqref{BB} associated with the boundary $\tmfd=\partial \fmfd$ is exactly canceled by the anomaly of the theory $\cal T$ on $\tmfd$.

As we mentioned around \eqref{qoneform}, the expression \eqref{fourda} has a one-form gauge symmetry, which is violated in the presence of a boundary.  As we will soon see, this phenomenon is closely related to the anomaly in $\cal B$ and explains how \eqref{fourda} is related to it.

We now want to couple the system to a background $\U$ gauge field $A$ in a standard way.  For simplicity, take $ {\cal B}=dA/2\pi$. More precisely, $\cal B$ is a  $\bZ_s$ cocycle and $dA$ is a continuous differential form, such that for every closed two-cycle $\Sigma$, we have $\int_\Sigma {\cal B}=\left({1\over 2\pi }\int_\Sigma dA\right)\text{ mod } s$. Equivalently, denote the mod $s$ reduction of $dA /2\pi$ by $\left[dA/2\pi\right]_s \in H^2(\fmfd,\mathbb{Z}_s)$.  Then, the more precise version of ${\cal B}=dA/2\pi$ is
\begin{equation}\label{BdA}
{\cal B}=\left[{dA\over 2\pi}\right]_s\,.
\end{equation}
Either way, we see that the two-form background field $\cal B$ is determined by the $\U$ background gauge field $A$. (See e.g., \cite{Benini:2018reh,Delmastro:2022pfo,Brennan:2022tyl, Transmutaion}).

Using \eqref{BdA}, the anomaly theory \eqref{BB} becomes
\begin{equation}\label{dAdA}
{r\over 4\pi s}\int_\fmfd dA\wedge dA\,.
\end{equation}
\{For $\cal B$ given by \eqref{BdA}, $\cal B$ has an integer lift and the Pontryagin square in \eqref{BB} is simply a product of differential forms \cite{Kapustin:2013qsa}.  Compare with the discussion around \eqref{PB}-\eqref{dAdAt}.\}
This expression can be interpreted as the boundary term \eqref{AdAresp}, provided
\begin{equation}\label{ruspq}
{r\over s}=\sigma_H \mod 2={p\over q}\mod 2\qquad , \qquad \gcd(p,q)=1\,.
\end{equation}
Here, the modulo 2 follows from the fact that the anomaly expression \eqref{BB} depends only on $r\mod 2s$.  Equivalently, we can add to the boundary theory a counterterm ${n\over 2\pi}\int_\tmfd A\wedge dA$ with an arbitrary integer $n$, which shifts $\sigma_H$ by $2n$. Physically, this corresponds to stacking bosonic IQH states, though we note that this will change the Hall conductivity. Nevertheless, we club together all theories with the same fractional part of $\sigma_H$ to focus on their anomalous one-form symmetry.  See the discussion in  Sec.~\ref{sec:tqft-gappedphase}.

We conclude that the response theory \eqref{AdAresp} is reproduced by a TQFT with a $\mathbb{Z}_s^{(1)}$ one-form symmetry and anomaly labeled by $r$ such that $r/s=p/q\mod 2$.\footnote{A related discussion of the fractional part of the coefficient of a Chern-Simons term as a diagnostic of a nontrivial IR theory has appeared (in the context of duality and gapless theories) in \cite{Closset:2012vp}.} In App.~\ref{app:generator} we will show that what we have done is essentially the only way to reproduce \eqref{AdAresp}.

Let us discuss the one-form symmetry of this TQFT in a different light. A physical manifestation of the anomaly is that the generator line $a_1$ of the $\bZ_s^{(1)}$ one-form symmetry transforms under it with the phase $\exp(2\pi ir/s)$ \cite{Hsin:2018vcg}. This statement is equivalent to $B(a_1,a_1)=r/s\mod 1$. Furthermore, the spin of the line is given by $h(a_1)=r/(2s)\mod 1$. Thus, in a TQFT, the $\bZ_s^{(1)}$ one-form symmetry with anomaly \eqref{BB} can be understood equivalently as a $\bZ_s$ group of Abelian anyons generated by $a_1$ with the braiding statistics given above.

In  Sec.~\ref{sec:symmetrylines}, we will discuss these anyons in detail.  For now, we will simply note that they are given by $a_j$ with spins and $\U$ charges
\begin{equation}
\begin{split}
    &a_j=a_1^j\qquad,\qquad j=0,1\cdots, s-1\\
    &h(a_j)={rj^2\over 2s}\mod 1\\
    &Q(a_j)={rj\over s}\mod 1.
\end{split}
\end{equation}
In particular,
\begin{equation}
\left(h(a_1),Q(a_1)\right)=\left({p\over 2q} , {p\over q}\right)\mod 1\,,
\label{eq:hQa1}
\end{equation}
consistent with the identifications discussed in Sec.~\ref{sec:tqft-basics}. We immediately recognize the vison in \eqref{vison-h-Q} as $v=a_1$, and $s$ is the order of the vison, i.e., the minimal positive integer such that $v^s=1$. 

While the discussions after \eqref{ruspq} are specific to bosonic theories, the conclusions can be straightforwardly adopted to other cases. For spin{}$^c$ (or electronic) theories, the response theory still leads to a $\bZ_s^{(1)}$ one-form symmetry with anomaly labeled by $r$, now with
\begin{equation}
    \frac{r}{s}=\sigma_H\mod 1\,,
\end{equation}
since one can add $\frac{1}{4\pi}\int_{\tmfd}A\wedge dA$ to the boundary theory to shift $\sigma_H$ by $1$. The vison $v=a_1$ still generates the $\bZ_s^{(1)}$ group, with
\begin{equation}
\begin{split}
    &a_j=a_1^j\qquad,\qquad j=0,1\cdots, s-1\\
    &h(a_j)={rj^2\over 2s}\mod \frac12\\
    &Q(a_j)={rj\over s}\mod 1.
\end{split}
\end{equation}

In both the bosonic and the electronic cases, one can understand $s$ through the minimal fractional charge of anyons: there must exist an anyon with charge $1/s$ mod 1 \cite{levin2009fractional,  Dehghani:2020jls, Cheng:2022nds}.

\section{Charged particles in a magnetic field}
\label{sec:UV}

The statement that a fractionally quantized Hall conductivity implies the emergence of Abelian anyon excitations with fractional charge is well known. When restated in terms of the higher-form symmetry of the low-energy theory, the conclusion of  Sec.~\ref{from response to} might look strange.  Where did the one-form symmetry of the low-energy theory come from?  An obvious possible answer is that it can be an emergent symmetry, i.e., a symmetry of the IR theory without a precursor in the UV.  Although this is certainly a possibility, in this section, we will demonstrate that in the case of particles in a uniform magnetic field, the IR one-form symmetry and its anomaly started their lives in the UV theory as ordinary space translations.

This discussion follows a long and productive tradition of studying UV systems with some special exact symmetries. (In many examples, these special UV symmetries make the model solvable.) Then, one tracks these symmetries to the IR, finding some special symmetries there.  One example is the 1+1d Ising model with its large symmetry leading to integrability.  Another example is the 2+1d toric code with its large UV symmetry, which becomes a one-form symmetry in the IR.

The study of such models with a large UV symmetry is particularly useful when the symmetry of the IR theory is robust.  This is the case when the IR theory does not have any relevant operators that violate it.  In that case, we can break the special symmetry slightly in the UV, and even though the UV system does not have that symmetry, the IR theory has that symmetry as an emergent symmetry.

In our case, we will discuss a UV system with an exact translation symmetry with an anomaly.  This translation symmetry acts in the IR as a one-form symmetry.  As in the examples above, this one-form symmetry is robust.  We can break the translation symmetry in the UV and still find that one-form symmetry in the IR.  In that case, the IR one-form symmetry is an emergent symmetry.  This emergent one-form symmetry is diagnosed by the fractional part of the Hall conductivity $\sigma_H$.

More generally, even if this UV translation symmetry is completely absent (such as in lattice systems showing a $B = 0$ FQHE, or real materials with impurities that break its exact translation symmetry), the fractionally quantized $\sigma_H$ signals the one-form symmetry of the low-energy theory.

\subsection{The UV theory}

We start by discussing the symmetry and anomalies in the UV theory, and close by showing how these are realized in the IR theory.

We consider $\cal N$ charge-one particles on a 2+1d  torus with coordinates $x\sim x+L_x$ and $y \sim y+L_y$ with a constant magnetic field $B={2\pi {\cal K}/  L_xL_y}$.  The total flux ${\cal K}$ is quantized.  The filling fraction is
\begin{equation}\label{nudef}
\nu={{\cal N}\over {\cal K}} = {p\over q}\quad , \qquad \gcd(p,q)=1\quad, \qquad {\cal N},{\cal K}\gg 1\,.
\end{equation}
The integers $\cal N$ and ${\cal K}$ are macroscopic, while the integers $p$ and $q$ are of order one.  (For simplicity, we will take $\cal N$ and ${\cal K}$ to be divisible by any number of order one that we will need.) In this section, we will impose translation symmetry and study its consequences. We will phrase these in the framework of anomalies and their matching between the UV and IR. Our discussion will be complementary to the classic analysis of many particle translation symmetries of charged particles in a Landau level in  \cite{haldane1985many}. 

Below, we will also consider the second-quantized theory, where the total number of particles is
\begin{equation}\label{NandQ}
Q={\cal N}+\delta Q
\end{equation}
with $\delta Q$ an operator whose eigenvalues are of order one.  Similarly, we will write the background gauge field as
\begin{equation}\label{Adef}
A_\mu={\cal K}A_\mu^{(0)} +\delta A_\mu   \quad, \qquad \mu=t,i  \quad, \qquad i=x,y\,,
\end{equation}
where $A^{(0)}_\mu$ is a reference gauge field with with constant magnetic field with ${\cal K}=1$. Then $\delta A_\mu$ is an order one change in the background gauge field.  In particular, we will often take it to be topologically trivial and even constant.

Importantly, the reference gauge field $A_\mu^{(0)}$ in \eqref{Adef} depends on a choice of origin in space.  To see that, consider the Wilson lines (holonomies) of $A^{(0)}_\mu$ around the $x$ and $y$ directions.  They can be taken to be
\begin{equation}\label{Holonomies0}
\begin{split}
w_x&=\exp\left( i\int dx' A_x^{(0)}(x',y)\right)=\exp\left( -2\pi i{y\over L_y}\right)\\
w_y&=\exp\left( i\int dy' A_y^{(0)}(x,y)\right)=\exp\left(2\pi i{x\over L_x}\right)\,.
\end{split}
\end{equation}
(These Wilson lines receive contributions both from the gauge field and the transition functions, so that the expression in terms of the gauge fields $A^{(0)}_i$ is incomplete.)
Alternatively, we can shift $x$ and $y$ by constants.  This has the effect of shifting $A^{(0)}_i$ by constants and multiplying the Wilson lines \eqref{Holonomies0} by phases.  We will follow the choice \eqref{Holonomies0} and will set $A^{(0)}_t=0$.  Then, a change of the origin in space can be absorbed in a constant shift of $\delta A_i$ in \eqref{Adef}.

For a constant  magnetic field $B=2\pi {\cal K}/L_xL_y$, the Wilson lines around $x$ and $y$ are
\begin{equation}\label{Holonomies}
\begin{split}
&W_x(\sigma_x,{\cal K})=e^{i\sigma_x} w_x^{\cal K}=\exp\left(i\sigma_x -2\pi i{{\cal K}y\over L_y}\right)\\
&W_y(\sigma_y,{\cal K})=e^{i\sigma_y} w_y^{\cal K}=\exp\left(i\sigma_y +2\pi i{{\cal K}x\over L_x}\right)\,.
\end{split}
\end{equation}
The parameters $\sigma_i$ correspond to $\delta A_i={\sigma_i/  L_i}$.  Changing the origin of $x$ or $y$ shifts $\sigma_y$ and $\sigma_x$ .  Hence, these constants signal the fact that the classical $\U$ translation symmetry around every direction is only $\mathbb{Z}_{\cal K}$, rather than $\U$.  They are generated by
\begin{equation}\label{coords}
\begin{split}
&x\to x+{L_x\over {\cal K}}\\
&y\to y+{L_y\over {\cal K}}\,.
\end{split}
\end{equation}
Under these translations, the holonomies remain the same.

We now write down operators that implement the $\bZ_{\cal K} \times \bZ_{\cal K} $ symmetry in order to study their algebra and determine the anomaly.\footnote{Here, $\bZ_{\cal K} $ is a translation symmetry, rather than an internal symmetry.  The superscript $(0)$, which we use to denote zero-form symmetries, signals the fact that the symmetry operators act on all of space.  We will soon see how this symmetry is transmuted into a one-form internal symmetry.}  First, let us define the many-particle version of the holonomies as follows:
\begin{equation}
\begin{split}
    \mathcal{U}_x&=\exp\left[-\frac{2\pi i}{L_y}\sum_{\alpha=1}^{{\cal N}}y_\alpha\right]\quad,
   \qquad \mathcal{W}_x(\sigma,{\cal K}) = e^{i\mathcal{N}\sigma_x}\mathcal{U}_x^{\cal K}\,,\\
   \mathcal{U}_y&=\exp\left[\frac{2\pi i}{L_x}\sum_{\alpha=1}^{{\cal N}}x_\alpha\right]\quad,
   \qquad
    \mathcal{W}_y(\sigma,{\cal K}) =e^{i\mathcal{N}\sigma_y}\mathcal{U}_y^{\cal K}\,.
\end{split}
\end{equation}
Here $x_\alpha, y_\alpha$ are the Cartesian coordinates of the $\alpha$-th particle. $\mathcal{U}_i$ is also referred to as the large gauge transformation, as it changes the holonomy by $2\pi$.

Let the (second-quantized) operators implementing \eqref{coords} be $T^i_0$.  Without a background $A_\mu$, $(T_0^i)^{\cal K}=1$  and with a nonzero $A_\mu={\cal K}A^{(0)}_\mu$, we have twisted boundary conditions, $(T_0^i)^{\cal K}=\mathcal{W}_i(\sigma_i,{\cal K})$.  (Recall that we consider a system with charge $\cal N$.)  In this case, it is standard to redefine $T^i=T_0^i \mathcal{U}_i^{-1}$ and find
\begin{equation}\label{Trela}
\begin{split}
&(T^i)^{\cal K}=e^{i{\cal N}\sigma_i}\\
&T^xT^y=e^{2\pi i {{\cal N}\over {\cal K}}} T^y T^x\,.
\end{split}
\end{equation}
More generally, in the second quantized theory, we can replace $\cal N$ with the charge $Q$ of \eqref{NandQ}. Here, $T_i$ are symmetries of the Hamiltonian.

In the first equation in \eqref{Trela}, $\sigma_i$ are related to a background $\delta A_i^{(0)}$ and can be thought of as due to topological defects associated with the global $\U$ symmetry. (Standard flux threading discussions analogous to those in  Sec.~\ref{flux threading} and its counterpart in  Sec.~\ref{from response to} shift these values.) The analysis of such defects leads to a mixed anomaly between the discrete ${\mathbb Z}_{\cal K} $ translations and $\U$, which in turn leads to the famous filling constraints \cite{OshikawaLSM, cheng_translational_2016, Else:2020jln,Cheng:2022sgb}.  Such a change of the background field by $\sigma_i$ with vanishing field strength is often referred to as ``flux insertion,'' and it is done when the total field strength vanishes.  Here, the background magnetic field is nonzero, but we can still deform the theory with $\delta A$ with vanishing field strength and use the standard techniques of such defects. 

The second equation in \eqref{Trela} means that this discrete translation symmetry is realized projectively.  This can be interpreted as anomalies in the discrete translation symmetry and the global $\U$ symmetry.\footnote{See Appendix~A in \cite{Seiberg:2024yig}, for a pedagogical discussion of the well-known Landau levels from the perspective of 't Hooft and Adler-Bell-Jackiw anomalies.}  This interpretation is significant because it leads to 't Hooft anomaly matching conditions.  In particular, the same anomalous phases should exist in the IR theory.  We will use this constraint to classify the possible gapped phases.

To connect with the discussions in Sec.~\ref{flux threading}, we notice that $T_0^x$ changes the holonomy by $2\pi$. The same effect can be enacted by turning on $\delta A_y=2\pi/L_y$, which is the torus version of the flux-threading process. $T_0^x$ combined with the large gauge transformation $h_y$ yields the symmetry transformation $T^x$ of the system. 

\subsection{From the UV to the IR}\label{sec:UVIR}

The anomaly in \eqref{Trela} should be matched by the long-distance theory.  Given that we are interested in gapped phases, which are described by topological field theories, this raises three questions. 
\begin{itemize}
    \item How can the translation symmetry act in a topological theory, where the translation symmetry is trivial?  
    \item How can we realize an anomaly in that translation symmetry? 
    \item How can a symmetry with a macroscopic number of elements \eqref{Trela} act in a TQFT with a finite number of objects?
\end{itemize}

Let us start with the third question.  Most of the symmetry \eqref{Trela} is anomaly-free, and therefore it can act trivially in the IR.  Let the generators $T^i$ flow in the IR to $T_{\rm IR}^i$. Then, the elements $(T_{\rm IR}^i)^q$ do not participate in the anomaly and they can act trivially;  i.e., act as central elements in the IR theory.  We will allow a more general situation depending on another integer of order one $\ell$ and define\footnote{This matches with the notation used in Secs.~\ref{from response to} and \ref{sec:symmetrylines}.  Comparing with the discussions in Appendix~\ref{app:generator}, here we have $u=1$.}
\begin{equation}
\begin{split}
&s=\ell q\\
&r=\ell p\,.
\end{split}
\end {equation}
Then, for $\sigma_i=0$, we have
\begin{equation}\label{sdefa}
(T_{\rm IR}^i)^s=C^i \qquad , \qquad i=x,y\,,
\end{equation}
with $C^i$ acting as a central element in the IR theory.  Then, the number of symmetry elements in the IR is finite.

In the examples below, $C^i$ carries spin and charge. To match with \eqref{Trela}, $C^i$ has $\U$ charge $r$ (use, ${\cal N}/{\cal K}=r/s$) and then \eqref{Trela} becomes
\begin{equation}\label{TrelaIR}
\begin{split}
&(T_{\rm IR}^i)^s=e^{ir\sigma_i}C^i\\
&T_{\rm IR}^xT_{\rm IR}^y=e^{2\pi i {r\over s}} T_{\rm IR}^y T_{\rm IR}^x\,.
\end{split}
\end{equation}
This is a $\mathbb{Z}_s  \times \mathbb{Z}_s $ symmetry, where the two factors correspond to $x$ and $y$ translations, and it is realized projectively.  The anomalous phases, which are labeled by $r$, signal the 't Hooft anomaly.

In conclusion, the elements $T_{\rm IR}^i$ generate a finite symmetry \eqref{TrelaIR}, which is a quotient of the UV symmetry  \eqref{Trela}.  This quotient realizes the full anomaly.  We note that at this stage we could have set $\ell=1$, but below we will see why it is better to keep an arbitrary $\ell$.

To address the first two questions at the beginning of this subsection, let us see how this anomaly is matched in the simple case where $\nu=1/q$ and the IR theory is a $\U_q$ Chern-Simons theory with the Lagrangian
\begin{equation}\label{Uoneq}
{q\over 4\pi } b\wedge db -{1\over 2\pi} A\wedge db\,.
\end{equation}
Here, $b$ is a dynamical $\U$ gauge field and $A$ is the background gauge field.  For $A=0$, this Lagrangian has a continuous translation symmetry (in fact, it is topological) and a ${\mathbb Z}_q^{(1)}$ one-form symmetry shifting $b$. With $A={\cal K}A^{(0)}$, it seems that there is no continuous translation symmetry. However, the theory based on \eqref{Uoneq} is invariant under a continuous translation transformation $x\to x +\epsilon^x$ combined with the one-form transformation $b_y\to b_y+{2\pi {\cal K}\epsilon^x \over qL_xL_y}$ (use $A^{(0)}dA^{(0)}=0$) and similarly for $y$ translations.  Equivalently, we can redefine $\hat b=b-{\cal K}A^{(0)}/q$ (recall that we assumed that ${\cal K}/q\in \mathbb{Z}$) in \eqref{Uoneq} to write it as
\begin{equation}\label{bdbh}
{q\over 4\pi } \hat b\wedge d\hat b \,,
\end{equation}
which is manifestly translation invariant and has the ${\mathbb Z}_q^{(1)}$ one-form symmetry.  Furthermore, that ${\mathbb Z}_q^{(1)}$ one-form symmetry has an anomaly labeled by $p=1$.

To relate to \eqref{TrelaIR}, we can turn on $\delta A_i=\sigma_i/L_i$.  Then, \eqref{bdbh} becomes  
\begin{equation}
{q\over 4\pi } \hat b\wedge d\hat b -{1\over 2\pi} \delta A \wedge d\hat{b}\,
\end{equation}
and we have
\begin{equation}
    e^{iq\oint_i \hat b}=e^{i\sigma_i}\,
\end{equation}
where $\oint_i$ is a contour integral around the $i$ direction.  (For odd $q$ we also need a transparent fermion line on the right-hand side.)
This is consistent with the fact that this theory has $p=r=1$.

Tracking the exact discrete $\mathbb{Z}_{\cal K}  \times \mathbb{Z}_{\cal K} $ translation symmetry of the UV theory \eqref{coords} and its anomaly, we see that it becomes a discrete $\mathbb{Z}_{\cal K}^{(1)}$ one-form symmetry, with only a $\mathbb{Z}_q^{(1)}$ quotient of it acting faithfully.  The anomaly in the UV translations is matched by an anomaly in that one-form symmetry.\footnote{The fact that in this case, a UV translation symmetry can act in the IR as a one-form symmetry has already been pointed out in \cite{JensenRaz}. Other symmetry transmutations of zero-form symmetries into higher-form symmetries and their anomalies are discussed in \cite{Transmutaion}.} This anomaly matching was studied using the G-crossed category theory in \cite{Manjunath:2020kne}. 

We conclude that the UV symmetry \eqref{Trela} can be realized in the IR as a TQFT with a $\mathbb{Z}_s^{(1)}$ one-form symmetry with anomaly labeled by $r$.  Specifically, $T_{\rm IR}^i$ are realized as a one-form symmetry elements along the direction $i$.  Furthermore, the coupling of this TQFT to the background gauge field $A$ is labeled by the integer $r$.

  Finally, given that the IR version of the translation generators $T_{\rm IR}^{i}$ act in the IR as one-form symmetry operators, we can tie this system to the flux threading argument in Secs.~\ref{flux threading} and \ref{from response to}.  We identify $T_{\rm IR}^{i}$ with the one-form symmetry generated by the vison $v$.  Then, the first equation in \eqref{TrelaIR} is the statement that $v$ carries charge $r/s=p/q$, the second equation in \eqref{TrelaIR} is the statement that $B(v,v)=r/s=p/q$, and equation \eqref{sdefa}, which states which subgroup of the translation group is realized trivially in the IR, determines the power of the vison that is trivial.  In other words, we have learned that the UV precursor of the vison $v$ is a UV translation symmetry generator.

\section{Connection between translation and flux threading}\label{conntranf}

We now provide another perspective on the relation between the UV translations and the IR one-form symmetry, which is also applicable in other symmetry enriched topological (SET) phases.

\begin{figure*}
    \centering
    \subfigure[\ Setup]{
        \includegraphics[width=0.45\textwidth]{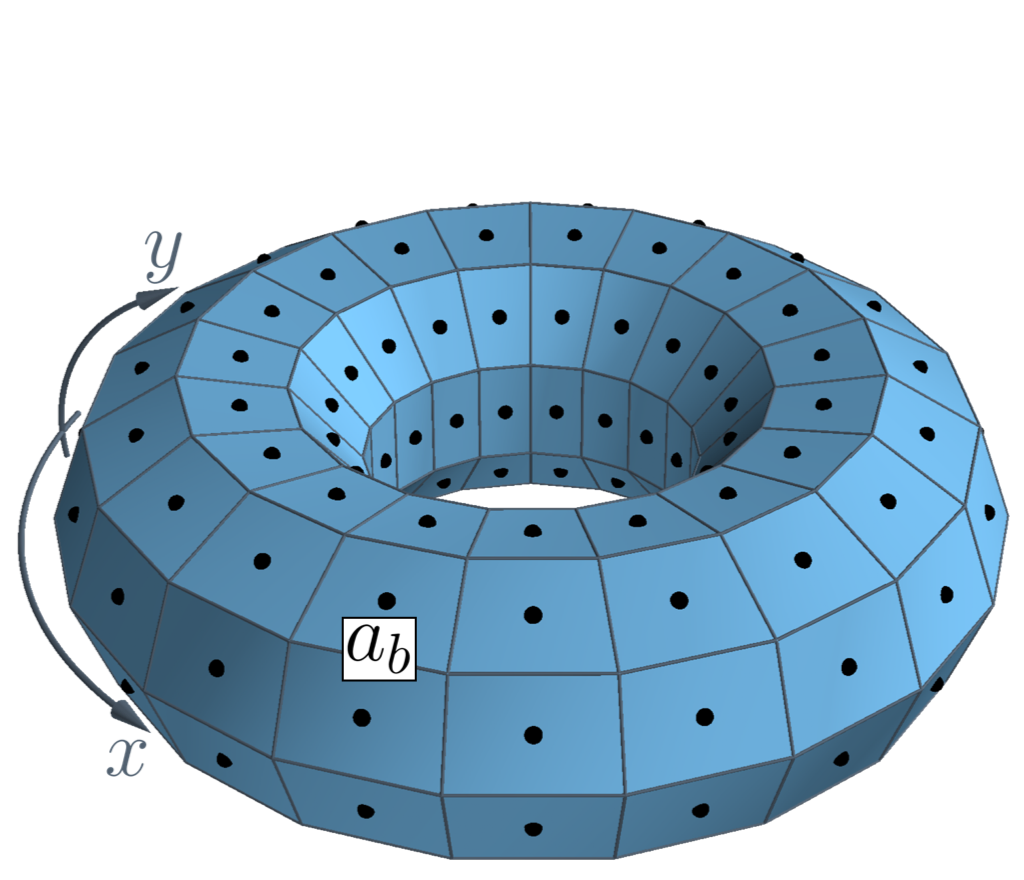}
        \label{fig:base}
            }
    ~ 
    \subfigure[\ Translation with flux]{
        \includegraphics[width=0.45\textwidth]{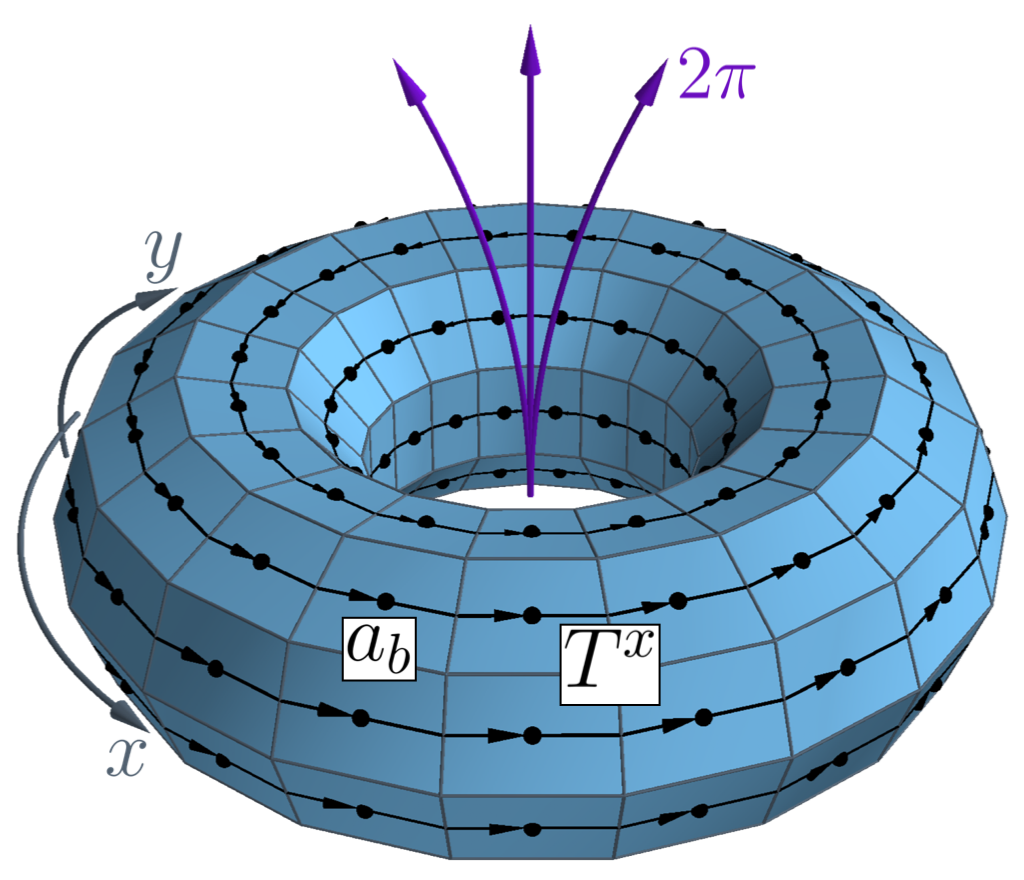}\
        \label{fig:translation}
            }
    \caption{A cartoon of the action of translation on a TQFT that respects translation and $\U$. (a) Display of the setup in a toroidal geometry with $N_x$ unit cells along the $x$ direction and $N_y$ along the $y$ direction. The filling constraints ensure that there is a background anyon, $a_b$, at the center of each unit cell with a charge $Q(a_b) = \nu_l$. If the system is translated in the $x$ direction by the size of the cell, as in (b), each $a_b$ moves to its neighbor. This can be seen to produce $N_y$ loop defects of $a_b$. In the presence of $2\pi$ flux threading through the hole of the torus, translation thus leads to a phase of $2\pi \nu_l N_y$. We note that we cannot embed a flat torus in three-dimensional space. Hence, the torus in the figure is not flat. However, since we discuss only discrete translation, we can have interactions preserving them.}
	\label{fig:torus}
\end{figure*}

Consider a system of particles, either fermions or bosons, with a conserved particle number corresponding to a global $\U$ symmetry. The particles are taken to move in a periodic potential with discrete translation symmetries $\hat{T}^x, \hat{T}^y$ along the $x,y$ directions of space. These translations commute with each other.  We assume a fixed rational filling per unit cell $\nu_l$.\footnote{We denote the lattice filling by $\nu_l$. It is defined as the ratio of the total number of particles to the total number of unit cells in the thermodynamic limit.}

We expect a gapped ground state to be described by a TQFT, which might be trivial.  Then, from the Lieb-Schulz-Mattis-Oshikawa-Hastings theorem \cite{lieb_two_1961,OshikawaLSM, oshikawa_commensurability_2000, hastings_lieb-schultz-mattis_2004}, we know that for fractional $\nu_l$ gapped ground states must have nontrivial topological order. In particular, they must have an Abelian background anyon, $a_b$, with charge $Q(a_b) = \nu_l$. This is because translation acts projectively on every anyon $a$ as $(\hat{T}^y)^{-1}(\hat{T}^x)^{-1}\hat{T}^y\hat{T}^x[a] = e^{2\pi i \varphi(a)}[a]$.\footnote{Here $[a]$ refers to the state of a single anyon. In the continuum theory, this is a defect line of $a$ stretching along the time direction.} In other words, the action of translation is given by a phase $\varphi$ that respects fusion. (We restrict attention to phases where the action of translation does not permute anyons.) The results of \cite{barkeshli_symmetry_2019}, summarized in Appendix~\ref{app:spinc-classification}, thus tell us that we can write $(\hat{T}^y)^{-1}(\hat{T}^x)^{-1}\hat{T}^y\hat{T}^x[a] = e^{2\pi i B(a,a_b)}[a]$ for some Abelian $a_b$. The projective action of translation thus means we can think of the background as having $a_b$ in each unit cell~\cite{jalabert1991spontaneous,senthil_z_2_2000, Zaletel:2014epa, cheng_translational_2016, bonderson2016topologicalenrichmentluttingerstheorem}. As such $a_b$ is often called the ``background anyon.''

Now thread $2\pi$ flux of the background $\U$ gauge field to nucleate the vison $v$ discussed in Sec.~\ref{flux threading}. Because of the assumed discrete translation symmetry, each unit cell should have charge $\nu_l$. That charge density can be measured by braiding the vison $v$ around the unit cell. But we saw earlier that the braiding phase of $v$ around the unit cell is given by $B(v,a_b)$. Thus $\nu_l = B(v,a_b) = Q(a_b)$. These statements on the background anyon have a very intuitive interpretation: at fractional lattice filling $\nu_l$, a symmetry-preserving gapped state is obtained by placing the Abelian anyon $a_b$ with fractional charge $\nu_l$ in each unit cell.
We note that these constraints can be generalized to the magnetic lattice translation symmetry~\cite{lu2020filling}.

Now, consider the action of a unit translation on the ground state. For concreteness, consider a system with $N_x \times N_y$ unit cells placed on a torus. We label the unit cells by $\vec R = m_x \hat{x} + m_y \hat{y}$, $m_i = 1,....., N_i$. Acting with, say, $\hat{T}^x$,  shifts $a_b(\vec R)$ to $a_b(\vec R + \hat{x})$. Equivalently, we can think of this process as creating a separate  $a_b - \bar{a_b}$ pair along each of the $N_y$ chains labeled by distinct $m_y$, moving each $a_b$  around its chain and annihilating with $\bar{a_b}$. In other words, $\hat{T}^x$ maps, in the IR theory, to a line operator of $a_b$ on each of the $N_y$ chains. Now consider acting with $\hat{T}^x$ in a ground state, where there is a vison line stretched along the $y$ cycle. Such a vison line can be created by ``threading a $2\pi$ flux of the background $\U$ gauge field" [i.e., creating a holonomy of the background $\U$ gauge field]. More precisely, starting with any ground state, we first apply a vison line operator along the $y$-cycle, and then act on the state with $\hat{T}^x$, and compare the result with the one obtained from acting in the opposite order. Because of the mutual braiding between $a_b$ and $v$, the $\hat{T}^x$ action now produces a phase $2\pi \nu_l N_y$. We show a cartoon of this in Fig.~\ref{fig:translation}. This is exactly the answer one would infer from the Oshikawa adiabatic flux-threading argument applied to this system \cite{bonderson2016topologicalenrichmentluttingerstheorem}. 

Given this general understanding of the translation action in SET phases with filling constraints, let us turn to the system of interest in this section, namely, charged particles in a Landau level. As we have seen, on a torus with $\cal K$ flux quanta, the system has an anomalous $\bZ_{\cal K} \times \bZ_{\cal K} $ translation symmetry. Below, we assume that the length of the torus along the $x$ ($y$) direction is $L_x$ ($L_y$).

To draw connections with the  discussions in  Sec.~\ref{sec:UV}, we start with the $\bZ_{\cal K} \times \bZ_{\cal K} $ magnetic translation group generated by $T^x$ and $T^y$ and consider an anomaly free subgroup generated by $\hat{T}^x=(T^x)^{n_x}$ and $\hat{T}^y=(T^y)^{n_y}$, with integers $n_x, n_y$ satisfying $n_xn_y={\cal K}$. It is easy to see that $\hat{T}^x$ and $\hat{T}^y$ commute and generate a $\bZ_{n_y} \times \bZ_{n_x} $ subgroup. This is similar to the lattice translations discussed earlier, and hence the same notations. In fact, one can imagine turning on a weak potential modulated periodically along the $x$-direction with periodicity $n_xL_x/{\cal K}=L_x/n_y$, and along the $y$ direction with periodicity $n_yL_x/{\cal K}=L_x/n_x$, so there are $n_xn_y$ unit cells, each containing exactly $2\pi$ flux. The average charge density per unit cell is ${\cal N}/(n_xn_y)=\nu=p/q$. The quantized Hall conductivity $\sigma_H$ and other universal aspects of the IR theory are preserved under this deformation. The FQH ground state remains invariant both under $\hat{T}^x, \hat{T}^y$ and the global $\U$. 

As with generic SETs with filling constraints, a background anyon $a_b$ with charge $\nu = p/q$ populates each unit cell in the ground state. Further, as we have seen, the presence of fractional $\sigma_H$ guarantees the presence of a vison $v$ with fractional charge $\sigma_H = p/q$. In contrast to a generic SET, in the FQHE state, the vison itself is identified with the background anyon. This can be understood as follows. Suppose we braid an anyon $a$ around the unit cell. Since the unit cell encloses $2\pi$ flux this will lead to a phase of $2\pi Q(a)$. But earlier we saw that braiding $a$ around a unit cell must lead to a phase of $2\pi B(a,a_b)$. So then $B(a,a_b) = Q(a)$. By definition of the vison $Q(a) = B(a,v)$, so $B(a,a_b) = B(a,v)$. Thus $a_b=v$. This makes physical sense; earlier we identified $v$ with $2\pi$ flux. Since each unit cell contains $2\pi$ flux each unit cell contains $v$.

Then, from our discussion earlier, the action of $\hat{T}^x$ corresponds to a line operator of $v$ around the $x$ cycle.  Since the unit cell has length $L_y/n_x$ along the $y$ direction, there are $n_x$ defect loops parallel to each other. We see that the action of $\hat{T}^x$  corresponds, in the IR theory, to a line operator of $v^{n_x}$ along the $x$ cycle. This identification does not depend on the strength of the perturbing potential, and hence holds even in the limit that the potential is turned off. Repeating the same argument for the $y$ direction, we see that $\hat{T}^y$ corresponds in the IR to a line defect of $v^{n_y}$ along the $y$ cycle. 

In Sec.~\ref{sec:UV}, we discussed the magnetic translation symmetry $\bZ_{\cal K} \times \bZ_{\cal K} $ of particles in a homogeneous magnetic field with its anomaly ${\cal N}\over {\cal K}$. The relation to the $\bZ_{n_x} \times \bZ_{n_y} $ translation symmetry discussed here is that the latter is an anomaly free subgroup of the former. 
For example, we can specialize to $n_x=1, n_y={\cal K}$, which leads to the conclusion that $\hat{T}^x=T^x$ creates a line defect of $v$ along the $x$ cycle, in complete agreement with the conclusion reached in the previous section, and the same for $T^y$ by choosing $n_x={\cal K}, n_y=1$.

\section{The symmetry lines}\label{sec:symmetrylines}

In this section, we will study the coupling of the TQFT $\cal T$ to the background gauge field $A$. 

We will start in  Sec.~\ref{pureTQFT} by ignoring the background $A$.  We will consider a TQFT $\cal T$ with a subset of Abelian anyons ${\cal A}\subseteq {\cal T}$.  This subset generates a one-form global symmetry $\otimes_I\mathbb{Z}_{s_I}^{(1)}$. Then, we will focus on a subset of the anyons in $\cal A$ generating a subgroup of that total one-form symmetry
\begin{equation}
\begin{split}
&\mathbb{Z}_s^{(1)}\subseteq \otimes_I\mathbb{Z}_{s_I}^{(1)}\\
&{\cal A}^{s,r}\subseteq {\cal A}\,.
\end{split}
\end{equation}
We will describe the index $r$ of ${\cal A}^{s,r}$ as associated with the 't Hooft anomaly in $\mathbb{Z}_s^{(1)}$.

In  Sec.~\ref{TQFTUb}, we will couple a bosonic TQFT to a background $\U$ gauge field $A$. And in  Sec.~\ref{TQFTUf}, we will discuss spin$^c$ theories, which also come equipped with a $\U$ gauge field (more precisely a spin$^c$ connection $A$.)

\subsection{Pure TQFT -- no charge}\label{pureTQFT}

We consider a TQFT $\cal T$ and use interchangeably the terms anyons and lines for its objects. 
A subset of them are the symmetry lines of an invertible one-form global symmetry $\cal A$.  The properties of these lines were analyzed in \cite{Hsin:2018vcg}.  (See \cite{Moore:1988qv} and \cite{bonderson_non-Abelian_2007} for earlier discussions.) We review the relevant results from \cite{Hsin:2018vcg} and extend them. 

Consider a $\mathbb{Z}_s^{(1)}$ one-form symmetry.  The symmetry lines $a_j$ are generated by $a_1$
\begin{equation}
a_j=a_1^j\,.
\end{equation}
Usually, the integer $j$ is taken to run over $j=0,1,\cdots,s-1$.  Instead, below, we will take it to be an arbitrary integer subject to the identification
\begin{equation}\label{jiden}
    j\sim j+s\,.
\end{equation}

In addition to the symmetry index $s$, the properties of the anyons $\cal A$ depend on another integer $r$ capturing the 't Hooft anomaly of the $\mathbb{Z}_s^{(1)}$ one-form symmetry.  Therefore, we denote the set of these $s$ symmetry lines as ${\cal A}^{s,r}$. This set of Abelian anyons was denoted by $\mathbb{Z}^{(r/2)}_s$ in \cite{bonderson_non-Abelian_2007}. We will follow our notation above, 
\begin{align}
   &\ell = \gcd(r,s)\\
& r = p\ell \quad ,\qquad s = q\ell \,. 
\end{align}

In Appendix~\ref{app:generator}, we will consider a generalization of this.

We should distinguish between bosonic theories and fermionic/spin theories.
\begin{itemize}
\item In a bosonic theory, the spins and braidings of these lines are
\begin{equation}\label{spinbb}
\begin{split}
&h(a_j)={rj^2\over 2s}\mod 1\\
&B(a_j,a_{j'})={rjj'\over s}\mod 1\,,
\end{split}
\end{equation}
and we have the relation
\begin{equation}\label{a1s}
    a_1^s=1\,.
\end{equation}
To make these expressions meaningful with the identification \eqref{jiden}, we need to impose
\begin{equation}\label{rsbos}
rs\in 2\mathbb{Z}\,.
\end{equation}
This restriction is part of the description of $r$ as labeling the anomaly.
Clearly, $r$ and $r+2s$ lead to the same results, i.e.,
\begin{equation}\label{2sshift}
{\cal A}^{s,r}\cong {\cal A}^{s,r+2s}\,.
\end{equation}
\item In a fermionic/spin theory \cite{Dijkgraaf:1989pz,Gaiotto:2015zta,Bhardwaj:2016clt}, it is common to introduce a transparent fermion line $c$ with spin $1\over 2$ satisfying $c^2=1$. In this section, we will identify $c$ with the trivial (identity) anyon. Then, the spins $h$ of the lines are defined modulo $1\over 2$, corresponding to the ambiguity of attaching $c$ to each anyon. The spins and braidings of the symmetry lines are
\begin{equation}\label{spinbf}
\begin{split}
&h(a_j)={rj^2\over 2s}\mod {1\over 2}\\
&B(a_j,a_{j'})={rjj'\over s}\mod 1\,.
\end{split}
\end{equation}
Here, we do not have the restriction \eqref{rsbos}.  One way to say it is that we can replace \eqref{a1s} with 
\begin{equation}
    a_1^s=c^{rs}\,.
\end{equation}
In a fermionic/spin theory, \eqref{2sshift} is replaced with
\begin{equation}\label{1sshift}
{\cal A}^{s,r}\cong {\cal A}^{s,r+s}\,.
\end{equation}
\end{itemize}

In addition to the identifications \eqref{2sshift} and \eqref{1sshift}, for any $m$ such that $\gcd(m,s)=1$, we can take the line $a_m=a_1^m$ to be the $\mathbb{Z}_s^{(1)}$ generator.  This has the effect of relating
\begin{equation}\label{changemg}
{\cal A}^{s,r}\cong {\cal A}^{s,m^2r}\qquad , \qquad \gcd(m,s)=1\,.
\end{equation}
As expected, this expression is invariant under $m\to m+s$.

The symmetry lines of any spin TQFT can be written as a product of a set of symmetry lines of a bosonic theory and the almost trivial, spin TQFT, $\cA^{2,2}\cong \cA^{2,0}\cong\{1,c\}$. See Appendix~\ref{Necessary definitions}. In particular, every Abelian spin TQFT is a product of an Abelian bosonic TQFT and $\{1,c\}$.

Let us demonstrate it for the symmetry lines $\cA^{s,r}$.  If $rs$ is even, $\cA^{s,r}$ can be viewed as bosonic, and the general statement above is trivially true.  If $rs$ is odd, $\cA^{s,r}$ is not bosonic.  However, using \eqref{1sshift}, as spin theories, ${\cal A}^{s,r}\cong {\cal A}^{s,r+s}$.  Since for odd $rs$, ${\cal A}^{s,r+s}$  also makes sense as a bosonic TQFT, the spin theory $\cA^{s,r}$ factorizes as $ {\cal A}^{s,r+s}\boxtimes\{1,c\}$, where the first factor is bosonic.  Alternatively, using $m=2$ in \eqref{changemg}, as spin theories, $\cA^{s,r}\cong\cA^{s,4r}$ . Since $\cA^{s,4r}$ also makes sense as a bosonic TQFT, the spin theory $\cA^{s,r}$ factorizes as $ \cA^{s,4r}\boxtimes\{1,c\}$, where the first factor is bosonic.

\subsubsection{$\ell=\gcd(r,s)=1$}

Interestingly, when
$\ell=\gcd(r,s)=1$,
the set of lines ${\cal A}^{s,r}$ forms a complete TQFT.  Furthermore, when our original TQFT $\cal T$  includes ${\cal A}^{s,r}$ and some other lines, it factorizes as \cite{drinfeld_braided_2010, Hsin:2018vcg}
\begin{equation}\label{bosfact}
{\cal T}={\cal A}^{s,r}\boxtimes {\cal T}'\qquad, \qquad \ell=\gcd(r,s)=1\,,
\end{equation}
where all the lines in ${\cal T}'$ are invariant under the $\mathbb{Z}_s^{(1)}$ one-form symmetry.

\subsubsection{$\ell=\gcd(r,s)\ne 1$}

When $\ell=\gcd(s,r)\ne 1$,
the lines ${\cal A}^{s,r}$ do not form a complete TQFT.  In particular, the line $a_1^{s/\ell} = a_1^q$  braids trivially with all the lines in ${\cal A}^{s,r}$.  One consequence of that is that ${\cal A}^{s,r}$ does not realize faithfully the $\mathbb{Z}_s^{(1)}$ symmetry. But since the whole TQFT $\cal T$ does realize that symmetry faithfully, it should include lines that braid nontrivially with $a_1^{s/\ell}=a_1^q$.  As a result, unlike the case of $\ell=1$, the whole TQFT $\cal T$ does not factorize.

Let us then consider the closed subset of lines generated by $a_1^{s/\ell} = a_1^q$. It leads to ${\cal A}^{\ell,pq\ell}\subset {\cal A}^{s,r}$.  Since $h(a_1^{q})= pq/2$ then in either a fermionic/spin theory or bosonic theory with even $pq$ we have ${\cal A}^{\ell, pq\ell} \simeq {\cal A}^{\ell, 0}$ and it can be gauged. We will also be interested in gauging that symmetry in a bosonic theory with odd $pq$.  In that case, the resulting theory is fermionic/spin. 

We denote the result of the gauging as ${\cal T}/{\cal A}^{\ell,{0}}$.  The effect of this gauging on the symmetry lines is to form a quotient
\begin{equation}\label{AmodA}
    {{\cal A}^{s,r}\over {\cal A}^{\ell,0}}\cong {\cal A}^{{q},{p}} \quad, \quad {\rm bosonic}\quad pq\in 2\mathbb{Z}\qquad \text{or  fermionic/spin}\,.
\end{equation}
Since the two labels of ${\cal A}^{{q},{p}}$ are coprime this means that after gauging the full theory factorizes as
\begin{equation}\label{TmodA}
   { {\cal T}\over {\cal A}^{\ell,0}}={\cal A}^{{q},{p}} \boxtimes {\cal T}' \quad, \quad {\rm bosonic}\quad pq\in 2\mathbb{Z}\qquad \text{or fermionic/spin}\,.
\end{equation}
In a bosonic theory with odd $pq$, which means that $\ell$ is even, we can gauge only ${\cal A}^{{\ell\over 2},0}$ to find
\begin{equation}\label{AmodAi}
    {{\cal A}^{s,r}\over {\cal A}^{{\ell\over 2},0}}\cong {\cal A}^{{2q},{2p}} \quad,\quad  {\rm bosonic}\quad pq\in 2\mathbb{Z}+1\,
\end{equation}
and $\cal T$ does not factorize.

\subsection{Bosonic TQFT coupled to $\U$}\label{TQFTUb}

Next, we add to the discussion in \cite{Hsin:2018vcg} the $\U$ (zero-form) symmetry.  The formalism to do it was derived in  \cite{barkeshli_symmetry_2019} (see also \cite{etingof2010fusion}).

Let us start with a bosonic theory with symmetry lines
${\cal A}^{s,r}$, where, as in \eqref{rsbos}, we should impose
\begin{equation}\label{rsbosc}
sr\in 2\mathbb{Z}\,.
\end{equation}
Once we have this group of symmetry lines, we should assign $\U$ charges. The $\U$ charges and the coupling to the background gauge field $A$ depend on a class in $H^2(\U,\mathbb{Z}_s)=\mathbb{Z}_s$.  This class is labeled by an integer $u$ modulo $s$.\footnote{As in \eqref{changemg}, it is sometimes the case that different values of $u$ are related by an automorphism of the TQFT and lead to essentially the same coupling of $A$ to the TQFT. This will be addressed in Appendix~\ref{app:generator}.} Physically, the value of $u$ corresponds to the choice of vison. For simplicity, in the main text, we consider the group of symmetry lines to be generated by the vison, which effectively sets $u=1$.  (In Appendix~\ref{app:generator} we will restore generic values of $u$.)

We will then refer to the symmetry lines with $u=1$ as ${\cal V}^{s,r}$ to highlight that they are generated by the vison, $v$.

In the problem without the $\U$ symmetry in  Sec.~\ref{pureTQFT}, we discussed the behavior of the TQFT $\cal T$ depending on whether $\ell=\gcd(r,s)= 1$ or $\ell\ne 1$.  The situation here is similar.

For $\ell=\gcd(r,s)= 1$, the $\mathbb{Z}_s^{(1)}$ one-form symmetry acts faithfully on ${\cal V}^{s,r}$ and we can take all the other lines in the TQFT $\cal T$ to be invariant under it and therefore carry no $\U$ charge.  Then, as in \eqref{bosfact}, the theory factorizes
\begin{equation}\label{bosonfac}
{\cal T}={\cal V}^{q,p} \boxtimes {\cal T}'\qquad, \qquad \ell=1\,,
\end{equation}
and all the lines in ${\cal T}'$ are $\U$ neutral. [More precisely, we can attach to them the transparent charged boson to set their $\U$ charge to zero.]

When $\ell=\gcd(r,s)\ne 1$, the symmetry lines ${\cal V}^{s,r}$ do not realize the $\mathbb{Z}_s^{(1)}$ symmetry faithfully.  In particular, $a_1^{s/\ell} = v^{s/ \ell}=v^q$ does not act on the lines in ${\cal V}^{s,r}$, but it does act on the other lines in the TQFT $\cal T$.  Therefore, $\cal T$ does not factorize.

In the case of $\ell=\gcd(r,s)\ne 1$, we can follow the discussion around \eqref{AmodA}-\eqref{AmodAi}, and try to gauge the $\mathbb{Z}_\ell^{(1)} \subset\mathbb{Z}_s^{(1)}$ that is generated by $a_1^{s/ \ell} = v^q$. This line generates ${\cal A}^{\ell,{pq\ell}}$. These lines have integer charges and can thus be taken to be neutral. For even $pq$, ${\cal A}^{\ell,{pq\ell}}\cong{\cal A}^{\ell,0}$ and we can gauge it, i.e., it plays the role of ${\cal A}_0$ in Sec.~\ref{sec:symm_anom_gauge}. Here it is important that the lines in $ {\cal A}^{\ell,0} $ are $\U$ neutral. Then, as in the problem without $\U$,
\begin{equation}\label{AmodAUb}
    {{\cal V}^{s,r}\over {\cal A}^{\ell,0}}\cong {\cal V}^{{q},{p}} \qquad,\qquad pq\in 2\mathbb{Z}\,.
\end{equation}
And after the gauging, the full theory factorizes as
\begin{equation}
   { {\cal T}\over {\cal A}^{\ell,0}}={\cal V}^{{q},{p}} \boxtimes {\cal T}' \qquad, \qquad pq\in 2\mathbb{Z}\,. \label{eqn:gauge_product_boson}
\end{equation}
And again, the lines in $\cal T'$ are $\U$ neutral.

As in \eqref{AmodAi}, if $pq$ is odd, and hence $\ell$ is even, we can gauge only a subgroup\footnote{In this case, we can still gauge the full $\mathbb{Z}_\ell$ and end up with a fermionic/spin theory.  This theory might not satisfy the spin/charge relation.}
\begin{equation}\label{AmodAiQ}
    {{\cal V}^{s,r}\over {\cal A}^{{\ell\over 2},0}}\cong {\cal V}^{{2q},{2p}} \quad , \quad pq\in 2\mathbb{Z}+1\,
\end{equation}
and $\cal T$ does not factorize.

\subsection{TQFT coupled to spin$^c$ connection}\label{TQFTUf}

Now we consider coupling our TQFT with symmetry lines ${\cal A}^{s,r}$ to a spin$^c$ connection. Recall that, as reviewed in Sec.~\ref{sec:tqft-basics}, a spin$^c$ theory is a fermionic/spin TQFT satisfying the spin/charge relation. Namely, all local particles with odd charge are identified with the electron $c$. This relation is often easier to deal with mathematically if we keep the electron as an anyon in our theory, so we make that choice here.

In Appendix~\ref{app:spinc-classification}, we will review the classification of spin$^c$ TQFTs in a way that incorporates information about the charges of all of the anyons in the theory. There, we will see that these theories still possess a unique Abelian vison. The braiding of this vison with all other anyons continues to determine their charge modulo one. Moreover, the vison continues to have $(h(v), Q(v)) \sim (\sigma_H/2, \sigma_H)$ modulo the identifications discussed in Sec.~\ref{sec:tqft-basics}. Thus, if the vison has order $s=q\ell$ it generates  ${\cal V}^{s,r}$.

The classification of charge assignments of anyons is more subtle in the spin{}$^c$ case, and we defer a full discussion to Appendix~\ref{app:spinc-classification}. The first new subtlety is that the spin/charge relation introduces additional consistency conditions and therefore the vison cannot be an arbitrary Abelian anyon.  The second subtlety is that some superficially distinct choices are actually the same.

Let us illustrate these two subtleties with the example of the spin{}$^c$ $\cA^{2,1}$ theory. It can be described as $\U_2\times \U_{-1}\cong \U_{-2}\times \U_1$.  Denote the charge of $a_1$ by $Q(a_1)$, which is defined modulo 2. The fusion rule $a_1^2=1$ and the spin/charge relation lead to
\begin{equation}
    2Q(a_1)=0\mod 2\,.
\end{equation}
This demonstrates the first subtlety above that excludes $Q(a_1)={1\over 2}$.  To demonstrate the second issue, we compare the choices $Q(a_1)=0\mod 2$ and  $Q(a_1)=1\mod 2$. By changing the generator of $\cA^{2,1}$ from $a_1$ to $a_1c$, the anyon now has vanishing charge, but its spin is shifted by $1/2$ from $1/4$ to $-1/4$ mod 1. In other words, $\U_2\times \U_{-1}$ with $Q(a_1)=1\mod 2$ is equivalent to $\U_{-2}\times \U_{1}$ with $Q(a_1)=0\mod 2$.  

One way to see that is to realize the theory by the Chern-Simons Lagrangian:
\begin{equation}
    \frac{2}{4\pi}b\wedge db - \frac{t}{2\pi}A\wedge db-{1\over 4\pi} \hat b \wedge \hat b -{1\over 2\pi} A\wedge d\hat b\,,
\end{equation}
where $t$ is an even integer\footnote{Clearly, this Lagrangian is consistent with the spin/charge relation \eqref{tplusK}. The dependence on $t\mod 4$ can be shown by shifting $b\to b+2A$, which shifts $t$ by $4$ and adds a counterterm.  Note that $b\to b+2A$ is compatible with the fact that $A$ is a spin$^c$ connection.}   and $Q(a_1)=t/2\mod 2$. For $t=0$, $b$ decouples from $A$.  And for $t=2$, we can use the redefinition
\begin{equation}
\begin{split}
    & b\to-b+\hat b\\
    &\hat b \to 2b-\hat b
\end{split}  
\end{equation} 
to turn the Lagrangian into 
\begin{equation}\label{shiftb}
    -\frac{2}{4\pi}b\wedge db +{1\over 4\pi} \hat b \wedge \hat b -{1\over 2\pi} A\wedge d\hat b \,,
\end{equation}
which again shows that $b$ decouples from $A$.

Having said all of this, we nonetheless can understand a lot about these theories by studying ${\cal V}^{s,r}$. In this case, the spin and charge of $a_1^s$ are compatible with the spin/charge relation, i.e., they are an integer multiple of $\left({1\over 2}, 1\right)$, only when 
\begin{equation}\label{rsrelspinc}
    r(s+1)\in 2\mathbb{Z}\,.
\end{equation}
This equation replaces the condition \eqref{rsbosc} in the bosonic theory. 

Consider first the case of $\ell=\gcd(r,s)= 1$. Then, as in \eqref{bosonfac}, the theory factorizes
\begin{equation}\label{spinc_fac}
{\cal T}={\cal V}^{q,p} \boxtimes {\cal T}'\qquad, \qquad \ell=1\,.
\end{equation}
As in the bosonic case, ${\cal T'}$ is invariant under the $\bb{Z}_s^{(1)}$ one-form symmetry generated by the vison.  Therefore, all the anyons in $\cT'$ carry integer charges and by multiplying them with the transparent charged fermion, we can set all their $\U$ charges to zero (see Corollary \ref{cor:MTC-zero-charge}).

As in the bosonic case, when $\ell=\gcd(r,s)\ne 1$, the symmetry lines ${\cal V}^{s,r}$ do not realize the $\mathbb{Z}_s^{(1)}$ symmetry faithfully. Therefore, $\cal T$ does not factorize. We can then follow the discussion around \eqref{AmodA}, \eqref{TmodA}, and \eqref{AmodAi}, and try to gauge the $\mathbb{Z}_\ell^{(1)} \subset\mathbb{Z}_s^{(1)}$ that is generated by $a_1^{s/ \ell} = v^q$.
Here we have to be careful that after gauging, we still have a spin$^c$ theory. This is the case when the line $v^q$ is transparent, which leads to $p(q+1)\in 2\mathbb{Z}$.  Since $p$ and $q$ are co-prime, this happens when $q$ is odd. Then, we can write this set as ${\cal A}^{\ell,0} $ and we can gauge it to find
\begin{equation}\label{AmodAUbx}
    {{\cal V}^{s,r}\over {\cal A}^{\ell,0}}\cong {\cal V}^{{q},{p}}  \qquad, \qquad q\in 2\mathbb{Z}+1\,.
\end{equation}
And after the gauging, the full theory factorizes as
\begin{equation}
   { {\cal T}\over {\cal A}^{\ell,0}}={\cal V}^{{q},{p}} \boxtimes {\cal T}' \qquad, \qquad q\in 2\mathbb{Z}+1\,. \label{eqn:gauge_product_bosonx}
\end{equation}
Just as in the case of $\ell = 1$, we can take all the anyons in ${\cal T}'$ to have vanishing $\U$ charge.

When $p(q+1)$ is odd and therefore $q$ is even and $p$ is odd, we cannot gauge the $\mathbb{Z}_\ell^{(1)}$ one-form symmetry to find a spin$^c$ theory.  Because of \eqref{rsrelspinc}, we should also have even $\ell$.  Then, as in \eqref{AmodAi} or  \eqref{AmodAiQ}, we can gauge only $ {\cal A}^{{\ell\over 2},0}$
\begin{equation}\label{AmodAUbxs}
    {{\cal V}^{s,r}\over {\cal A}^{{\ell\over 2},0}}\cong {\cal V}^{{2q},{2p}}  \qquad, \qquad q\in 2\mathbb{Z}\,.
\end{equation}
and the theory does not factorize.

\section{Ordering}
\label{sec:IR}

We would now like to use the relations derived in Sec.~\ref{sec:symmetrylines} to learn about the long-distance TQFTs. The presence of $\mathcal{V}^{s,r}$ will provide a tight constraint. These constraints can essentially be understood to arise from  \eqref{AmodAUb} - \eqref{AmodAiQ}, and \eqref{AmodAUbx} - \eqref{AmodAUbxs}. In this section, we will provide a picture for thinking about these equations as a sort of ``descent'' to the minimal theory. We will also demonstrate how this perspective is useful in a number of ways: it provides a natural ordering to FQH states, it makes constructing new FQH states easier, and it allows the quantum Hall numericist to use a few pieces of numerically accessible information to produce a finite list of consistent TQFTs. A nice corollary of this latter perspective is that it allows us to prove a few historical conjectures about the ``smallest'' topological order in a number of cases: the spin{}$^c$ TQFTs with the smallest number of anyons for $\sigma_H=1/2$ are generalizations of Moore-Read, the T-Pfaffian is the smallest surface TO of a class AII topological insulator, or  $\nu=1$ class AIII topological superconductor,  and $D(\bb{Z}_4)$ is the smallest TO for a $1/2+1/2$ quantum Hall bilayer. Our results regarding minimal theories have already proven useful in numerical studies of the FQHE at $\nu = 3/4$ \cite{huang_non-Abelian_2024}. 

For the benefit of the reader, we define and explain our perspective by first working through a few familiar examples. Since spin{}$^c$ theories are most common in condensed matter physics, we begin with a few familiar examples of single-layer spin{}$^c$ theories.

\subsection{Quantum dimension}

In preparation for the discussion below, we review the notions of the quantum dimension of an anyon $d_a$ and the total quantum dimension 
\begin{equation}
    {\cal D}=\sqrt{\sum_a d_a^2}\,.
\end{equation}

One way to define $d_a$ is the following. An important characteristic of a TQFT is its modular $S$-matrix.  (Hence the M in MTC.)  In terms of it, the three-sphere partition function with a closed loop of the anyon $a$ is \cite{Witten:1988hf}
\begin{equation}
    {\cal Z}(S^3, W_a) =S_{a1}\,,
\end{equation}
with the special case with $a=1$
\begin{equation}\label{ZS3D}
    {\cal Z}(S^3)=S_{11}\,.
\end{equation}
Then, we define the quantum dimension of $a$ as
\begin{equation}
    d_a=\langle W_a\rangle ={{\cal Z}(S^3, W_a)\over {\cal Z}(S^3)}={S_{a1}\over S_{11}}\,.
\end{equation}
Using the facts that $S_{a1}=S_{1a}$ are real, and the matrix $S$ is unitary, $\sum_aS_{a1}^2=1$, and therefore, the total quantum dimension is
\begin{equation}\label{totalquantd}
    {\cal D}=\sqrt{\sum_a d_a^2}={1\over S_{11}}\,.
\end{equation}

Below, we will use the fact that if $d_a=1$, the anyon $a$ is Abelian, and otherwise it is non-Abelian. 

This definition of $d_a$ does not explain in what sense it can be viewed as a dimension.  Here we give three explanations. 

First, consider the dimension of the Hilbert space of the TQFT on a sphere in the presence of $n$ identical defects of type $a$. One can show that as $n\rightarrow \infty$, the dimension grows asymptotically as $\sim d_a^n$. 

Alternatively, the original motivation for this definition came from the underlying RCFT, where it is related to the dimension of a representation $a$ of a chiral algebra.  Since these representations are infinite dimensional, they do not have a standard notion of dimension.  However, we can use the characters of these representations $\chi_a(\tau)$ (with $\tau$ the modular parameter), to define the quantum dimension of $a$ as \cite{Dijkgraaf:1988tf}
\begin{equation}
    d_a=\lim_{\tau\to 0}{\chi_a(\tau)\over \chi_1(\tau)}\,.
\end{equation}
Then, using the action of the modular group $S:\ \chi_a\to \sum_b S_{ab} \chi_b$ and the fact that the lowest dimension operator is the identity, 
\begin{equation}
    d_a=\lim_{\tau\to 0}{\chi_a(\tau)\over \chi_1(\tau)}=\lim_{\tau\to 0}{\sum_bS_{ab}\chi_b(-{1\over \tau})\over \sum_cS_{1c}\chi_c(-{1\over \tau})}={S_{a1}\over S_{11}}\,.
\end{equation}
A relation between these RCFT representations and the 2+1d TQFT is obtained by studying the TQFT on a disk with an anyon $a$.  The edge states are in the representation of the chiral algebra associated with $a$ \cite{Witten:1988hf,Elitzur:1989nr} and $\chi_a(\tau\rightarrow 0)$ is the trace of this Hilbert space. 

Finally, the Verlinde formula \cite{Verlinde:1988sn,Moore:1988uz} shows that the quantum dimensions satisfy $d_ad_b=\sum_c N_{ab}^c d_c$, where the non-negative integers $N_{ab}^c$ are the fusion coefficients.  The fact that this expression is similar to the analogous relation for multiplying ordinary finite characters gives another motivation to interpret them as dimensions.

\subsection{Single-layer spin{}$^c$ theories}

\subsubsection{Odd $q$, $\ell=1$}

Most historical examples of odd denominator FQH states have the order of the vison equal to $q$, i.e., $\ell = 1$. Thus, in this section, we will take $\ell = 1$, but generalize to larger values in the next section. Here, $s=q$ and $r=p$ and hence, the symmetry lines are ${\cal V}^{q,p}$.  As mentioned in Sec.~\ref{sec:symmetrylines}, this is a complete TQFT and the full TQFT thus factorizes as ${\cal V}^{q,p}\boxtimes {\cal T'}$, where anyons in ${\cal T'}$ carry integer charges. In this case, one can attach electrons to the $\cT'$ anyons to make all of them charge-neutral, so the theory can be formed with only bosonic degrees of freedom, i.e., a purely bosonic TO. This is formalized in Theorem~\ref{thm:triv_vison_split} in Appendix~\ref{app:spinc-classification}. 

Let us first discuss a few examples with trivial $\cal T'$. These examples were the first FQH states to be described. Starting with Laughlin \cite{Laughlin1}, wavefunction-based approaches to FQH states were developed \cite{HaldaneHierarchy, HalperinHierarchy, ReadHierarchy1990, JainCF, Moore:1991ks}. 
In the Laughlin states at filling $\nu = 1/q$, the TQFT is just our $\mathcal{V}^{q,1}$ theory. In Chern-Simons language, these are the $\U_q$ theories. 

Now we will show how more general Abelian FQH states fit in. First, hierarchical FQH states can be represented by a $K$ matrix that is tridiagonal with ones above and below the diagonal \cite{wen1995topological, zhang2025hierarchy} and a $t$ vector given by $t_{ij} = \delta_{i1}$. To satisfy the spin{}$^c$ condition, we need to have $K_{11}\in 2\bZ+1$ and $K_{jj}\in 2\bZ$ for $j>1$. These theories are just the $\mathcal{V}^{q,p}$ minimal theories with $q=|\det K|$ and $p = q K^{-1}_{1,1}$.  They are generated by the vison. 

Another set of examples are the Jain sequence states with $n$-dimensional $K$ matrix and $t$ vector \cite{wen1995topological}
\begin{equation}
    K_{ij} = \pm\delta_{ij}+ 2p\qquad , \qquad \qquad t_i = 1\qquad , \qquad \qquad 1 \leq i,j \leq n\,. 
\end{equation}
These states have filling $\nu =n/(2np\pm 1)$, and a $\mathbb{Z}_{2np \pm 1}^{(1)}$ one-form symmetry group. As before, the vison generates the entire symmetry group, and so the Jain states are really just the $\mathcal{V}^{2np\pm 1,n}$ minimal theories. In fact, all the states that arise from additional flux attachments to these Jain states are the $\mathcal{V}^{q,p}$ minimal models, as shown in Appendix~A of \cite{JensenRaz}. 

In general, every ${\cal V}^{q,p}$ can be realized by an Abelian Chern-Simons theory \cite{ Wang:2020nmz, ma_fractonic_2022}.   

Some examples with nontrivial $\cal T'$ are the Read-Rezayi states at filling $\nu = p/q$ with $p$ and $q$ odd. Let us review this fact. The Read-Rezayi states have filling fractions $\nu = k/(km+2)$ for $k, m\in \bb{Z}$ \cite{read_beyond_1999, bonderson_non-Abelian_2007}. For odd $m$, these are spin{}$^c$ theories, while for even $m$ they are bosonic. In all cases, they are constructed as the following Chern-Simons theory \cite{bonderson_non-Abelian_2007, seiberg2016gapped} 
\begin{equation}
    R_{m,k} =  \frac{{\rm U(2)}_{k,-2k} \times  \U_{k(km+2)}}{\mathbb{Z}_{k}}\,.
\end{equation}
Here, we used the notation for Chern-Simons theories ${\rm U}(2)_{k_2,k_1}=[{\rm SU}(2)_{k_2} \times \U_{k_1}]/\mathbb{Z}_2$.\footnote{The notation in \cite{bonderson_non-Abelian_2007} is slightly different. There, the  $\bb{Z}_k$ para-fermion theory ${\rm U}(2)_{k,-2k}$ is referred to as $\mathrm{Pf}_k$. The anyon content of $\mathrm{Pf}_k$ is written explicitly in Eq.~(5.50) along with the identifications.  Also, \cite{bonderson_non-Abelian_2007} uses $\bb{Z}_{2k}^{(-1/2)}$ in place of $\U_{-2k}$. Similarly, the anyon content of $R_{m,k}$ is written explicitly with the identifications in Eq.~(5.51), and it is identical to the $R_{m,k}$ written above. Note that here $\bb{Z}_{2k(km+2)}^{(1/2)}$ is used in place of $\U_{k(km+2)}$ for $k,m$ odd.}

When $p$ and $q$ are odd (and hence $k,m$ are odd), one can factorize 
 both ${\rm U}(2)_{k,-2k}$ and $\U_{k(km+2)}$, resulting in
\begin{equation}
    R_{m,k} = {\cal F}_k \boxtimes \mathcal{V}^{km+2,k}\,. 
\end{equation}
Here, the theory ${\cal F}_k$ is defined in Appendix~\ref{app:Fk}.
Thus, these states are all minimal states stacked with an additional theory. For $\nu = 13/5=2+3/5$ this additional theory is ${\cal F}_3$, the Fibonacci anyon model. The case of $\nu=12/5 = 2+2/5$ is the particle-hole conjugate of $13/5$. These states have been proposed to be realized at these fillings \cite{read_beyond_1999,sreejith2013tripartite, FibonacciDMRG1, FibonacciDMRG2, Pakrouski2016}, though there is some debate over their applicability. 

\subsubsection{Even $q$, $\ell=2$}

For even $q$ there is an important difference from odd $q$. Here, the spin/charge relation forces $\ell$ to be even. Historical examples again have the smallest allowable value, $\ell =2$.
When $\ell =2$, $s=2q$ and $r=2p$. Hence, the symmetry lines are ${\cal V}^{2q,2p}$.  This is not a complete TQFT, and the full low-energy TQFT does not factorize into this factor and the rest of the lines. 

Nonetheless, we can still learn through a few examples. We take the simplest case of $p=1,q=2$. Then one example of $\ell = 2$ is $\U_8$. This theory has a $\mathbb{Z}_8^{(1)}$ one-form symmetry, and we write it as ${\cal A}^{8,1}$.  The line $a_1^2$ of ${\cal A}^{8,1}$ is identified with the vison. It generates $\mathcal{V}^{4,2}$ and can thus be matched with the UV theory with $\nu=1/2$.  Another example is the Moore-Read state~\cite{Moore:1991ks}, whose Abelian sector is generated by the vison and also given by ${\cal V}^{4,2}$. 

The astute observer will note that both of these examples are special cases of the Pfaffian theories. They are defined as 
\begin{equation}
\mathsf{Pf}^f_{2,1,n}=\frac{\mathrm{Spin}(n)_1\boxtimes \mathcal{A}^{8,1}}{(\psi a_1^4c)}=\frac{\mathrm{Spin}(n)_1\times \U_8}{\mathbb{Z}_2}\qquad,\qquad n=0,1,\ldots, 7\label{eqn:fermion_mTOs},
\end{equation}
where the subscript indicates that we are gauging the $\bb{Z}_2^{(1)}$ one-form symmetry generated by the even-charge boson $\psi a_1^4c = \psi v^2c$. 
We chose the name $\mathsf{Pf}^f_{q,p,n}$ because for $p=1$ and $n$ odd, these are the (fermionic) Pfaffian states at $1/q$ filling. We discuss these theories in more detail in Appendix~\ref{app:anyontheory}. The $\U_8$ state is $\mathsf{Pf}^f_{2,1,0}$, while the Moore-Read state is $\mathsf{Pf}^f_{2,1,1}$. Indeed, \textit{all} of these theories contain the required ${\cal V}^{4,2}$ subgroup. They are thus all consistent with the required symmetries and anomalies. Four theories are Abelian with eight anyons, while the other four theories are non-Abelian. Moreover, they all have the same quantum dimension $\mathcal{D}=2\sqrt{2}$.

\subsection{Defining a descent}

The observation that many of the important FQH states seem to be as ``small'' as possible motivates us to define an ordering to the allowable FQH states. First, our observation in Sec.~\ref{sec:symmetrylines} that theories with $\ell =1$ factorize suggests the following rule: 
\begin{itemize}
    \item[1)] Suppose $\mathcal{T}$ factorizes as $\mathcal{T} = \mathcal{S}\boxtimes \cT'$ where $\cT'$ is nontrivial but does not couple to the U(1) background gauge field. Then we call $\mathcal{S}$ more minimal than $\cT$. Alternatively, we may say that ${\cal S}$ descends from $\mathcal{T}$, or that $\mathcal{T}$ is a parent TQFT for $\cal S$.
\end{itemize}

However, if $\ell \neq 1$ then the TQFT does not factorize. Nonetheless, we were able to produce a smaller TQFT by gauging a subgroup of the one-form symmetry. This motivates a second rule:
\begin{itemize}
    \item[2)] Let $\mathcal{T}$ have an anomaly free one-form symmetry, whose gauging leads to $\mathcal{S}$. If anyons generating the one-form symmetry are charge-neutral, the gauging does not affect the $\U$ symmetry, and $\mathcal{S}$ has the same Hall conductivity as $\mathcal{T}$. We then call $\mathcal{S}$ more minimal than $\mathcal{T}$.  Again, we may use parent/child language to emphasize the relationship between $\mathcal{S}$ and $\mathcal{T}$.
\end{itemize}

This allows us to define a kind of ``descent'' within the set of TQFTs with the same $\sigma_H$. Start with some TQFT $\mathcal{T}$ and first apply rule two, and then rule one. Note that this descent always decreases the quantum dimension \cite{Bais:2008ni, Eliens:2013epa} (see also Appendix~\ref{app:gauging-quantum-dim}). If ${\cal T}$ is Abelian, then it will also decrease the anyon count; however, if ${\cal T}$ is non-Abelian, then step two may increase the anyon count.\footnote{As an example of this, consider the bosonic theory $\cA^{n,k}$ with $nk\in 2\bZ$. Gauging the charge conjugation zero-form symmetry leads to a new theory $\cal T$ with $(n+7)/2$ anyons when $n$ is odd, and $n/2+7$ when $n$ is even. Now, gauging the (dual) one-form $\bZ_2^{(1)}$ symmetry in this new theory $\cal T$ increases the anyon number if $n$ is sufficiently large.}

Let us see what this does in the case of odd denominators. Applying rule two will take all theories with $\ell \neq 1$ to $\ell =1$, by gauging the one-form symmetry generated by $v^q$ (with suitable electrons attached). The resulting theory is still spin{}$^c$, with the vison $v$ generating ${\cal V}^{q,p}$. Further, by Sec.~\ref{sec:symmetrylines}, it factorizes as ${\cal V}^{q,p}\boxtimes {\cal T'}$. Applying rule one leaves ${\cal V}^{q,p}$ remaining. For single-layer spin{}$^c$ theories with odd denominators, this suggests we should call $\mathcal{V}^{q,p}$ \textit{the} minimal TQFT; it is the unique endpoint of our procedure to produce more minimal theories. We display this in Fig.~\ref{fig:odd}. Note that in this case, while the minimal TQFT is unique, there are infinitely many gapped phases corresponding to the same TQFT by stacking invertible phases (see  Sec.~\ref{sec:tqft-gappedphase}).

Next, consider the case of even denominators. Applying rule two, we gauge the one-form symmetry generated by $v^{2q}$. Since its charge is $0\mod 2$ the gauging does not affect the charge conservation, and the resulting theory is still spin{}$^c$ with the same Hall conductivity. This takes all theories to $\ell = 2$. However, unlike the case of odd $q$, now there is no longer a unique endpoint theory. For concreteness, specialize to $\sigma_H = 1/2$. One collection of endpoints are  $\mathsf{Pf}^f_{2,1,n}$ of \eqref{eqn:fermion_mTOs}. These theories have $\ell =2$, and it can be checked that there is no uncharged theory that factors out. There are also an infinite number of other endpoints, though as we will see later, all have larger numbers of anyons and quantum dimension. Example of these other endpoints are $[\mathrm{SU}(2)_{4k+2}\boxtimes \cA^{8,1}]/{(\Phi_{2k+1}a_1^4c)}$ where $k\in \bZ$.\footnote{Note that for $k=2$ [i.e., ${\rm U}(2)_{10,8}$], the $(\Phi_3,1)$ anyon is a charge-neutral non-Abelian boson that can be gauged. Gauging it ``descends" the theory to $\pf_{2,1,5}^f$.  This gauging is similar to the gauging that turns ${\rm SU}(2)_{10}$ to ${\rm Spin} (5)_1$, and it can be interpreted as gauging a non-invertible symmetry. (See footnote \ref{gaugingvscon}.) Thus, if we allow this more general type of gauging, then it is no longer an ``endpoint". However, even if this type of gauging is allowed, the other ${\rm U}(2)_{4k+2,8}$ theories with $k\neq 2$ remain endpoints.} 
Here $\Phi_j$ refers to the isospin-$j$ anyon in ${\rm SU}(2)_{4k+2}$ \{alternatively they can be represented as ${\rm U}(2)_{4k+2,8}$ Chern-Simons theories, see Appendix~C4 of \cite{seiberg2016gapped}\}. For a general $p,q$ with $q$ even, these theories are given by 
\begin{equation}\label{eqn:U_defn}
{\cal U}_{q,p,4k+2} = \frac{{\rm SU}(2)_{4k+2} \boxtimes {\cal A}^{4q,p}}{(\Phi_{2k+1}a_1^{2q}c)}\,,
\end{equation}
where $v = a_1^2 \in {\cal A}^{4q,p}$ and the boson $\Phi_{2k+1}a_1^{2q}c$ generates an order two subgroup. Note that $k=0$ is equivalent to $\pf_{q,p,3}$, but other values of $k$ are distinct.

\begin{figure*}
    \centering
    \subfigure[Descent for $q$ odd]{
        \includegraphics[width=0.45\textwidth]{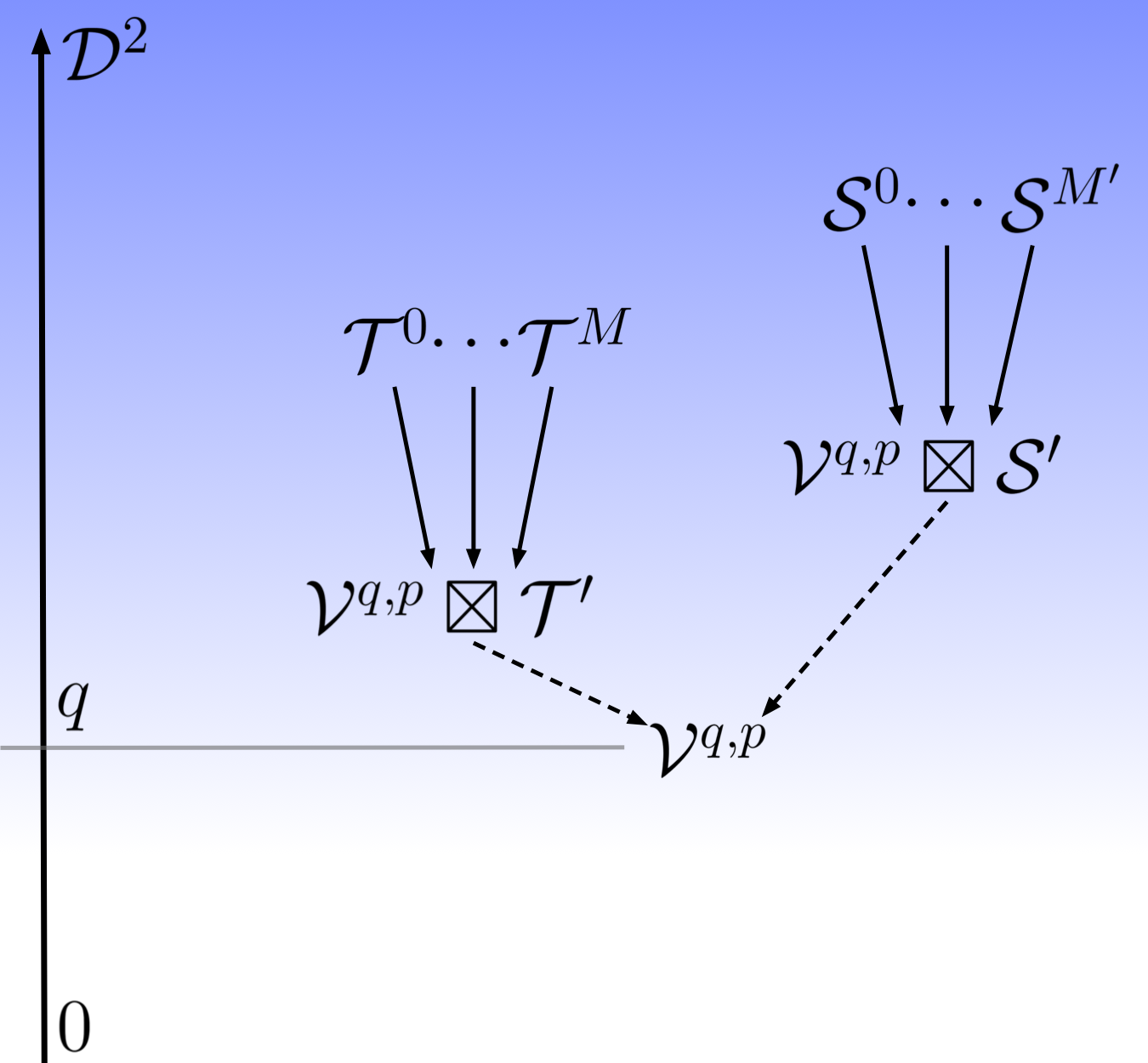}
        \label{fig:odd}
            }
    ~ 
    \subfigure[Descent for $q$ even]{
        \includegraphics[width=0.45\textwidth]{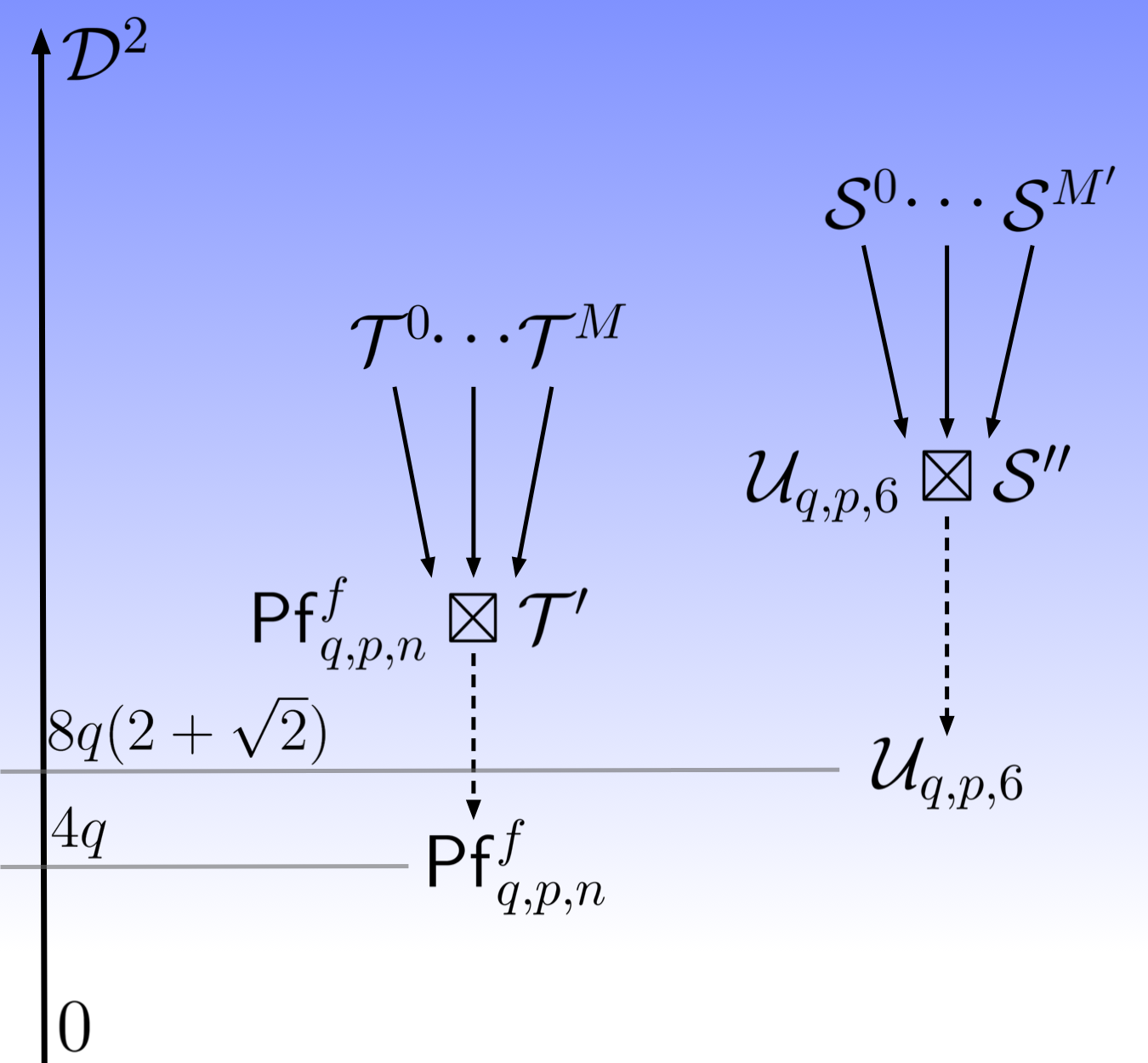}\
        \label{fig:even*}
            }
    \caption{Illustration of the descent of single-layer spin{}$^c$ theories with $\sigma_H = p/q$. Solid arrows indicate descent under gauging, while dashed arrows indicate removing neutral, decoupled theories. The vertical axes display theories with increasing total quantum dimension, $\mathcal{D}$. The larger $\mathcal{D}$, the larger the number of theories, illustrated by the darker color at larger values.  In (a), the lower bound on $\mathcal{D}^2$ is $q$, illustrated by the gray line. It is seen to be saturated by $\mathcal{V}^{q,p}$, which is the unique endpoint of the descent discussed in the main text. In (b), we see that for even $q$ the descent no longer has a single endpoint. Some, but not all, are $\mathsf{Pf}^f_{q,p,n}$ and ${\cal U}_{q,p,6}$, discussed in the main text.}
	\label{fig:flow}
\end{figure*}

\subsection{Ordering TQFTs using topological data beyond $\sigma_H$}

In the previous section, we defined a descent procedure and saw how it can be used to order TQFTs consistent with a given $\sigma_H$. As discussed in the introduction, this may be the only topological data accessible to experimentalists. However, numerical studies of FQH systems often have access to additional topological data. In this section, we will show how to use a few pieces of topological data beyond $\sigma_H$, along with the descent of the previous section, to further constrain the possible TQFTs. We will specialize to odd $q$, though the argument can be extended to even $q$. The information we will use is: 
\begin{itemize}
    \item the Hall conductivity (or filling fraction), $\sigma_H=p/q$;
    \item the number of anyons, $N_{{\cal T}}$. This can be determined by the torus ground-state degeneracy, measured using exact diagonalization (ED) in small system sizes;\footnote{In fermionic/spin theories, if the electron is counted as a nontrivial anyon, then the number of anyons is twice the torus ground-state degeneracy. We review this in more detail in Appendix~\ref{app:spinc-classification}.}
    \item the order of the vison, $s=q\ell$. This can be determined by studying how the degenerate ground states on the torus are permuted by threading flux. If it takes a flux of $2\pi s$ to return the system to itself, then $s = q\ell$ is the order of the vison;
    \item the total quantum dimension, $\mathcal{D}_{{\cal T}}$. It can be determined by studying the universal correction to the entanglement entropy, found via DMRG.
\end{itemize}
All of these can, in principle, be obtained in quantum Hall numerics. We note, however, that this may not be easy. Finite size effects often make determining $N_{{\cal T}}$ difficult in practice and the cylinder size required by DMRG to obtain ${\cal D}_{{\cal T}}$ can be prohibitively large. Nonetheless this information is sometimes possible to obtain, as in the numerical study of the FQHE at $\nu = 3/4$ \cite{huang_non-Abelian_2024} where our reasoning proved useful.

We want to use this information to construct a list of consistent TQFTs, ${\cal T}$, that are as minimal as possible. As in the previous section, we begin by gauging the one-form symmetry generated by $v^q$. This descends to $\mathcal{S} = \mathcal{V}^{q,p} \boxtimes \mathcal{T}'$, where ${\cal T}'$ is a bosonic TQFT that is neutral under the global $\U$. It has quantum dimension ${\cal D}_{\mathcal{T}'} = \mathcal{D}_{{\cal T}}/(\ell \sqrt{q})$, where $1/\sqrt{q}$ arises from removing ${\cal V}^{q,p}$ and $1/\ell$ from gauging the one-form symmetry (see Appendix~\ref{app:gauging-quantum-dim} for a discussion of this result from various perspectives).  We now enumerate the finitely\footnote{Note that  \cite{Bruillard:2013hdm} offers a proof that the number of TQFTs with a fixed number of anyons, $N_{{\cal T}}$, is finite. The proof can easily be extended to quantum dimension. Note that $N_{{\cal T}}\leq \mathcal{D}^2_{{\cal T}}$ in a TQFT. For a fixed $\mathcal{D}$ one simply enumerates all TQFTs, $\mathcal{T}$, with anyon number $N_{\mathcal{T}}\leq \mathcal{D}^2$ and keeps those with $\mathcal{D}_{\mathcal{T}} = \mathcal{D}$. There are finitely many of them by  \cite{Bruillard:2013hdm} and the inequality $N_{{\cal T}}\leq \mathcal{D}^2_{{\cal T}}$ ensures we found all possible TQFTs of a given quantum dimension.} many bosonic TQFTs ${\cal T}'$ with this quantum dimension. Since $\mathcal{D}_{\mathcal{T}'} < {\cal D}_{{\cal T}}$, this list of possibilities is smaller than enumerating all theories with ${\cal D}_{{\cal T}}$.

Having produced a finite list of descendant orders $\mathcal{S} = \mathcal{V}^{q,p} \boxtimes \mathcal{T}'$, we now want to ascend to obtain a list of parent orders. We do this by gauging the dual $\bb{Z}_\ell^{(0)}$ zero-form symmetry in ${\cal S}$ that was produced by gauging the one-form symmetry generated by $v^q$. Finding the possible parent orders then amounts to classifying $\bZ_{\ell}^{(0)}$ symmetry-enriched topological order $\mathcal{S}$. Now, the $\bZ_\ell^{(0)}$ symmetry action on $\mathcal{V}^{q,p}$ is determined by the fact that after gauging it, $\mathcal{V}^{q,p}$ is extended in the parent theory to $\mathcal{V}^{q\ell,p\ell}$ .  Therefore, in the parent theory, the generator $b$ of the global symmetry $\bZ_\ell^{(1)}$ satisfies $v^q=b$. (Then, gauging the $\bZ_\ell^{(1)}$ symmetry generated by $b$, we find back $\mathcal{V}^{q,p}$ in the child theory.), So, we only have to determine the $\bb{Z}_\ell^{(0)}$ symmetry action on $\mathcal{T}'$. Once this is done, we can produce the parent theories by gauging the $\bb{Z}_\ell^{(0)}$ symmetry in ${\cal S}$.

We have thus used $\sigma_H, \ell,$ and $\mathcal{D}_{{\cal T}}$ to produce a small list of candidate TQFTs. We did this by first descending and enumerating the possible neutral bosonic TQFTs ${\cal T}'$ with ${\cal D}_{{\cal T}'} = \mathcal{D}_{{\cal T}}/(\ell \sqrt{q})$ and the $\bb{Z}_\ell^{(0)}$ symmetry action on them. Then we gauged  this $\bZ_\ell^{(0)}$ zero-form symmetry to ascend to the list of candidate TQFTs. At this point, we can rule out all candidates that do not have the right number of anyons $N_{{\cal T}}$.

Let us give an example to illustrate how this works. Suppose that ${\cal D}_{{\cal T}} = \ell \sqrt{q}$. Then enumerating the orders with ${\cal D}_{{\cal T}'} = {\cal D}_{{\cal T}}/(\ell \sqrt{q}) = 1$ is simple; $\cal T'$ is trivial. Thus, there is only one possible child order, ${\cal S} = {\cal V}^{q,p}$. Our next step is to classify $\bb{Z}_\ell^{(0)}$ symmetry-enriched topological orders ${\cal S}$. The problem of classifying symmetry enrichment in a given TQFT has been systematically studied in recent literature. For bosonic TQFTs, a complete classification has been established in \cite{barkeshli_symmetry_2019}, and generalizations to fermionic/spin TQFTs were studied in \cite{Aasen:2021vva, Bulmash:2021hmb, Bulmash:2021ryq}, and these results can be generalized to spin$^c$. For our purposes, they reveal that if ${\cal W}^{q,p,\ell}$ is one possible parent TQFT\footnote{See Appendix~\ref{app:mme} for an explicit construction. When $p=1$, ${\cal W}^{q,p,\ell}$ can be taken to be ${\cal V}^{q\ell^2,1}$. \label{footnoteM}} then all others are given by
\begin{equation} \label{eq:mme_all}
{\cal T} = \frac{{\cal W}^{q,p,\ell} \boxtimes {\cal U}}{(v^{-q}b)}\,.
\end{equation}
Here, the theory ${\cal U}$  is the result of gauging the $\bZ_l^{(0)}$ symmetry of a fermionic SPT protected by $\bb{Z}_\ell^{(0)} \times \bb{Z}_2^{(0),f}$, and the boson $b$ is the generator of $\bZ_\ell^{(1)}$  in $\cal U$. These SPTs were classified in \cite{cheng_classification_2018}.  

For concreteness, suppose that $\ell$ is odd. Then the classification of fermionic SPT phases is identical to that of bosonic SPT phases with $\bZ_\ell^{(0)}$ symmetry. They are labeled by elements of $H^3(\bZ_\ell, \U)=\bZ_\ell$. Gauging $\bZ_\ell^{(0)}$ results in a $\bZ_\ell$ Dijkgraaf-Witten gauge theory \cite{Dijkgraaf:1989pz}, represented as $\U\times \U$ Chern-Simons theory with the following $K$ matrix \cite{Kapustin:2014gua}: 
\begin{equation}
    K_n = \begin{pmatrix} 0& \ell\\ \ell& 2n\end{pmatrix}\qquad, \qquad n=0,1,\ldots,\ell-1.
\end{equation}
Then, since $\ell$ is odd and this is a spin theory, there are precisely $\ell$ possible TQFTs consistent with ${\cal D}_{{\cal T}} = \ell \sqrt{q}$. Moreover, all of these theories are Abelian with a total of $N_{{\cal T}} = q\ell^2$ anyons.

We conclude by noting that in a few cases with a small anyon count, it is not necessary to go through the argument above. Instead, a list of candidate theories can be produced from knowledge of $\sigma_H$ and $N_{{\cal T}}$ alone. In the simplest case, where $N_{{\cal T}}=q$, this is enough to uniquely fix the state. As we have seen, only ${\cal V}^{q,p}$ has a total of $q$ anyons for $\sigma_H = p/q$ with $q$ odd.

Now suppose that $q<N_{{\cal T}}\leq 2q$. In this case we must still have $\ell = 1$, as $\mathcal{V}^{2q, 2p}$ is not modular. Thus we conclude that $\mathcal{T} = \mathcal{V}^{q,p}\boxtimes \mathcal{T}'$. The number of anyons contained in $\mathcal{T}'$, $N_{\mathcal{T}'}$, is given by $N_{\mathcal{T}'} = N_{{\cal T}}/q$. Since both $N_{{\cal T}}$ and $N_{\mathcal{T}'}$ must be integer we have $N_{{\cal T}} = 2q$, and $N_{\mathcal{T}'}=2$. The possibilities for $\mathcal{T}'$ are then $\U_{\pm 2}$, the Fibonacci theory, or its image under time reversal~\cite{Rowell:2007dge}.

\subsection{Minimality by anyon count/quantum dimension}
As we emphasized, the observation of a fractional $\sigma_H$ implies the existence of a set of Abelian anyons. We now pose the question of what the ``minimal" topological order is in a system with a given fractional $\sigma_H$. As some of the authors argued in a recent paper \cite{musser_fractionalization_2024}
two reasonable definitions are that minimal theories have the smallest: 1) number of anyons $N_{{\cal T}}$ or 2) total quantum dimension $\mathcal{D}_{{\cal T}}$, among all theories consistent with the IR symmetries and anomalies.  We will determine minimal theories for FQH systems through either definition below. 

These minimal theories are important for a number of reasons. First, as we discuss in Sec.~\ref{sec: summary}, the vast majority of (currently) experimentally observed FQH states seem consistent with the minimal theories described below. Thus, even in situations where the microscopic description is not fully understood, the minimal theories may be a good guess for the FQH state with a given $\sigma_H$. Second, in the theoretical discussion thus far, we considered theories that are endpoints under gauging and removing  neutral, decoupled factors that preserve the responses. While there is essentially a unique endpoint theory for odd $q$, for even $q$ there are infinitely many. Thus, it is also useful to impose a further minimality condition to narrow down the list.  Finally, defining minimality by $N_{{\cal T}}$ and ${\cal D}_{{\cal T}}$ can be useful in numerical studies which often find the lower bound on anyon count, see, e.g.,~\cite{huang_non-Abelian_2024}. 

Below, we take $N_\cT$ to be the total number of anyons, identifying $c\sim 1$ in the counting. We also define the quantum dimension ${\cal D}_{\cal T}$ with the same convention.

\subsubsection{Single-layer spin{}$^c$, even $q$}\label{singlelayers}

\begin{figure*}
\includegraphics[width=0.5\textwidth]{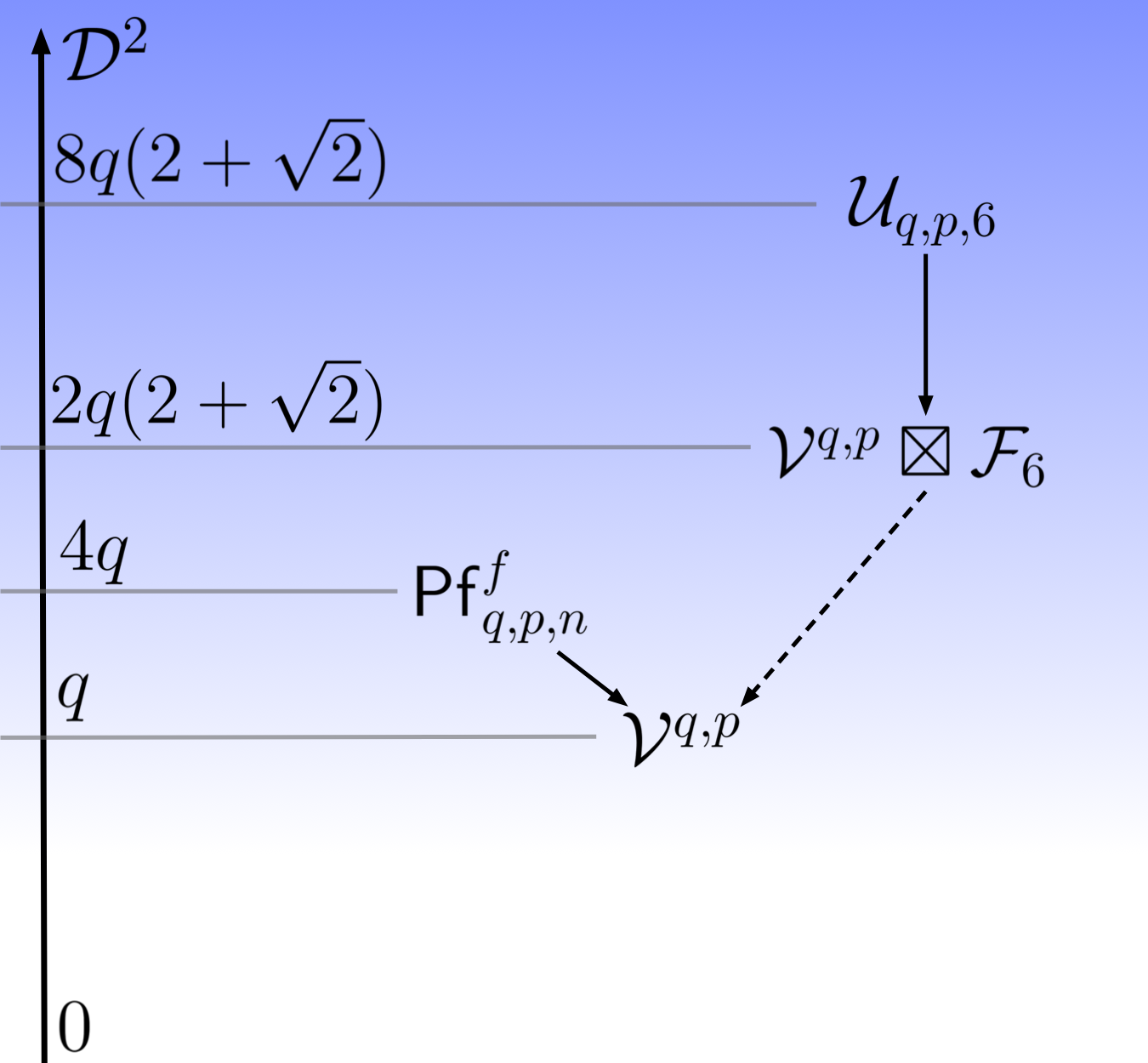}
\caption{For even $q$, we see that allowing ourselves to violate the spin/charge relation restores the unique endpoint. To do this, we gauge the one-form symmetry associated with $v^q$ and descend to $\ell=1$, indicated by solid arrows. We then remove the uncharged sectors, indicated by dashed arrows. As shown in the main text, the theories $\mathsf{Pf}^f_{q,p,n}$ are special because they are the only theories without an uncharged sector upon descending to $\ell=1$.}
\label{fig:even}
\end{figure*}

Let us set out our proof strategy for proving minimality by $N_{{\cal T}}$ or $\mathcal{D}_{{\cal T}}$. We denote these minimal theories by $\mathcal{T}$. 

Step one of our argument involves noting $\mathcal{V}^{2q,2p}\subseteq \mathcal{T}$. Thus the number of anyons and quantum dimension squared of $\mathcal{T}$, $N_{{\cal T}}$ and $\mathcal{D}^2_{{\cal T}}$, respectively, must be strictly larger than $2q$, as $\mathcal{V}^{2q,2p}$ is not a complete theory. Step two involves noting that $\pf_{q,p,n}^f$ is a complete theory that is consistent with the symmetries and anomalies. So $2q< N_{{\cal T}}\leq 3q$. This immediately implies that in the minimal theory by $N_{{\cal T}}$, $\ell =2$. Let us first establish what the values of $N_{{\cal T}}$ and $\mathcal{D}_{{\cal T}}$ will be in the minimal theory by anyon count, and then produce the list of possible minimal theories.

We start with $N_{{\cal T}}$. Because of charge conservation for any anyon $x\in \cT$, $x\times v^k$ for $k=1,\cdots, q-1$ are all distinct from each other because $v^k$ carries charge $pk/q$. This implies that $N_\cT$ must be a multiple of $q$. Together with the bounds found above, we conclude $N_\cT=3q$.

We next group anyons in $\cT$ according to their braiding with $v^q$. We take $\cT_{0}$ ($\cT_1$) to be those with braiding phase $0$ ($1/2$). Clearly ${\cal V}^{2q,2p}\subset \cT_0$. Also, since $\cT_1$ is non-empty, the argument above implies there are at least $q$ anyons in $\cT_1$. Since $N_{\cT} = 3q$ this tells us that these inequalities are equalities and $\cT_0={\cal V}^{2q,2p}$, $\cT_1=\{x \times v^k|k=0,1,\dots, q-1\}$. In particular, this implies $x\times v^qc=x$.\footnote{Here we have not identified $c\sim 1$ to make it manifest that this fusion rule obeys the spin/charge condition.} It is easy to see that $\cT = \cT_0\oplus \cT_1$ gives a $\bZ_2$ grading on $\cT$, using Lemma \ref{lemma:G-graded-D},  we see that $2q = {\cal D}^2_{{\cal T}_0} ={\cal D}_{{\cal T}_1}^2 = qd_x^2$ and thus that $d_x=\sqrt{2}$. The total quantum dimension squared is ${\cal D}^2_{{\cal T}} = 2{\cal D}^2_{{\cal T}_0}=4q$.

To construct this theory, we tentatively violate the spin$^c$ condition, as in Fig.~\ref{fig:even}, and gauge the one-form symmetry generated by the charge one boson $v^q$. Then, $\cT$ is reduced to ${\cal V}^{q,p}$. To recover the parent TQFTs, we gauge the dual zero-form $\bb{Z}_2^{(0)}$ symmetry acting in the child TQFT ${\cal A}^{q,p}$. The action on the anyons is determined such that after the gauging, ${\cal A}^{q,p}$ becomes ${\cal V}^{2q,2p}$. One choice for $\cT$ is ${\cal A}^{4q,p}$ with vison $v=a_2$. All other choices of $\cT$ can be obtained by tensoring invertible theories when gauging. In this case, the invertible theories are fermionic SPT phases with global $\bZ_2^{(0)}$ symmetry \cite{cheng_classification_2018}. Gauging the $\bZ_2^{(0)}$ symmetry in the fermionic SPT phases gives ${\rm Spin}(n)_1$, with $\psi c$ being the gauge charge. We thus obtain the full list of candidates for $\cT$, by tensoring ${\cal A}^{4q,p}$ and ${\rm Spin}(n)_1$, and then gauging the diagonal $\bZ_2^{(1)}$ one-form symmetry generated by $\psi a_{2q}c$,
\begin{equation}
    \frac{{\rm Spin}(n)_1\boxtimes {\cal A}^{4q,p}}{(\psi a_{2q} c)}\,,
\end{equation}
which is the definition of the $\pf_{q,p,n}^f$ theory. These TQFTs have $N_\cT=4q$ anyons for even $n$ and $3q$ anyons for odd $n$, all with quantum dimension squared $4q$. We conclude that the minimal theories by anyon count are precisely those with odd $n$.

Finding the minimal theories by quantum dimension is even more straightforward. Since the full theory contains $\mathcal{V}^{2q,2p}$ and $\pf_{q,p,n}^f$ is a complete theory we must have $2q< {\cal D}_{{\cal T}}^2 \leq 4q$. Then, since ${\cal V}^{2q,2p}\subset \cT_0$ we must have ${\cal D}_{{\cal T}}^2 = 2{\cal D}_{{\cal T}_0}^2 \geq 4q$ \cite{barkeshli_symmetry_2019}. The combination of these two bounds reveals that ${\cal D}_{{\cal T}}^2=4q$. Using the same logic as above, we conclude that the minimal theories by quantum dimension are $\pf_{q,p,n}^f$ for all values of $n$.
 
We stop to note that all these minimal TQFTs can also be understood through the parton construction. Write the electron as $c=fb$, where $f$ is the neutral fermion and $b$ is a charge-1 boson. The $b$ bosons are at filling $p/q$, and the unique minimal TQFT for them is $\mathcal{V}^{q,p}$.\footnote{We show this in Sec.~\ref{singlelayerbo}, when discussing minimal orders of bosonic FQH states.} We then let the neutral fermions form a chiral topological superconductor with Chern number $n$.

\subsubsection{Minimal surface TO of a 3D topological insulator}

As an application of the above result, we show that the T-Pfaffian \cite{chen_symmetry_2014, bonderson_time-reversal_2013} is the unique minimal TQFT with the anomaly of a class AII topological insulator (TI) surface \cite{metlitski_symmetry-respecting_2015, wang_gapped_2013, wang_interacting_2014}.

To connect the two problems, we use the slab trick~\cite{metlitski_symmetry-respecting_2015}. Consider a slab of TI, where the top surface hosts a time-reversal-preserving TQFT,  and the bottom surface breaks the time-reversal symmetry (by adding a mass term to the Dirac fermion). Such a slab is equivalent to a 2+1d time-reversal breaking TQFT with $\sigma_H=1/2$. On the other hand, since the bottom surface is a massive Dirac fermion, the TQFT of the slab is completely determined by the top surface. The above discussion gives the full list of minimal TQFTs with $\sigma_H=1/2$. 

The problem then reduces to finding which of these can be consistent with (anomalous) time reversal. Let $\sigma$ be one of the two anyons in $\pf^f_{2,1,n}$ that have charge $Q(\sigma)=1/4$. It can be checked that $h(\sigma) = (n+1)/16$. Then since the full symmetry is  $\U\rtimes \bb{Z}_2^{\cal T}$,  i.e., time reversal preserves the $\U$ charge, ${\cal T}\sigma = \sigma$ or $v^2c\times \sigma$, with the latter being distinct from $\sigma$ for even $n$. Time reversal is anti-unitary, so $h({\cal T}\sigma)=-h(\sigma) $. This, along with ${\cal T}\sigma = \sigma$ or $v^2c\times \sigma$, means that $2h(\sigma) = 0\pmod{1}$. Since $h(\sigma) = (n+1)/16$, the only time-reversal-invariant TQFT among this list is the T-Pfaffian TQFT, corresponding to $\pf^f_{2,1,-1}\simeq\pf^f_{2,1, 7}$ in the above notation. Moreover, this action of time reversal is anomalous, as can be seen from the fact that the vison $v$ carries charge $1/2$; in a non-anomalous time-reversal-invariant 2+1d system, the vison must be charge neutral because $\sigma_H=0$.

In \cite{metlitski_symmetry-respecting_2015, wang_gapped_2013, seiberg2016gapped}, another spin{}$^c$ TQFT was constructed preserving time-reversal symmetry. As a topological order, it can be described as $\pf^f_{2,1,1}\boxtimes \cA^{2,-1}$ (recall that $\pf^f_{2,1,1}$ is the Moore-Read theory). Here $\cA^{2,-1}$ is charge-neutral, so if  we ignore the time reversal symmetry, $\cA^{2,-1}$ can be removed to get a minimal theory. However, the time-reversal symmetry mixes the two factors, and as a result, no further descent is possible without breaking time reversal. Thus $\pf^f_{2,1,1}\boxtimes \cA^{2,-1}$  is a non-minimal endpoint theory on the TI surface.

\subsubsection{Bilayer spin{}$^c$ theories}\label{bilayersection}

The previous discussion can be generalized to multilayer FQH systems.  Here ``layer" should be understood generally as an independent, discrete quantum number. For example, it can also be the spin, or the valley degree of freedom in graphene or moir\'{e} bands. The key assumption is that each layer comes with a separate $\U$ charge, which are all preserved in the ground state. Note that this requirement means that we will not allow ``interlayer" tunneling. If we did, it would break the two $\U$ symmetries to a single $\U$, taking us back to our single-layer discussion. 

In order to determine the minimal TQFT by anyon count/quantum dimension we will now need the full Hall matrix, as defined in  Sec.~\ref{QHP}.\footnote{When the full Hall matrix is known  \eqref{bilaryJE} and \eqref{bilaryrhoB} can be summarized by the response theory $\sum_{ab}{\sigma_H^{ab}\over 4\pi}\int_\tmfd A_a\wedge dA_b$. 
\label{responsebilayer}}  While the analysis above can be generalized, the results are much more involved. This is because for certain values of the drag Hall conductivity, the visons in each layer can be identified, while for others they cannot.

We now specialize to the case of $\nu = 1/2+1/2$ bilayers of electrons moving in a uniform magnetic field, and where each layer is restricted to be in the lowest Landau level.  This system has, in addition to the $\U\times \U$ symmetry associated with charge conservation in each layer, an anti-unitary particle-hole symmetry (denoted ${\cal CT}$). This symmetry means that each half-filled Landau level can be viewed either in terms of electrons at filling $\nu_e = 1/2$, or as holes doped into a filled Landau level at hole filling $\nu_h = 1/2$. This ${\cal CT}$ symmetry implies that $\sigma_H^{ab} = \delta^{ab}/2$.  To see why, recall that under $\cal T$ the Hall matrix changes sign: $\sigma_H \rightarrow -\sigma_H$. Under $\cal C$ (a unitary symmetry),  a fully filled Landau level is added in each layer. Thus, under $\cal CT$, the Hall matrix transforms as
\begin{equation}
    \sigma_H^{ab} \rightarrow \delta^{ab}-\sigma_H^{ab}.
\end{equation}
The particle-hole symmetry then implies $\sigma_H^{ab}=\delta^{ab}-\sigma_H^{ab}$, and $\sigma_H^{ab}=\frac12 \delta^{ab}$.

If we refer to the two layers as 1 and 2, then the previous arguments reveal that we must have two Abelian visons $v_1$ and $v_2$ with $\U^{(1)}\times\U^{(2)}$ charges $(Q_1,Q_2)=(1/2,0)$ and $(0,1/2)$, respectively. Their topological spins are $h(v_1) = h(v_2) = 1/4$ and they braid trivially with one another, i.e., $B(v_1, v_2) = 0$. Just as in the single-layer case, we denote the order of $v_1$ ($v_2$) by $2\ell_1$ ($2\ell_2$) for integers $\ell_{1,2}$. One can see that there must be at least $ 4 \ell_1 \ell_2\geq 16$ anyons in $\sc{S} = \mathcal{A}^{2\ell_2,\ell_2} \boxtimes \mathcal{A}^{2\ell_1,\ell_1} \subseteq \sc{T}$, so the number of anyons and quantum dimension squared is lower bounded by 16.

In the bilayer example, the spin/charge relation requires all local anyons to be bosons (fermions) if they have even (odd) total charge, $Q_1+Q_2$. Since $v_1^2, v_2^2$ are both bosons with total charge $1$, the spin/charge relation means $v_1^2,v_2^2\neq 1$. Thus $\ell_1$ and $\ell_2$ must be even. There must then be at least eight anyons in $\sc{S}\subseteq \sc{T}$, lower bounding the anyon count and quantum dimension squared.\footnote{The lower bound is $8 = 4\times 4/2$ and not $16$ because of the possibility that $v_1^2 = v_2^2$, as seen in the example of $D(\bb{Z}_4)$.} On the other hand, there is a known candidate TQFT
\begin{equation}
\sc{T}_{\mathrm{candidate}} =  D(\bb{Z}_4)\,,
\end{equation}
with the identification that $v_1 = em$, $v_2 = e\overline{m}c$ \cite{sodemann_composite_2017}. This has $16$ total anyons, upper bounding the anyon count and quantum dimension squared.

This upper bound immediately tells us that in any minimal theory, we must have $\ell_1 = \ell_2 =  2$ and  $v_1^{2}v_2^2=1$. Just as in the single-layer case, we can gauge the $\bb{Z}_2^{(1)}$ one-form symmetry generated by $v_1^2=v_2^2$ to produce the child TQFT $\mathcal{A}^{2,1}\boxtimes \mathcal{A}^{2,1}$. We then gauge the dual $\bb{Z}_2^{(0)}$ $0$-form symmetry in this child TQFT, with the ambiguity of stacking fermionic SPT states protected by $\bb{Z}_2^{(0)}\times \bb{Z}_2^{(0),f}$. As in the single-layer case, this gives us the possible parent TQFTs
\begin{equation}
\sc{W}_n = \frac{\mathrm{Spin}(n)_1\boxtimes D(\bb{Z}_4)}{(\psi e^2m^2c)}\,,
\label{eqn:fermion_bilayer_mTOs}
\end{equation}
where again, the denominator indicates that we gauge the one-form symmetry generated by $\psi v_1^2c = \psi e^2m^2c$. Note that as spin TQFTs, ${\cal W}_n$ is equivalent to ${\cal W}_{n+8}$, but depending on the $\U$ charge assignment, they may be distinct spin{}$^c$ TQFTs.

For odd $n$ the TQFT ${\cal W}_n$ has twelve anyons, smaller than the sixteen of the $D(\bb{Z}_4)$ example. Four of these anyons are non-Abelian, two have topological spins of $n/16$, and two have $(n+8)/16$. Pick one of these, $\sigma$, such that $h(\sigma) = n/16$. Since ${\cal CT}$ is anti-unitary it must map $\sigma$ to another non-Abelian anyon, ${\cal CT}(\sigma)$, such that $h\big({\cal CT}(\sigma)\big) = -n/16\mod 1/2$. But the only other choice of $h$ is $(n+8)/16\neq -n/16$ mod $1/2$. So for odd $n$, these theories do not have any anti-unitary symmetry and are ruled out. This means that the smallest TQFTs that satisfy $\nu = 1/2+1/2$ and are ${\cal CT}$-invariant are Abelian with sixteen anyons.

A similar argument can be used to rule out $n=2\pmod 4$. So the only possible minimal theory apart from $D(\bb{Z}_4)$ is ${\cal W}_4$. Although this theory can be consistent with an anti-unitary symmetry \cite{jian_minimal_2024}, we will see that it is not consistent with the full $\U^{(1)}\times \U^{(2)} \times \bb{Z}_2^{{\cal CT}}$ symmetry. We begin by noting that, without loss of generality, ${\cal W}_4$ is generated by $a$ and $v_1$, where $a$ is given by the fusion of an order-two semion in ${\rm Spin}(4)_1$ and $e\in D(\bb{Z}_4)$. It is an order-four anyon with $h(a) = 1/4$ and $B(a,v_1) = 1/4$. Further, in terms of these generators, $v_2 = a^2\overline{v}_1c$ and $a$ has an integer total charge.

Now, assume, by contradiction, that ${\cal W}_4$ is consistent with all the symmetries. First, we want to determine the action of ${\cal CT}$. Since ${\cal CT}$ flips the sign of the $\U^{(i)}$ charge for all anyons \cite{sodemann_composite_2017} we must have that ${\cal CT}v_i = v_i$ or $v_ic$ for $i=1,2$. Then $h({\cal CT}v_i) = -h(v_i)=-1/4$ fixes ${\cal CT}v_i = v_ic$. Using this and $v_2 = a^2\overline{v}_1c$, we see that ${\cal CT}a^2=a^2$ and thus that ${\cal CT}a = a\times d$ for some order two anyon $d\in {\cal W}_4$. Since $a$ has integer total charge, this must mean that $d$ has even integer total charge. The only possibilities for $d$ are then $1,a^2,v_1^2c,a^2v_1^2c$. But it can be checked that for all of these possibilities $h(a\times d) = h(a) \neq -h(a) = h({\cal CT}a)$. So we have arrived at a contradiction and can rule out ${\cal W}_4$.

This provides a proof of the conjecture in \cite{sodemann_composite_2017} that $D(\bb{Z}_4)$ is indeed the unique minimal theory consistent with $\nu = 1/2+1/2$ and $\U^{(1)}\times \U^{(2)} \times \bb{Z}_2^{{\cal CT}}$ symmetry. 

\subsection{Single-layer bosonic theories}\label{singlelayerbo}

For bosonic theories, the discussion is different.  Now, the consistency condition is $rs\in 2\mathbb{Z}$.  Then, for even $pq$  we can have $\ell=1$, so the minimal theories are $\mathcal{V}^{q,p}$.  

However, when both $q$ and $p$ are odd, ${\cal V}^{q,p}$ is not a valid bosonic theory. Thus $\ell$ must be even. Therefore, the minimal theory has $\ell=2$ and includes\footnote{In the first step, we used the fact that for odd $q$, as groups $\bZ_{2q}=\bZ_2\otimes \bZ_q$. And then, there cannot be a mixed anomaly between the two factors. In the second step, we used the relations derived in Sec.~\ref{sec:symmetrylines} and the fact that $p$ and $q$ are odd.  It should also be noted that the vison in the left-hand side is the product of the generators of the two factors in the right-hand side.} $\mathcal{V}^{2q,2p}\cong \mathcal{A}^{2,2pq}\boxtimes\mathcal{A}^{q,4p} \cong \mathcal{A}^{2,2}\boxtimes\mathcal{A}^{q,4p}$.  Since this is not a complete theory, the full theory does not factorize, and there must be additional lines charged under the $\mathbb{Z}_{2}^{(1)}$ one-form symmetry. As an example, for $p=q=1$ (corresponding to the bosonic $\nu=1$ state) the minimal symmetry is ${\cal A}^{2,2}$. It includes the identity and a nontrivial vison with spin $1/2$.  This is not a complete TQFT.  It can be completed in various ways including  Spin$(n)_1$ and ${\rm SU}(2)_{4k+2}$. This leads to infinitely many endpoint theories, similar to the single-layer spin{}$^c$ case with even $q$. 

More generally, the minimal theories as measured by $N_{\cT}$ and $\mathcal{D}_{\cT}$ are generalizations of the bosonic Pfaffian states. They are given by
\begin{equation}
\pf^b_{q,p,n} = \mathrm{Spin}(n)_1 \boxtimes \mathcal{A}^{q,4p} \qquad\text{ for } \qquad,\qquad n=1,\ldots, 16\qquad,\qquad p,q\in 2\bZ+1\,,
\label{eqn:boson_mTOs}
\end{equation}
(following the notation in Appendix~\ref{app:anyontheory}). 

We note that for bosons in a Landau level at $\nu = 1$,  a Pfaffian ground state is found in numerical calculations with a contact~\cite{cooper2001quantum,regnault2003quantum,regnault2004quantum} or with a Coulomb~\cite{chang2005composite} repulsive interaction.  
This state (usually referred to as ``the" bosonic Pfaffian) is $\pf^b_{1,1,3}$ in our notation. This state, as well as the other minimal bosonic Pfaffian states at $\nu = 1$, can be given a composite fermion construction~\cite{wang2016composite}.  

All the theories in \eqref{eqn:boson_mTOs} are endpoints with minimal total quantum dimension.  Those with odd $n$ are minimal by number of anyons.

\section{Summary and discussion}
\label{sec: summary} 

The goal of this paper is to use the value of the fractionally quantized Hall conductivity to characterize the topological order (i.e., TQFT) of a quantum Hall system. Specifically, we answered two questions: 
\begin{itemize}
    \item 
What is the minimal topological order that is consistent with a given value of $\sigma_H$? 
\item Is there an ordering of various allowed topological orders for a given $\sigma_H$? 
\end{itemize}

The quantum Hall effect may be defined through the response of the system to a background gauge field $A$ that couples to the electric current.  
This response theory is
\begin{equation}\begin{split}\label{AdArespc}
    &{\sigma_H\over 4\pi } \int_\tmfd A\wedge dA\\
    &\sigma_H={p\over q} \qquad,\qquad p,q\in \mathbb{Z}\qquad,\qquad \gcd(p,q)=1\,.
\end{split}\end{equation}
We began with the crucial observation that a fractional $\sigma_H$ implies the existence of a fractionally charged Abelian anyon $v$, known as the vison. Equivalently, we can state that 
the IR TQFT has an anomalous global one-form symmetry. We identified these facts as the organizing principle of the quantum Hall systems.

The vison $v$ has a fractional charge $p/q$ and spin $p/2q$. The quasiparticle $v^j$ (obtained by fusion) has 
\begin{equation} 
Q(v^j) = \frac{pj}{q}\qquad, \qquad h(v^j) = \frac{pj^2}{2q} \,.
\end{equation} 

Before stating our results in the general case, let us specialize to the question of minimal topological order in electronic systems where the fractional quantum Hall effect has been observed in experiments. First, note that  $(1, v, v^2,.....v^{q-1})$ are all distinct quasiparticles.  Second, $v^q$ braids trivially with $v^j$ for all $j$.  For odd-denominator states (i.e., odd $q$), it has charge $p$ and Fermi (Bose) statistics when $p$ is odd (even). The minimal topological order is one where $v^q$ is simply identified with a local particle. Thus, the minimal topological order is unique and has exactly $q$ anyons (and is denoted ${\cal V}^{q,p}$). For even-denominator states, $v^q$ is an odd-charged boson, which therefore cannot be local in an electronic system. There must therefore be anyons distinct from those generated by the vison with which $v^q$ braids nontrivially. We showed that in the minimal topological order for even $q$, $v^{2q}$ can be identified with a local particle. Though the minimal order in this case is not unique, there are a small number of them that are all Pfaffian states. In particular, the minimal orders by anyon count have non-Abelian anyons with which $v^q$ braids nontrivially.  

Our conclusions about the minimal theories for a given $\sigma_H$ (by anyon count or quantum dimension) were summarized in Table \ref{tab:mTOs} in our summary of results.

Now, let us summarize the rest of our results using the framework of the anomalous one-form symmetries of the IR TQFT. 
In general, for our purposes, we distinguish between three kinds of quantum field theories (which are not necessarily topological):
\begin{itemize}
\item[Bosonic] If all the fundamental degrees of freedom are bosons, then our spacetime does not have to be a spin manifold.  And if it is a spin manifold, there is no dependence on its spin structure.  In that case, shifting $\sigma_H$ by an even integer corresponds to adding a local counterterm and does not affect the bulk dynamics.  Even though it can affect the boundary dynamics, we will ignore it and identify
    \begin{equation}\label{pidentifcb}
    p\sim p+2q\,.
\end{equation}
\item[Fermionic] If the fundamental degrees of freedom include also fermions, the theory can be formulated only on spin manifolds, and the physical answers can depend on the choice of spin structure. In that case, shifting $\sigma_H$ by an arbitrary integer corresponds to adding a local counterterm and does not affect the bulk dynamics.  Even though it can affect the boundary dynamics, we will ignore it and identify
     \begin{equation}\label{pidentif}
    p\sim p+q\,.
\end{equation}
\item[Electronic] If all the fundamental fermions have odd $\U$ charges and all the fundamental bosons have even $\U$ charges,\footnote{For example, this is not the case when neutral fermions play a role in the dynamics. Though not pertinent to electronic solids, this may be realized in other systems, e.g., Bose-Fermi cold atom mixtures, where we only keep track of the total number of atoms.} the theory is an electronic theory.  It can be formulated on nonspin manifolds, provided the background field $A$ is a spin$^c$ connection.  This means that its fluxes through various cycles are correlated with the obstruction to defining spinors.  Concretely, $\left(2\int {dA\over 2\pi}\right) \mod 2=w_2$.  In this case, shifting $\sigma_H$ by an arbitrary integer can also be ignored for the study of bulk properties, and we identify
    \begin{equation}\label{pidentisc}
    p\sim p+q\,.
\end{equation}
\end{itemize}
The largest class of theories is the fermionic one. Some of the fermionic theories are also bosonic.  Their correlation functions do not depend on a choice of spin structure.  And some of them are also electronic.

We assumed that the theory is gapped and is described by a TQFT.  Then, the response theory \eqref{AdArespc} arises because the TQFT has a $\bZ_s^{(1)}$ one-form global symmetry generated by the vison $v$.  Its anomaly is labeled by the integer $r$.  This means that the TQFT includes the symmetry lines
\begin{equation}\begin{split}
&{\cal V}^{s,r}=\{1,v,v^2,\cdots , v^{s-1}\}\subseteq {\cal T} \\
&r=\ell p\qquad, \qquad s=\ell q\qquad,\qquad \ell=\gcd(s,r)\,.
\end{split}
\end{equation}
The parameters $s$ and $r$ are always integers.  But their allowed values depend on the nature of the theory:
\begin{equation}\label{allowedsr}
\begin{split}
    &{\rm Fermionic\ theories:} \qquad\qquad\qquad\quad\qquad r\sim r+s\\
    &{\rm Bosonic\ theories:}\quad \quad rs\in 2\bZ \qquad ,\qquad\  r\sim r+2s\\
    &{\rm Electronic\ theories:} \quad r(s+1)\in 2\bZ \quad ,\quad  r\sim r+s\,.
\end{split}
\end{equation}

The theories with $\ell=1$, are the simplest ones.  In that case, the TQFT factorizes as a product of two decoupled TQFTs
\begin{equation} 
{\cal T}={\cal V}^{q,p} \boxtimes {\cal T}'\,,
\end{equation}
where the anyons in $\cal T'$ are $\U$ neutral and do not contribute to the response.  In that sense, $\cal T'$ is trivial and can be ignored.

If $\ell\ne 1$, no such factorization takes place.  However, in that case, we can attempt to gauge the one-form symmetry $\bZ_\ell^{(1)} \subset \bZ_s^{(1)}$ to find a more minimal theory.  In order to be able to do that, ${\cal V}^{q,p}$ should be a valid set of symmetry lines.  If this is the case, after gauging, we have $\ell=1$, the theory factorizes, and we end up with the minimal theory ${\cal V}^{q,p}$.

Let us examine it in more detail.

For fermionic theories, there is no restriction on $q$ and $p$, and therefore, we can always gauge the $\bZ_\ell^{(1)}$ one-form symmetry to find the minimal theory ${\cal V}^{q,p}$.

For bosonic theories with $pq\in 2\mathbb{Z}$, we can gauge $\mathbb{Z}_\ell^{(1)}$ and end up with the minimal theory ${\cal V}^{q,p}$.  Otherwise, i.e., when $pq\in 2\mathbb{Z}+1$, we can gauge only  $\mathbb{Z}_{\ell\over 2}^{(1)}$, leading to the symmetry lines ${\cal V}^{2q,2p}$.  Then, the theory does not factorize.

For electronic theories, with $p(q+1)\in 2\mathbb{Z}$, which means that $q\in 2\bZ+1$, we can gauge $\mathbb{Z}_\ell^{(1)}$ and end up with the minimal theory ${\cal V}^{q,p}$.  Otherwise, i.e., when $q\in 2\bZ$, we can gauge only  $\mathbb{Z}_{\ell\over 2}^{(1)}$, leading to the symmetry lines ${\cal V}^{2q,2p}$.  Then, the theory does not factorize.

Let us now discuss how our results on the minimal TQFTs (defined by anyons count) and the more general ordering for a given $\sigma_H$ pertain to other known results on quantum Hall systems. Experimentally, the fractional quantum Hall effect has only been observed in electronic systems thus far (corresponding to our spin$^c$ theories), and we will focus on these. 

There are on the order of $100$ fractional quantum Hall states that have been observed thus far. The vast majority of them occur in the lowest Landau level. Direct experimental measurements of the topological data are mostly limited to the Hall conductivity. For a small number of states, the thermal Hall conductivity, which corresponds to the chiral central charge, has also been measured~\cite{banerjee2017observed,banerjee2018observation}. The fractional charge~\cite{saminadayar1997observation,de1998direct,martin2004localization,dolev2008observation,radu2008quasi} and fractional statistics~\cite{nakamura2020direct,bartolomei2020fractional} of quasiparticles have been measured in even fewer states (for a review, see \cite{feldman2021fractional}). However, given the simplicity of the $2d$ electron gas in a Landau level, highly accurate microscopic calculations are possible to nail down the state. For electrons in the lowest Landau level, the prominently observed states occur at the Jain fillings $\nu = \frac{n}{2np \pm 1}$, and are well described by the composite fermion construction (for a review, see  \cite{jain2007composite}). As we discussed, these are described by minimal TQFTs ${\cal V}^{2np \pm 1,n}$. 

Interestingly,  Jain fractions (with $p = 1$) are also what are observed in the FQAH systems discovered recently~\cite{cai_signatures_2023,park_observation_2023,zeng_thermodynamic_2023, lu_fractional_2024}. Existing numerical studies~\cite{reddy2023toward,yu2024fractional} in ${\rm tMoTe}_2$ support the identification of these with the standard Jain states. The microscopic situation in rhombohedral graphene is less clear, but the appearance of FQAH states at the Jain fractions~\cite{lu_fractional_2024}, together with the composite Fermi liquid at half-filling of the moir\'{e} lattice, suggest that these FQAH states are also Jain states, and hence are described by minimal TQFTs. 

All of these observed states have $\sigma_H = p/q$ with $q$ odd. An even denominator state occurs, famously, for electrons in a magnetic field in the second Landau level at filling $\frac{5}{2}$, and evidence has accumulated for a paired quantum Hall state with non-Abelian quasiparticles. The main candidates are the Moore-Read Pfaffian state $\pf^f_{2,1,1}$, its particle-hole conjugate (the anti-Pfaffian), and the T-Pfaffian $\pf^f_{2,1,-1}\simeq \pf^f_{2,1,7}$. There is some uncertainty still on exactly which of these states is realized. For our purposes, it is enough to note that these are also minimal states. 

Thus, we see that the vast majority of observed fractional quantum Hall states are described by minimal TQFTs. However, based on numerical calculations, a few non-minimal TQFTs have been proposed as well, primarily in the second Landau level.  The most prominent of such states is one at Landau level filling $\nu = \frac{12}{5}$. Numerical studies~\cite{read_beyond_1999,sreejith2013tripartite, FibonacciDMRG1, FibonacciDMRG2, Pakrouski2016} (for the Coulomb interaction) appear consistent with the particle-hole conjugate $\bar{R}_{1,3}$ of the Read-Rezayi state $R_{1,3}$, though a Jain state ${\cal V}^{5,2}$ is allowed at that filling. As $ \bar{R}_{1,3} = \bar{{\cal F}_3} \boxtimes \mathcal{V}^{5,2}$, the Read-Rezayi state has five more anyons than what is needed for minimality. In experiments in GaAs~\cite{xia2004electron,kumar2010nonconventional}, the corresponding state at $\nu = \frac{13}{5}$ is not observed, an effect that has been attributed to Landau level mixing~\cite{FibonacciDMRG2,Pakrouski2016}. For a few other proposed nonminimal states in the second Landau level, see \cite{balram2020fractional,coimbatore2021non,balram2020z} though these have not been studied as extensively as the possibility of the $\bar{R}_{1,3}$ state at filling $\frac{12}{5}$. 

Interesting possible examples of non-Jain states occur at fillings such as $\frac{8}{17}$ or $\frac{7}{13}$~\cite{singh2024topological}. Though conventional Jain states are allowed, these have been proposed to be so-called ``daughter states" of parent quantum Hall states at the proximate filling $\nu = \frac{1}{2}$~\cite{levin2009collective}. The parent $\nu = \frac{1}{2}$ states are Pfaffians: at this filling, the composite fermions interact with a zero effective magnetic field, and the Pfaffian states result when Cooper pairs of composite fermions are condensed.   When the filling is moved away from $1/2$, the composite fermion Cooper pairs interact with a non-zero effective magnetic field, which suppresses their condensation. Then, the daughter states appear at fillings for which the composite fermion Cooper pairs form bosonic integer quantum Hall states rather than condensates~\cite{yutushui2024paired}. Different Pfaffian parent states have daughters at different fillings, and thus the observation of these daughters has been suggested to be a fingerprint of which Pfaffian state is realized at half-filling~\cite{yutushui2024paired,zheltonozhskii2024identifying}, a question that so far has not yet been settled.  From the point of view in the present paper, these daughter states are {\it minimal} at their fillings. This is because they can all be described as Jain states tensored with invertible electrically neutral $E_8$ states~\cite{zheltonozhskii2024identifying,zhang2025hierarchy}. Thus, the difference with the Jain states is not at the level of anyon content, and these daughter states have the minimal number of anyons at their fillings.  

More generally, our discussion is intrinsic to the TQFT and does not rely on the dynamical mechanism underlying its creation.  This is to be contrasted with the analysis based on the composite fermions and their interactions, which is more detailed and relies on additional information.

Why are minimal TQFTs realized so commonly in the fractional quantum Hall effect?  For the Coulomb interaction, non-interacting (i.e., a free fermion Hamiltonian, perhaps with the inclusion of pairing terms) composite fermions provide a good ``mean-field" starting point to describe most FQH states, and these inevitably lead to minimal states.\footnote{In this construction, the quantum Hall states correspond to composite fermions in an integer quantum Hall state (the Jain states) or they form pair condensates (the Pfaffian states). Both sets of states are minimal according to our definition.} Thus, the question of why minimal states are apparently preferred can be replaced with the question of why the non-interacting composite fermion construction is so successful. Of course, if the interaction is modified, then other non-minimal states can occur, as exemplified by the exact parent Hamiltonians for many quantum Hall wavefunctions.  

The commonality of minimal states makes them good guesses as candidates for future realizations of FQH states in platforms very different from electrons in a Landau level, such as in FQAH systems on a lattice. Even in the absence of a microscopic understanding in terms of standard composite fermions, it may be reasonable to hypothesize that the observed state is a minimal one, and see if that is supported by further experiments/theory.

\acknowledgements 
We thank Maissam Barkeshli, Jaume Gomis, Zhaoyu Han, Jainendra Jain, Kristan Jensen, Chao-Ming Jian, Anton Kapustin, Zohar Komargodski, Ethan Lake, Ho Tat Lam, Greg Moore, Salvatore Pace, Shu-Heng Shao, Sahand Seifnashri, G. Sreejith, Cenke Xu, and Carolyn Zhang for discussions. MC is partially supported by NSF Grant DMR-2424315. SM is supported by the Laboratory for Physical Sciences. AR is partially supported by the U.S. Department of Energy under Grant No. DE-SC0022021. NS is partially supported by DOE Grant DE-SC0009988. TS is partially supported by NSF Grant DMR-2206305. AR, NS, and TS are also partially supported by the Simons Collaboration on Ultra-Quantum Matter, which is a grant from the Simons Foundation (Grant No. 651678, AK; Grant No. 651444, NS; and Grant No. 651446, TS).

\appendix

\section{Useful examples of some anyon theories}
\label{app:anyontheory}

Here we review a few anyon theories that we reference throughout the paper. For a more thorough review of these theories, see, e.g.,  \cite{bonderson_non-Abelian_2007}. Unless otherwise specified, we follow the notations of this reference.

\subsection{$\bb{Z}_N$ gauge theory or $D(\bb{Z}_N)$}

Here $N\in \bb{Z}$ is a positive integer. The $D$ indicates that this theory is the quantum double of $\bb{Z}_N$, i.e., $D(\bb{Z}_N)\simeq \bb{Z}_N \otimes \bb{Z}_N$ as an Abelian group under fusion. It describes the anyons (or equivalently, one-form symmetries) in a $\bZ_N$ gauge theory. The anyons in this theory can thus be written as a pair of integers modulo $N$, $(a_0, a_1)$. The identity anyon is given by $(0,0)$ and the fusion rules are 
\begin{equation}
(a_0, a_1)\times (b_0, b_1) = (a_0 + b_0, a_1 + b_1)\,.
\end{equation}
It is common to take $e=(1,0)$ and $m=(0,1)$ as the generators of $D(\bb{Z}_N)$, and we follow this notation. The $e$ and $m$ anyons can be understood as the fundamental electric and magnetic charges in the $\bZ_N$ gauge theory.

The braiding data is given by
\begin{equation}
h(e^{a_0}m^{a_1}) = \frac{a_0a_1}{N}\qquad , \qquad B(e^{a_0}m^{a_1}, e^{b_0}m^{b_1})=\frac{a_0b_1+b_0a_1}{N}\,.
\end{equation}
Clearly, this braiding is nondegenerate for all values of $N$, and hence, this theory is modular.  This braiding means that each $\mathbb{Z}_N^{(1)}$ one-form symmetry is anomaly free, but there is a mixed anomaly between the two $\mathbb{Z}_N$ factors. For even $N$, the theory does not  factorize in terms of independent $\mathbb{Z}_N$ factors. But for odd $N$ it factorizes as $\cA^{N,2}\boxtimes\cA^{N,-2}$, which are generated by $em$ and $em^{-1}$.

This theory can also be described using the  Chern-Simons Lagrangian \cite{Maldacena:2001ss},
\begin{equation}
    \frac{N}{2\pi}b_1\wedge db_2+\frac{t^1}{2\pi}b_1\wedge dA+\frac{t^2}{2\pi}b_2\wedge dA\,.
\end{equation}
Here, we introduced a background $\U$ gauge field $A$, and $t^1, t^2$ are its couplings. The corresponding vison is $(t^1,t^2)$.

\subsection{Spin$(n)_1$ and Spin$(n)_1\boxtimes\{1,c\}$}

Here, Spin$(n)_1$ with $n\mod {16}$ denotes the TQFT of the level-$1$ Spin$(n)$ Chern-Simons theory. Alternatively, it is the $n$th TQFT in Kitaev's 16-fold way \cite{kitaev_anyons_2006}. The corresponding RCFTs are theories of $n$ Majorana fermions summed over their spin structures.

This is a bosonic theory with $\sc{D}^2=4$.  It possesses an Abelian fermion $\psi \in \mathrm{Spin}(n)_1$, which squares to the identity. Indeed, these TQFTs can be understood as the ``minimal modular extensions'' of the fermionic/spin TQFT given by $\{1,\psi\}$.  

We first review these theories with $n$ an odd integer between $1$ and $15$. They are non-Abelian TQFTs with anyon content and fusion rules of the Ising TQFT, i.e.,
\begin{align}
\mathrm{Spin}(n)_1 &= \{1,\sigma,\psi\}\\
\psi^2 &= 1\\
\psi \times \sigma &= \sigma\\
\sigma \times \sigma &= 1+\psi\,.
\end{align}
From this it is clear that $d_1 = 1 = d_\psi$ and $d_{\sigma} = \sqrt{2}$. Eight choices of exchange and braiding statistics are consistent with these fusion rules, as classified in Appendix~D in \cite{Moore:1988qv} and in \cite{kitaev_anyons_2006},
\begin{align}
B(\psi, \sigma) &= {1\over 2}\,,\\
h(\psi) &= {1\over 2}\,,\\
h(\sigma) &= \frac{n}{16}\,.
\end{align}
{This theory has a $\mathbb{Z}_2^{(1)}$ one-form symmetry, generated by $\psi$.  Using our notation, the symmetry lines $\{1,\psi\}$ are ${\cal A}^{2,2}$.}  The chiral central charge of the theory is given by $c_- = n/2\mod {8}$, so they are half-integer.  

For $n$ an even integer, these are Abelian TQFTs. They were reviewed extensively by Kitaev \cite{kitaev_anyons_2006}, and we refer the reader to that reference. They happen to coincide with some of our previous definitions,
\begin{equation}
\mathrm{Spin}(n)_1 = \begin{cases} \cA^{4,n/2} &\mbox{if } n=2\mod {4}\\
\cA^{2,n/4}\boxtimes \cA^{2,n/4} &\mbox{if } n=4\mod {8}\\
D(\bb{Z}_2) &\mbox{if } n =0\mod {16}\,.
\end{cases}
\end{equation} 
Their chiral central charge is again given by $c_- = n/2\mod {8}$ and is thus an integer. For $n=8$ mod 16, the one-form global symmetry is $\mathbb{Z}_2^{(1)}\otimes\mathbb{Z}_2^{(1)}$. Each factor has its own anomaly, and there is also a mixed anomaly between the two factors.  Therefore, this theory does not factorize.

We now turn to the spin theories $\mathrm{Spin}(n)_1\boxtimes\{1,c\}$. Since $h_{{\rm Spin}(n)}(\sigma\times c) = h(\sigma)+\frac12 = h_{{\rm Spin}(n+8)}(\sigma)$, they satisfy $\mathrm{Spin}(n)_1\boxtimes\{1,c\} \simeq \mathrm{Spin}(n+8)_1\boxtimes\{1,c\}$.  There are thus only eight distinct spin theories consistent with these fusion rules: four Abelian and four non-Abelian. We note that in this case, the chiral central charge cannot be used to tell these TOs apart, since the chiral central charge in a spin TQFT can always be changed by $1/2$ by stacking invertible $p+ip$ superconductors. 

\subsection{${\cal F}_k$}
\label{app:Fk}

${\cal F}_k$ is defined as the integer-spin subcategory of SU(2)$_k$.\footnote{In \cite{bonderson_non-Abelian_2007}, this subcategory was referred to as ``SO(3)$_k$."  However, it is important to note that ${\cal F}_k$ is in general different from the standard definition of SO(3)$_k$ in RCFT, where it describes the 1+1d WZW models on the group manifold SO(3) with coefficient $k$ \cite{Gepner:1986wi}, or its chiral vertex operator description in \cite{Moore:1988ss}.  The latter definition also fits the 2+1d Chern-Simons theory with SO(3) gauge group with level $k$ \cite{Moore:1989yh}.  These SO(3)$_k$ theories are the result of gauging a $\bZ_2^{(0)}$ zero-form symmetry in the 1+1d SU(2)$_k$, or a $\bZ_2^{(1)}$ one-form symmetry in the 2+1d SU(2)$_k$.  
This can be done consistently only for $k= 0$ mod 4 as a bosonic TQFT, or for $k=2$ mod 4 as a fermionic/spin TQFT. Interestingly, ${\cal F}_k$ and SO(3)$_k$ coincide when $k= 2$ mod 4.\label{comment on notation}} 

\begin{description}
    \item[$k$ odd] ${\cal F}_k$ is a bosonic TQFT. For $k=1$, it is the trivial theory. For $k=3$, this is the Fibonacci anyon theory. In general we have ${\rm SU(2)}_k = {\cal F}_k\boxtimes \cA^{2,k}$.  Note that this kind of factorization, which follows from \cite{Hsin:2018vcg} (see also \cite{drinfeld_braided_2010}), is the one we use often in this paper.
    \item[$k=2$ mod 4] ${\cal F}_k$ is a spin TQFT, with $j=k/2$ the transparent fermion. In this case, the symmetry lines of SU(2)$_k$ are in ${\cal A}^{2,k}\cong{\cal A}^{2,2}$ and since $\ell=\gcd(2,2)\ne 1$, SU(2)$_k$ does not factorize. Instead, as we mentioned in footnote \ref{comment on notation}, ${\cal F}_k={{\rm SU(2)}_k\over \mathbb{Z}_2}$.
    \item[$k= 0$ mod 4] ${\cal F}_k$ is a premodular theory, with $k/2$ the transparent boson.
\end{description}

\subsection{Pfaffian TOs}\label{PfaffianTO}

In this section, we review the TOs $\pf^f_{q,p,n}$ given in  \eqref{eqn:fermion_mTOs} for $q$ even. As above, we denote the anyon content of $\mathrm{Spin}(n)_1$ as $\{1,\sigma,\psi\}$ for $n$ odd. $\pf^f_{q,p,n}$ is defined as 
\begin{equation}
    \pf^f_{q,p,n}=\frac{{\rm Spin}(n)_1\boxtimes \cA^{4q,p}}{(\psi a_{2q} c)}\,.
\end{equation}
 We label the anyons in $\pf^f_{q,p,n}$ as $x_j$, where $x=1,\sigma,\psi$ and $j=0,1,\dots,4q-1$. To have trivial braiding with $\psi a_{2q}$, we impose the following rule: if $x=1$ or $\psi$, $j$ must be even. If $x=\sigma$, $j$ must be odd.

One particularly special case comes when $p=1,q=2$. Then $n=1$ [so Spin(1)$_1={\rm Ising}$] gives the Moore-Read state. $n=7$ (alternatively, $n=-1$) gives the T-Pfaffian state (also called PH-Pfaffian in the context of half-filled Landau level).

Other authors denote the Abelian anyons as $I_0, I_2,\ldots, I_{4q-2}, \psi_0, \psi_2, \ldots, \psi_{4q-2}$ and the non-Abelian anyons as $\sigma_{1},\sigma_{3},\ldots,\sigma_{4q-1}$. Their fusion rules are given as Eq.~(20) in \cite{bonderson_time-reversal_2013} for the case of $q=2$; these generalize to other even $q$ in a straightforward manner.

Let us review a field theory description of the Moore-Read state $\pf_{2,1,1}^f$.
Following \cite{seiberg2016gapped}, we start with  ${\rm U}(2)_{2,-4}\times \U_{4q}$ Chern-Simons theory,
\begin{equation}
    \frac{2}{4\pi}\Tr (c\wedge dc+\frac23 c\wedge c \wedge c)-\frac{2}{4\pi}(\Tr c)\wedge d(\Tr c) + \frac{4q}{4\pi}b\wedge db  + \frac{2}{2\pi}A\wedge db.
\end{equation}

Then we gauge the diagonal $\bZ_2^{(1)}$. To do this, we make the following change of variables:
\begin{equation}
    \tilde{c}=c-b\mathbf{1}\quad,\quad \tilde{b}=2b\,.
\end{equation}
These combinations were chosen such that the $\bZ_2^{(1)}$ one-form global symmetry does not act
on them. Using the new variables, the action becomes 
\begin{equation}\label{tildebtildec}
    \frac{2}{4\pi}\Tr (\tilde{c}\wedge d\tilde{c}+\frac23 \tilde{c}\wedge \tilde c \wedge \tilde c)-\frac{2}{4\pi}(\Tr \tilde{c})\wedge d(\Tr \tilde{c}) + \frac{q-1}{4\pi}\tilde{b}\wedge d\tilde{b} -{1\over 2\pi} (\Tr \tilde c) \wedge d\tilde b + \frac{1}{2\pi}A\wedge d\tilde{b}.
\end{equation}

We readily recognize the Wilson line $e^{i\oint \tilde b}$ as the vison. It generates a $\mathbb{Z}_{2q}^{(1)}$ one-form symmetry with anomaly $1$ and hence the symmetry lines are ${\cal V}^{2q,2}$.

\section{More general coupling of bosonic TQFT to background $\U$}
\label{app:generator}

We start this appendix by showing that what we did in Sec.~\ref{from response to} was essentially the only way to reproduce the correct response theory. In doing so, we will see the role played by some integers $u_I$, which define the coupling of the theory to the background $\U$ gauge field. Later, we will work out some relations that allow us to change these $u_I$ without changing the TQFT. This generalizes the discussion about changing the generator of the one-form symmetry in Sec.~\ref{sec:symmetrylines}.  For simplicity, we will limit ourselves to bosonic theories.  The extension to electronic theories is straightforward.

\subsection{Generalizing Sec.~\ref{from response to}}

Our goal is to characterize bosonic TQFTs ${\cal T}$ such that upon integrating them out, they lead to the response theory \eqref{fourda},
\begin{equation}
\frac{\sigma_H}{4\pi}\int_{\fmfd} dA\wedge dA.
\end{equation}

Let us start with a generic TQFT $\cal T$.  A subset of its anyons ${\cal A}\subseteq {\cal T}$ are Abelian. The set $\cal A$ is an Abelian group $\otimes_I \mathbb{Z}_{s_I}^{(1)}$, which is the one-form global symmetry of $\cal T$.  Given such a symmetry, it is common to couple it to background two-form gauge fields ${\cal B}_I\in H^2(M,\mathbb{Z}_{s_I})$.  This coupling has an anomaly given by a quadratic expression in ${\cal B}_I$, similar to \eqref{BB}.  Then, we couple our system to $A$ in way that generalizes \eqref{BdA}
\begin{equation}
    {\cal B}_I= \left[{u_IdA\over 2\pi}\right]_{s_I}\,.
    \label{eqn:BI}
\end{equation}
This coupling depends on the integers $u_I$.

As we vary $A$, the background fields ${\cal B}_I$ vary and select a $\mathbb{Z}_s^{(1)}$ subgroup of the one-form symmetry group  $\otimes_I \mathbb{Z}_{s_I}^{(1)}$.
Specifically, let $X_I$ be the generator of $\mathbb{Z}_{s_I}$.  Then, the fact that all ${\cal B}_I$ are determined by the same $dA$, means that the group elements that $A$ couples to are powers of $X=\prod_I X_I^{u_I}$, which satisfies $X^s=1$ with $s$ the smallest nonzero solution of $s u_I =0\mod s_I$ for all $I$.  Furthermore, the background $\cal B$ that couples to this $X$ satisfies, as in \eqref{BdA}, ${\cal B}= \left[dA/2\pi\right]_{s}$.  

Therefore, we recognize the special construction in Sec.~\ref{from response to}, where we assumed that our TQFT possessed a one-form $\bb{Z}_s^{(1)}$ symmetry, leading to \eqref{dAdA} as the essence of the more general case with arbitrary $u_I$.

In more physical terms, we recognize $X$ as associated with the vison $v\in {\cal A}$. Then, the anomaly expression for general ${\cal B}_{I}$ determines $r$ for $\cal B$ and therefore $\sigma_H=r/s$. The vison generates the Abelian set of anyons ${\cal V}^{s,r}$ discussed in Sec.~\ref{sec:symmetrylines}.

Suppose, however, that we did not want to single out the vison as the generator of the $\bb{Z}_s^{(1)}$ one-form symmetry. Then, motivated by  \eqref{eqn:BI}, we can pick a (possibly larger) $\bb{Z}_s$ symmetry with an arbitrary $u$
\begin{equation}
 {\cal B} = \left[{u dA\over 2\pi}\right]_{s}\,.
\end{equation}
Inserting this into  \eqref{BB} reveals that this gives us the correct response theory if
\begin{equation}
\frac{ru^2}{s} = \sigma_H =  \frac{p}{q} \mod 2\,.
\end{equation}

This theory describes a slightly larger set of anyons ${\cal A}^{s,r}_u $ whose lines are given by $a_j$. They have spins and $\U$ charges
\begin{equation}
\begin{split}
    &a_j=a_1^j\qquad,\qquad j=0,1\cdots, s-1\\
    &h(a_j)={rj^2\over 2s}\mod 1\\
    &Q(a_j)={ruj\over s}\mod 1\,.
\end{split}
\end{equation}
The vison discussed above corresponds to
\begin{equation}
\begin{split}
    &v=a_1^u\\
    &h(v)={ru^2\over 2s} = \frac{p}{2q}\mod 1\\
    &Q(v)={ru^2\over s} = \frac{p}{q}\mod 1\,.
\end{split}
\end{equation}
It generates a $\mathbb{Z}_{s'}^{(1)}\subseteq \mathbb{Z}_s^{(1)}$ with $s'={s/\gcd(s,u)}$ one-form symmetry with anomaly  $r'={ru^2/\gcd(s,u)}$ and $u$-parameter $1$, i.e., ${\cal V}^{s',r'}$. We see that 
\begin{equation}
{\cal V}^{s',r'}\subseteq {\cal A}^{s,r}_u \subseteq {\cal A}\,.
\end{equation}

Let us demonstrate it in a simple example with $\sigma_H={1/2}$.  We need to have ${r'/s'}={ru^2/s}={1/2}$.  The simplest option is ${\cal A}^{2,1}_1\cong {\cal V}^{2,1}$.  This is $\U_2$ with the vison $v=a_1$ with $h={1/4}$ and $Q={1/2}$.  Alternatively, we can have ${\cal A}^{8,1}_2$, i.e.,  $\U_8$ with the vison $v=a_1^2$ also with $h={1/4}$ and $Q={1/2}$.  Either way, the vison generates ${\cal V}^{4,2}$.

\subsection{Identifications of ${\cal A}^{s,r}_u$}

We now want to discuss the identifications in ${\cal A}^{s,r}_u$. First, as in \eqref{2sshift}, the results depend only on $r\mod  2s$.   Similarly, they depend only on  $u\mod s$.  
Therefore,
\begin{equation}\label{bcident}
\begin{split}
&{\cal A}_u^{s,r}\cong {\cal A}^{s,r+2s}_{ u}\\
&{\cal A}_u^{s,r}\cong {\cal A}^{s,r}_{ u+s}\,.
\end{split}
\end{equation}
The first relation is the same as in the bosonic theory before we coupled to $A$ \eqref{2sshift}.  The second relation reflects the fact that $u$ labels the vison and our labeling of anyons is by $j\sim j+s$.\footnote{We could have replaced the second identification with a stronger version $\cA^{s,r}_{u}\cong\cA^{s,r}_{u+{s/\ell}}$.  We did not do it because even though this does not change $h$ and $Q$ of the anyons in $\cA^{s,r}_{u}$, it does change the values of $Q$ of  other anyons in $\cal T$.\label{strongeriden}}

In addition, as in \eqref{changemg}, we can also pick a different generator for the $\mathbb{Z}_s^{(1)}$ symmetry.  For every integer $m$ such that $\gcd(s,m)=1$, we can take $a_m$ to be the generator.  As in \eqref{changemg}, this change has the effect of mapping $r\to m^2 r$.  The corresponding change in $u$ is $u\to \hat m u$, where $\hat m$ is the inverse of $m$ modulo $s$,  i.e., the solution of
\begin{equation}\label{mmhat}
\hat m m+\hat s s=1 \qquad, \qquad \hat m,\hat s \in \mathbb{Z}\,.
\end{equation}
Therefore, taking $a_m$ as the generator, maps
\begin{equation}\label{changegb}
    \begin{split}
        &j\to j \hat m \\
        &r\to r m^2 \\
        &u\to u \hat m\,
    \end{split}
\end{equation}
and the relation \eqref{changemg} is extended to
\begin{equation}\label{changemge}
{\cal A}_u^{s,r}\cong {\cal A}^{s,m^2r}_{\hat m u}\qquad , \qquad \gcd(m,s)=1\,,
\end{equation}
which preserves the fact that $rs$ is even.  Note also that $m$ and $\hat m$ are defined by \eqref{mmhat} only modulo $s$.  This fact is consistent with this identification.

As checks, the spins and the charges of the lines
\begin{equation}
\begin{split}
&h(a_j)={rj^2\over 2 s}\mod 1\\
&Q(a_j)={ruj\over s}\mod 1
\end{split}
\end{equation}
as well as $\sigma_H=ru^2/s\mod  2$ are invariant under the identifications \eqref{bcident}, \eqref{changegb}, and \eqref{changemge}. (When checking them, recall that $rs$ is even.)

In the special case of $\gcd(s,u)=1$, the vison $v=a_u$ generates all the symmetry lines.  Therefore, it is natural to take it as the generator.  This is achieved by the identification \eqref{changemge} with $m=u$
\begin{equation}
{\cal A}^{s,r}_u \cong {\cal 
V}^{{s},{ru^2}} \qquad, \qquad \gcd(s,u)=1\,.
\end{equation}
Here we used
${\cal V}^{s,r}={\cal A}^{s,r}_1$.

When $\gcd(s,u)\ne 1$, we can consider the subset of ${\cal A}^{s,r}_u$ that is generated by the vison $v=a_u$. 
As discussed above, this is given by ${\cal V}^{s',r'}\subseteq {\cal A}^{s,r}_u$ with $s' = s/\gcd(s,u)$ and $r' = ru^2/\gcd(s,u)$.

For this reason, in most of our discussions, we focus on the set of symmetry lines with $u=1$.  Then, the only identification we should consider is
\begin{equation}\label{upluss}
    {\cal V}^{s,r}\cong {\cal V}^{s,r+2s}\,.
\end{equation}

\section{Algebraic classification of $\U$ SETs}
\label{app:spinc-classification}

In this appendix, we will classify spin$^{c}$ TQFTs, that is, fermionic/spin TQFTs enriched by $\U$ that obey the spin/charge relation. Along the way, we will review the classification of bosonic TQFTs enriched by $\U$, as well as classify fermionic/spin TQFTs enriched by $\U$ which do not satisfy the spin/charge relation. For a brief introduction, we refer the reader to Sec.~\ref{sec:tqft-basics}.  A general classification of spin TQFTs enriched by global symmetry has been developed in \cite{Bulmash:2021hmb}. Discussions of various features of spin$^c$ TQFTs are given in \cite{kobayashi_31-dimensional_2024}, although we will be self-contained in this appendix.

\subsection{Necessary definitions and results}\label{Necessary definitions}

We will extensively use the formalism of modular tensor categories (MTC) in this appendix, which are the formal mathematical objects describing bosonic TQFTs. As such, we use the terms MTC and bosonic TQFT synonymously. For reviews of the mathematical theory, see \cite{kitaev_anyons_2006, bonderson_non-Abelian_2007}. 

Some key properties are summarized in the following theorems. Along the way, we supplement the theorems with examples and remarks to add some insight.

\begin{theorem}\label{thm:TBone}
     In an MTC, ${\cal T}$, suppose each anyon type $x \in {\cal T}$ is associated with a phase factor $e^{i\phi(x)}$, such that
     \begin{equation}
     e^{i\phi(x)}e^{i\phi(y)}=e^{i\phi(z)}\qquad, \qquad N_{xy}^z>0\,.
        \label{TBone-eq}
     \end{equation}
     Then there exists a unique Abelian anyon $w$ such that 
     \begin{equation}
         e^{i\phi(x)}=e^{2\pi i B(x,w)}\qquad ,\qquad \forall x\,.
     \end{equation}
\end{theorem}
\begin{proof}
The proof can be found around Eq.~(47) of \cite{barkeshli_symmetry_2019}.
\end{proof}

Another way of phrasing the theorem is that the phase factors in \eqref{TBone-eq} are the eigenvalues of a one-form global symmetry, and $w$ is its symmetry line \cite{Gaiotto:2014kfa,Hsin:2018vcg}.

\begin{theorem}
If an MTC $\cal{T}$ has a subcategory $\cal{T}'$ which is also modular, then $\cal{T}$ factorizes as ${\cal T}={\cal T}'\boxtimes {\cal T}''$, where ${\cal T}''$ is another MTC.
    \label{mtc-factorization}
\end{theorem}
\begin{proof}
This is Theorem 3.13 in \cite{drinfeld_braided_2010}.
\end{proof}

\begin{definition}
A super-modular category \cite{Bruillard2017}\footnote{Other terminologies exist. For example, it was called ``slightly degenerate braided fusion category" in \cite{Johnson-Freyd_Reutter_2024}.} is a premodular tensor category whose transparent anyons are $\{1,c\}$. Here $h(c) = 1/2$ and $c^2=1$, so $c$ is a fermion that squares to the identity.
\end{definition}

\begin{remark}
MTCs describe line operators in a bosonic TQFT. Super-modular categories describe line operators in a spin TQFT. In this case, it may appear that one should identify $1$ and $c$, and reduce the super-modular category to an even smaller theory with half the number of anyons. For example, the Moore-Read state can be described by a super-modular category with 12 anyon types (including $c$). After identifying $c$ and $1$, one can write down a theory of six anyons. However, it is known that this theory does not admit any consistent braiding~\cite{bonderson_non-Abelian_2007}, i.e., it is merely a unitary fusion category, but not braided. Thus, we will often keep $c$ as an element of ${\cal T}$ throughout this appendix.
\end{remark}

\begin{definition}
A super-modular category splits if it can be written as ${\cal M}\boxtimes \{1,c\}$, where $\cal M$ is an MTC.
\end{definition}

Note that every Abelian super-modular category splits. This is proven in the appendices of \cite{cheng_fermionic_2019, ma_fractonic_2022}. A similar observation was made in  Sec.~\ref{sec:symmetrylines} for spin TQFTs $\cA^{s,r}$ when $sr$ is odd: we showed that it is equivalent to $\cA^{s,4r}\boxtimes \{1,c\}$, where $\cA^{s,4r}$ is bosonic. 

\begin{theorem} \label{thm:TBone_smtc}
    In a super-modular category, if each anyon type is associated with a phase factor $e^{i\phi(x)}$, with $e^{i\phi(c)}=1$, and satisfies \eqref{TBone-eq}, then there exists an Abelian anyon $w$ such that 
    \begin{equation}
        e^{i\phi(x)}=e^{2\pi i B(x,w)}\qquad,\qquad \forall x\,.
        \label{braiding-gen-char-2}
    \end{equation}
    The anyon $w$ is uniquely determined up to fusing with $c$.
\end{theorem}
\begin{proof}
This is a consequence of  (2.3) of  \cite{bruillard_classification_2017}.
\end{proof}

\begin{theorem} \label{thm:smtc-factorization}
If a super-modular category $\cal{T}$ has a modular subcategory $\cal{T}'$, then $\cal{T}$ factorizes as ${\cal T}={\cal T}'\boxtimes {\cal T}''$, where ${\cal T''}$ is a super-modular category.
\end{theorem}
\begin{proof}
This is Lemma 2.5 in \cite{bruillard_classification_2017}.
\end{proof}

It will be crucial for our later discussion to consider the minimal modular extension ${\cal T}^{\rm ext}$ of a super-modular category $\cal T$. 
The minimal modular extension is a minimal (by total quantum dimension) MTC $\cal{T}^{\rm ext}$ that contains $\cal{T}$ as a subcategory. It is $\bZ_2$ graded, ${\cal T}^{\rm ext}={\cal T}_0^{\rm ext}\oplus {\cal T}_1^{\rm ext}$, where ${\cal T}_0^{\rm ext}$ contains all anyons braiding trivially with $c$, and ${\cal T}_1^{\rm ext}$ contains all anyons which have $-1$ braiding with $c$. Moreover, ${\cal T}_0^{\rm ext} = {\cal T}$. Then using Lemma \ref{lemma:G-graded-D}, the $\bZ_2$ grading implies that $\mathcal{D}_{\cal{T}^{\rm ext}}=\sqrt{2}\mathcal{D}_{\cal{T}}$. Mathematically, the existence of a minimal modular extension for any super-modular category was recently established in
 \cite{Johnson-Freyd_Reutter_2024}.

We will sometimes refer to anyons in $\cT^{\rm ext}_1$ as ``Ramond anyons/lines."\footnote{These lines correspond to states in the Hilbert space of the theory on a 2D torus with spin structures $(+,+)$ and $(+,-)$.  Here, the two signs represent the boundary conditions of spinors around the two cycles of the torus.  $+$ means that they are periodic (Ramond or R in the string theory terminology), and $-$ means that they are anti-periodic (Neveu–Schwarz or NS in the string theory terminology).  The first sign corresponds to the boundary conditions as the spinors circle around the line, and the second sign corresponds to the boundary conditions along the line.}  We will denote them by bold fonts, e.g., $\bm {a}$.  Ramond anyons come in two types: Majorana and non-Majorana. An anyon $\bm{a} \in \cT^{\rm ext}_1$ is Majorana if $c\times \bm{a}=\bm{a}$, otherwise it is non-Majorana. The Majorana anyons correspond to states with odd fermion parity on the torus in the $(+,+)$ sector, which will be discussed further in  Appendix~\ref{HilbertspaceT}. 

As a TQFT, a minimal modular extension can be obtained from gauging the fermion parity, or summing over spin structures in the spin TQFT. It is often referred to as a ``bosonic shadow" of the spin TQFT \cite{Gaiotto:2015zta}.

Without gauging/summing over the spin structure, the Ramond ``anyons" are not actual anyonic excitations in the fermionic topological order because they have $-1$ braiding with the fermions. In other words, they are extrinsic objects carrying $\pi$ fluxes (in the sense of Aharonov-Bohm phase). A concrete example is the vortices in an (electronic) superconductor.  In the TQFT, they are line operators attached to fermion parity surfaces. After gauging, they become genuine line operators of the bosonic shadow.\footnote{Here we follow the terminology of \cite{Kapustin:2014gua}, according to which nongenuine lines are lines attached to topological surfaces.  The nongenuine lines, like this Ramond line, are higher-dimensional versions of the familiar disorder operator in 1+1d.  The latter is a non-genuine point operator, because it is attached to a topological line, which is the symmetry line of a global $\mathbb{Z}_2$ symmetry.}

Here we use the minimal modular extension mainly as a device to describe the fusion and braiding properties of the $\pi$-flux/Ramond objects, which will prove to be convenient for the classification of spin{}$^c$ theories. The reason behind the effectiveness of this approach is that the $\pi$-flux objects can be obtained through flux threading, which provides a continuous interpolation between $\pi$-flux objects and anyons. Equivalently, as will be discussed in the following subsection, different spin structures on the torus (which are labeled by objects in the modular extension) are essentially the same in a spin{}$^c$ theory because one can continuously change the boundary conditions to interpolate between them.

\begin{remark} \label{rem:all_shadows}
For any $\cal{T}$, there are 16 different minimal modular extensions, which have different chiral central charges (mod 8). The data of the minimal modular extensions, such as the $S$ and $T$ matrices, is completely determined by the data of the super-modular category $\cT$ and the integer $n$ below.   If we have one minimal modular extension ${\cal T}^{\rm ext}$, then all 16 of them can be obtained by
\begin{equation} \label{eqn:all_shadows}
{\cal T} = \frac{{\rm Spin}(n)_1 \boxtimes {\cal T}^{\rm ext}}{(\psi c)} \qquad,\qquad n = 0,\ldots, 15\,.
\end{equation}
This can be understood as stacking $n$ copies of $p+ip$ superconductors before gauging.
Equivalently, as a TQFT, different choices of the modular extension differ by phase factors in the sum over spin structures (see Appendix~C.5 in \cite{seiberg2016gapped}).\footnote{In various contexts, the freedom in such phases in the process of gauging is known as discrete torsion or as adding Dijkgraaf-Witten terms.  Equivalently, they correspond to adding local counterterms or tensoring an invertible field theory before the gauging.}

\end{remark}

Let us discuss some examples.

\begin{enumerate}
    \item The simplest super-modular category is $\{1,c\}$. The minimal modular extensions are given by the 16-fold way, ${\rm Spin}(n)_1$ for $n=0,1,\cdots, 15$. They can be thought of as gauging $n$ copies of $p+ip$ superconductors. If $n$ is odd, the theory has a single Ramond anyon, which is Majorana.
    \item ${\cal F}_{4m+2}$ can be extended to SU(2)$_{4m+2}$. The isospin-$\left(m+\frac12\right)$ anyon is a Majorana anyon.
    \item The Moore-Read theory can be extended to ${\rm Ising}\boxtimes \U_8$. All Ramond anyons are non-Majorana. 
\end{enumerate}

\subsection{The Hilbert space of the TQFT on a two-torus}\label{HilbertspaceT}

Consider the TQFT on a three-torus.  It has eight spin structures labeled in an obvious way as $(\pm,\pm,\pm)$.  The $(+,+,+) $ spin structure is referred to as odd, and the other seven spin structures are even.  We can change the cycles of the torus using SL$(3,\bZ)$ transformations.  This has the effect of permuting the even spin structures, while leaving the odd spin structure invariant.  

In a spin TQFT, the partition function depends on the spin structure.  We will denote it as  ${\cal Z}_{(\pm,\pm,\pm)}$.  Using  SL$(3,\bZ)$ transformations ${\cal Z}_{(\pm,\pm,\pm)}$ can have only two different values.  The partition functions for the even spin structures are $N_e$ and ${\cal Z}_{(+,+,+)}=N_o$.

These partition functions can be interpreted as traces over the Hilbert space of a spatial two-torus. The four spin structures of the spatial torus are labeled by $(\pm,\pm)$, and their Hilbert spaces are ${\cal H}_{(\pm,\pm)}$.  Interpreting the first index in ${\cal Z}_{(\pm,\pm,\pm)}$ as the spin structure around the compact Euclidean time and the other ones as associated with the spatial cycles, we have
\begin{equation}
\begin{split}
&{\cal Z}_{(-,\pm,\pm)}=\Tr_{{\cal H}_{(\pm,\pm)}} 1=N_e\\
&{\cal Z}_{(+,-,+)}=\Tr_{{\cal H}_{(-,+)}}(-1)^F =N_e\\
&{\cal Z}_{(+,+,-)}=\Tr_{{\cal H}_{(+,-)}}(-1)^F =N_e\\
&{\cal Z}_{(+,-,-)}=\Tr_{{\cal H}_{(-,-)}}(-1)^F =N_e\\
&{\cal Z}_{(+,+,+)}=\Tr_{{\cal H_{(+,+)}}}(-1)^F =N_o\,,
\end{split}
\end{equation}
where $(-1)^F$ is the fermion parity.  We immediately learn some facts about these four Hilbert spaces ${\cal H}_{(\pm,\pm)}$: 
\begin{itemize}
    \item The dimensions of ${\cal H}_{(\pm,\pm)}$ are the same and are given by the  positive integer $N_e$.  
    \item All the states in the Hilbert spaces of the even spatial spin structures have $(-1)^F=+1$.  
    \item The Hilbert space of the odd spatial spin structure ${\cal H}_{(+,+)}$ has $N_e+N_o\over 2$ states with $(-1)^F=+1$ and  $N_e-N_o\over 2$ states with $(-1)^F=-1$.  Therefore, $N_o$ must be an  integer satisfying $N_o=N_e\mod 2$ and $|N_o|\le N_e$.
\end{itemize}

 These integers are related to the data of the bosonic shadow \cite{Delmastro:2021xox}. (See Appendix~\ref{Necessary definitions}.) Denote the number of anyon types (including $c$) in the super-modular category $\cT$ by $|\cT|$, and similarly define $|\cT^{\rm ext}_{0,1}|$. Further we denote the number of Majorana-type anyons in $\cT^{\rm ext}_1$ as $|\cT^{\rm ext}_{\rm M}|$. Then we have 
\begin{equation}\label{numberof states}
    N_e=\frac12 |\cT|, \quad N_o=\frac12 |\cT|-2|\cT^{\rm ext}_{\rm M}|.
\end{equation}
(Note that since $\cT$ is a super-modular category and it includes $c$, $|\cT|$ must be even.) In other words, the number of $(-1)^F=-1$ states in ${\cal H}_{(+,+)}$ is $|\cT^{\rm ext}_{\rm M}|$. 

One way to understand this is the following.  Observe that $(-1)^F=-1$ means that there is a single $c$ defect on the torus. Therefore, such a state is obtained by inserting a defect of Majorana-type anyon in the bosonic shadow inside the solid torus and performing the path integral.  This explains why the number of such states is equal to $|\cT^{\rm ext}_{\rm M}|$. It is worth noting that the TQFT states in the $(+,+)$ sector depend to some extent on the choice of the minimal modular extension, or equivalently the chiral central charge [namely, for modular extensions given in \eqref{eqn:all_shadows}, $N_o$ depends only on the even/oddness of $n$].

If the TQFT is also a spin$^c$ theory, the previous discussion is still valid, but we can also turn on background $A$. That background allows us to interpolate continuously between the distinct spin structures without $A$ we discussed above.  Let us turn on only spatial $A$.  The expressions above in terms of traces are still valid.  And since they are given by integers, they cannot vary with $A$.  Therefore, $N_e=N_o$.  This means that none of the four Hilbert spaces ${\cal H}_{(\pm,\pm)}$ have any states with $(-1)^F=-1$. In particular, it implies that there are no Majorana anyons.  A simple way to understand it follows from the definition of  Majorana anyons as anyons $\bm{a}$, such that $c\times \bm{a}=\bm{a}$.  In a spin$^c$ theory, $c$ carries charge one under $\U$ and therefore there cannot be any such $\bm{a}$. We will discuss it in more detail in  Sec.~\ref{secondderiv}.

For example, $\U_k$ with odd $k$ is a spin$^c$ theory with $N_e=N_o=k$. It has $k$ states on a spatial two-torus for each spin structure, all with $(-1)^F=+1$ \cite{Delmastro:2021xox}. On the other hand, ${\rm SO}(3)_k={{\rm SU}(2)_k\over \bZ_2}$ with $k=2\mod 4$ is a spin theory, but it is not a spin$^c$ theory.  It has $N_e={k+2\over 4}$ and $N_o={k-6\over 4}$.  This corresponds to ${k+2\over 4}$ states for each spin structure on a spatial two-torus and only one state with $(-1)^F=-1$ in the $(+,+)$ spin structure \cite{Delmastro:2021xox}.

The key fact we used about spin$^c$ theories is that, unlike ordinary (non-spin$^c$) fermionic theories, in these cases, we can interpolate continuously between the different spin structures.  This continuous interpolation is a generalization of the flux threading argument, reviewed in Sec.~\ref{flux threading}.  The original $2\pi$-flux threading maps the theory back to itself, perhaps with some nontrivial action on the Hilbert space.  In spin$^c$ theories, we can extend it by threading $\pi$ flux.  It does not map the theory back to itself, but it changes the spatial spin structure, thus relating the Hilbert spaces with different spin structures.

Using the same reasoning, we can obtain a stronger result.  In spin$^c$ theories, not only is the dimension of the Hilbert space in all four spatial spin structures the same, $N_e$, there is also a one-to-one map between these states, obtained by tracking them during the flux-threading process.  Below, we will discuss this one-to-one map in more detail and will also see that this map is not unique.  

So far, we have been studying the Hilbert space on the torus.  Next, we consider a torus with an insertion of a defect $x$.  Again, we can find different states depending on the spin structure.  The states on the torus with a defect $x$ are in one-to-one correspondence with anyons $a$ such that the fusion of $x$ and $a$ includes $a$. The one-form symmetry of the problem leads to a selection rule stating that $x$ should be invariant under the one-form symmetry. 
For example, the $\U_k$ theory has a $\bZ_k^{(1)}$ one-form symmetry, and therefore, any nontrivial defect is charged under it, and the torus Hilbert space in the presence of such a defect is empty. The same applies to all Abelian TQFTs. However, the one-form symmetry of the Moore-Read state is $\bZ_4^{(1)}$, and the $v^2c=\psi_0$ anyon is invariant under it.  Consequently, a torus with an $x=v^2c$ defect has a non-empty Hilbert space.  In particular, it is two-dimensional, and the two states correspond to the two non-Abelian anyons of this model $a$, for which $x\times a$ includes $a$.

Both kinds of states on the torus, with or without a defect, are useful as the TQFT description of the ground states of FQH systems on a torus, which are routinely studied in finite-size numerical simulations. In general, one can study the system with any number of electrons ${\cal N}$ and any number of flux quanta ${\cal K}$ (as long as they are sufficiently large so that the thermodynamic limit is a valid approximation). For a FQH phase at filling fraction $\nu$, in order for the ground states on a torus to be described by the TQFT (possibly with a defect insertion), it is necessary that ${\cal N}/{\cal K}=\nu$. Otherwise, if the deviation ${\cal K}-\frac{1}{\nu}{\cal N}$ is nonzero, but of order one, one expects that in general the system is in an excited state with an order one number of anyon excitations, thus outside the realm of TQFT. Below we focus on the case when the filling condition is satisfied, with filling factor $\nu=p/q$.

The precise relationship between the microscopic FQH and the TQFT Hilbert space can be understood as follows: since there are ${\cal K}$ units of flux through the torus, and each flux corresponds to a $v$ defect, the FQH ground state should be associated with the TQFT with the $v^{\cal K}c^{\cal N}$ defect inserted.  

First, we consider the case of theories where $\ell=1$.  Here, we have $v^q=c^p$. Since ${\cal N}/{\cal K}=p/q$, we conclude that $v^{\cal K}c^{\cal N}=c^{2p}=1$ and the ground states on the torus always correspond to TQFT states without any defect.

Next, we consider the general case where $\ell>1$, e.g., the minimal theories for even $q$. Then we have $v^{\ell q}=c^{p\ell}$, or $(v^q c^p)^{\ell}=1$, and ${\cal K}=q\frac{{\cal N}}{p}$. Thus, the inserted defect is 
\begin{equation}
    v^{\cal K}c^{\cal N}=v^{q{\cal N}/p}c^{\cal N}=(v^qc^p)^{{\cal N}/{p}\mod \ell}. 
\end{equation}
That is, there is a defect on the torus depending on ${\cal N}/p\mod \ell$. So only when ${\cal N}$ is a multiple of $p\ell$, do the microscopic FQH ground states correspond to TQFT states with no defect insertions.

To illustrate, let us now specialize to the minimal theories $\pf^f_{q,p=1,n}$ with $\ell=2$ for $q$ even. The result now depends on whether $n$ is even or odd:
\begin{itemize}
    \item Odd $n$: If ${\cal N}$ is even, there is no defect insertion and the number of ground states is $3q$ for each spin structure. For ${\cal N}$ odd, the TQFT has a defect $v^qc=\psi_0$. There are $q$ lines that remain the same under fusing with $v^qc$. Using the notations in App.~\ref{PfaffianTO}, they are $\sigma_k$ with $k=1,3,\dots, 2q-1$.   Thus, we conclude that when ${\cal N}$  is odd (while maintaining the filling fraction), there are $q$ states for each spin structure. 
    \item Even $n$: Since the theory is Abelian, the Hilbert space with a nontrivial defect is empty. Therefore, in this case, the system cannot be described by states in the TQFT Hilbert space when ${\cal N}$ is odd.  If ${\cal N}$ is even, there are $4q$ states.
\end{itemize}

Finally, we notice that the same argument applies to topological phases in lattice systems at a given (rational) filling. Assuming that microscopic translations do not permute the anyons of the macroscopic theory, the microscopic ground states correspond to TQFT states with an $a_b^{N_{\rm uc}}$ defect, where $N_{\rm uc}$ is the total number of unit cells and $a_b$ is the background anyon.

\subsection{Classifying charge assignments for bosonic and spin TQFTs}
\label{charge-assignment-classification}

First, we review the classification of bosonic and fermionic/spin (which are not necessarily spin$^c$) $\U$ SETs. Here, by ``classification," we mean classifying $\U$ SETs given a certain topological order (i.e., anyon data and $c_-$). We will first classify the $\U$ symmetry action on the anyons, which is equivalent to assigning charges consistently to anyons.  These charges are defined modulo one. Then, the theory is essentially fixed, up to stacking invertible theories.  Later, we will use this information and the value of $\sigma_H$ to constrain the TQFT.

\begin{theorem} \label{thm:charge_assign_boson}
The charge assignments of a bosonic TQFT are classified by $\cA$. The charge assignments in a spin TQFT with unit charge bosons and fermions are classified by $\cA/\{1,c\}$.
\end{theorem}
\begin{proof}
In both of these theories, the fact that charges are defined modulo one means that
\begin{equation}\label{Qselec}
    Q(x)+Q(y) = Q(z) \mod {1}\qquad,\qquad N_{xy}^z>0\,.
\end{equation}
for anyon types $x,y,z$. In other words, $e^{2\pi iQ(x)}e^{2\pi iQ(y)}=e^{2\pi i Q(z)}$. 

Specialize first to the bosonic case. Then we can use Theorem~\ref{thm:TBone} to see that there exists a unique Abelian anyon $v$ such that 
\begin{equation}
    B(v,x) = Q(x) \mod {1}.
\end{equation}
Since this anyon detects the charges of all other anyons by braiding with them, it must be the quasiparticle created by threading $2\pi$ flux. We henceforth refer to it as the vison. The different charge assignments are therefore parametrized by the different choices of an Abelian vison, i.e., by $\cA$.

In the fermionic/spin case Theorem~\ref{thm:TBone_smtc} means that there exists an Abelian anyon $v$ such that 
\begin{equation}
    B(v,x) = Q(x) \mod {1},
\end{equation}
and $v$ is unique up to fusion with $c$. Thus, the different charge assignments are parametrized by the Abelian anyons in $\cA/\{1,c\}$.
\end{proof}

In the bosonic case, this is an MTC way of understanding the results we have found in App.~\ref{app:generator}.

If the Hall conductivity is fixed as $\sigma_H$, then only some choices of the vison are acceptable. These correspond to all $v\in \cA$ satisfying
\begin{equation}
    h(v)=\frac{\sigma_H}{2}\mod 1\,.
\end{equation}
Physically, the even integer ambiguity in $\sigma_H$ corresponds to the stacking of bosonic SPT phases.

\subsection{From spin TQFT to spin$^c$ TQFTs}

Here we consider the classification of fermionic systems that satisfy the standard spin/charge relation, or spin$^c$. Again, we focus on classifying the charge assignment of anyons, after which the remaining choices are just stacking invertible phases. Equivalently, we are classifying distinct ways to couple a spin TQFT to $\U$ so that it becomes a spin$^c$ TQFT. 

\begin{definition} \label{defn:spinc}
A super-modular category $\cal T$ is spin{}$^c$ if it can be endowed with the following data. Each anyon $x$ is associated with a charge $Q(x)$ defined mod 2, subject to the following relations:
\begin{equation}
    \begin{split}
        Q(1)&= 0 \mod {2}\\
        Q(c)&= 1 \mod {2}\\
    Q(x)+Q(y)-Q(z)&= 0\mod {2}\qquad,\qquad \text{ if } N_{xy}^z>0\,.
    \end{split}
\end{equation}
We will sometimes call such a $\cal T$ a spin$^c$ category.
\end{definition} 

We will present two derivations of the assignment of charges modulo 2 and will explore the properties of the spin$^c$ category. Our first approach will be to try generalizing the bosonic approach of Appendix~\ref{charge-assignment-classification}. We will see that this approach falls short of a full classification, motivating us to employ the bosonic shadow.

\subsubsection{First derivation: Generalizing the bosonic case}\label{first deriv}

We start by following the logic of Theorem~\ref{thm:charge_assign_boson}, using the consistency of  $Q(x)\mod 1$ to  expressed it in terms of the vison
\begin{equation}
    Q(x)=B(v,x)\mod 1\,.
\end{equation}
Note that here Theorem~\ref{thm:TBone_smtc} means the charge assignment only fixes $v$ up to $c$. 

Now we need to further determine $Q(x)\mod 2$. Let us choose an arbitrary lifting of $B$ mod 2, call it $\tilde{B}$. Then write 
\begin{equation}
    Q(x)=\left(\tilde{B}(x,v)+t(x)\right)\mod 2\,.
\end{equation}
Here $t(x)\in \{0,1\}\mod 2$.
We need to set $\tilde{B}(1,x)=\tilde{B}(c,x)=0$ mod 2, and $t(1)=0$. Then we have
\begin{equation}
    t(x)+t(y)-t(z)=\left(\tilde{B}(v,z)-\tilde{B}(v,x)-\tilde{B}(v,y)\right)\mod 2 \qquad,\qquad \text{ if } N_{xy}^z>0\,.
    \label{t-fusion}
\end{equation}
Any solution to these equations gives a consistent charge assignment. As we will see, not every $v$ allows solutions. Those choices of $v$ are not compatible with the spin$^c$ condition.

Assuming that \eqref{t-fusion} allows solutions, we can study the space of solutions. Start from one solution, say $t_0$, and write $t=t_0+\xi\mod 2$. Then $\xi$ should satisfy
\begin{equation}\label{xidef}
    \xi(x)+\xi(y)-\xi(z)=0\mod 2\qquad, \qquad N^{xy}_z>0\,.
\end{equation}
Since $Q(c)=1\mod 2$, we have $\xi(c)=0\mod 2$. According to Theorem \ref{thm:TBone_smtc}, there is an Abelian anyon $w$ that satisfies
\begin{equation}
    (-1)^{\xi(x)}=e^{2\pi i B(w,x)}\,.
\end{equation}
Since these braiding phases should all be $\pm 1$, $w^2=1$. In other words, the choices are parametrized by a torsor over the 2-torsion subgroup of $\cA$. Here torsor means that $\xi$ measures the difference between different charge assignments (all with the same $Q$ mod 1, or the same vison).

Note that the quadratic refinement discussed in Sec.~\ref{sec:tqft-basics}, $\hat{h}$, is
\begin{equation}
    \hat{h}(x)=\left(h(x)+\frac{\tilde{B}(x,v)}{2}+\frac{t(x)}{2}\right) \mod 1.
\end{equation}
If we denote the one for $t_0$ as $\hat{h}_0$, then
\begin{equation}
    \hat{h}(x)=\left(\hat{h}_0(x)+\frac12 \xi(x) \right)\mod 1= \left(\hat{h}_0(x)+B(w,x)\right)\mod 1\,.
\end{equation}

\begin{example}
We give a somewhat trivial example of assigning charges to a spin TQFT to turn it into a spin$^c$ TQFT. Suppose $\cal{T}$ splits, i.e., ${\cal T} = {\cal M}\boxtimes\{1,c\}$, where ${\cal M}$ is a bosonic TQFT. Then, one consistent choice is $Q(x)=0$ mod 2 for all anyons in ${\cal M}$, which corresponds to $v=1$. According to our classification, any other charge assignment with the same $Q$ mod 1 is given by
\begin{equation}
    (-1)^{Q(x)}=e^{2\pi i B(w,x)}\,.
\end{equation}
where $w^2=1$.

Unlike the bosonic TQFT, in general the vison can not be an arbitrary Abelian anyon owing to the spin-charge relation. Below we illustrate the subtlety with the example of ${\cal M}={\rm Ising}$ and provide a complete classification of charge assignments in this theory.

First, we show that $v=\psi$ is inconsistent, so $v=1$ is the only consistent choice. This can be proved using the generalized ribbon identity. Let us instead show it using \eqref{t-fusion}. Choose the lifting to be $\tilde{B}(\psi,\psi)=0, \tilde{B}(\psi,\sigma)=1/2\mod 2$. Then
\begin{equation}
    2t(\sigma)=-2\tilde{B}(v,\sigma)  \mod 2.
\end{equation}
Since $2t(\sigma)$ is even, we have $B(v,\sigma)=0\mod 1$, so the only allowed vison is $v=1$. In particular, $v$ cannot be $\psi$.

With $v=1$, $w$ has two choices: $w=1$ or $w=\psi$, corresponding to two distinct spin$^c$ Ising theories. They are distinguished by the quadratic refinement of the $\sigma$ anyon: $\hat{h}(\sigma)=1/16$ and $\hat{h}(\sigma)=9/16$, corresponding to $Q(\sigma)=0\mod 2$ or $Q(\sigma)=1\mod 2$, respectively.
\end{example}

As the above example reveals, this approach to classifying spin$^c$ TQFTs leaves something to be desired. It necessitates manually checking all Abelian anyons to find the choices consistent with being the vison of a spin$^c$ TQFT. After this is done, one needs to find all order two Abelian anyons to finish the classification. This motivates us to consider another approach.

\subsubsection{Second derivation: From a spin theory to its modular extension and back to a spin$^c$ theory}\label{secondderiv}

Now we discuss a second approach, which is more systematic and will allow us to find more complete consistency constraints on the vison.

We start with a spin$^c$ theory $\cal T$ and try to explore its properties.  In particular, we would like to classify its possible $\U$ symmetry enrichment.  We do this by using its minimal modular extension/bosonic shadow ${\cal T}^{\rm ext}$.  As reviewed earlier, the latter is a bosonic TQFT coupled to $\U$.  It can be decomposed as ${\cal T}^{\rm ext}={\cal T}^{\rm ext}_0\oplus {\cal T}^{\rm ext}_1$, with ${\cal T}^{\rm ext}_0$ being the original spin$^c$ theory $\cal T$.

We denote the U(1) charge in the bosonic TQFT as $Q^{\rm ext}$, normalized such that $Q^{\rm ext}=\frac12 Q\mod 1$.
Then local bosons in $\cT^{\rm ext}$ have integer charges, and $Q^{\rm ext}(c)=\frac12\mod 1$.  Conversely, suppose such a $Q^{\rm ext}$ can be assigned to ${\cal T}^{\rm ext}$. Then $2Q^{\rm ext}$ restricted to ${\cal T} = {\cal T}^{\rm ext}_0$ is a charge assignment satisfying Definition \ref{defn:spinc}.

As done in Sec.~\ref{flux threading}, to classify such charge assignment in $\cT^{\rm ext}$,  we can adiabatically thread $2\pi$ flux ($\pi$ flux) with respect to $Q^{\rm ext}$ ($Q=2Q^{\rm ext}$). As in Sec.~\ref{flux threading} we identify the result of doing so with an Abelian anyon $\bm{\gamma}$. In other words, $\bm{\gamma}$ is a vison for the bosonic theory ${\cal T}^{\rm ext}$. This anyon must have $\pi$ braiding with $c$ to give $Q^{\rm ext}(c)=\frac12 \mod 1$, so $\bm{\gamma} \in {\cal T}^{\rm ext}_1$.   The charge assignment is then completely determined by $\bm{\gamma}$,
\begin{align}\label{QinText}
Q(a)=2Q^{\rm ext}(a) &= 2B(a, \bm{\gamma}) \mod{2}\,.
\end{align}

Again, as in Sec.~\ref{flux threading}, the charge and topological spin of $\bm{\gamma}$ are given by
\begin{equation}
\begin{aligned}
\label{eqn:gamma-Q-h}
Q(\bm{\gamma}) &= 2Q^{\rm ext}(\bm{\gamma})=2\sigma_H^{\rm ext}=\frac{\sigma_H}{2}\mod 2\\
h(\bm{\gamma}) &= \frac{\sigma_H^{\rm ext}}{2} = \frac{\sigma_H}{8} \mod{1}\,.
\end{aligned}
\end{equation}
Note that this is because $Q=2Q^{\rm ext} \mod 2$ means $\sigma_H = 4\sigma_H^{\rm ext} \mod 8$.

Let us relate this construction to our first derivation in Appendix~\ref{first deriv}. The vison $v\in \cT^{\rm ext}_0=\cT$ is determined by 
\begin{equation}
    Q(a)=B(v,a)\mod 1.
\end{equation}
We thus find that $B(v,a)=2B(\bm{\gamma},a)=B(\bm{\gamma}^2,a)\mod 1$. Thus, we must have $v=\bm{\gamma}^2$.  Physically, this is obvious; $\bm\gamma$ corresponds to $\pi$ flux of $Q$ and $v$ corresponds to $2\pi$ flux of $Q$.

Thus, if the vison in $\cal{T}$ is given, then  $\bm{\gamma}^2=v$ means that the different choices form a torsor over the subgroup of order-2 elements in $\cA$. This is indeed just what we saw in Appendix~\ref{first deriv}.

A useful result, which is formalized as Lemma \ref{lemma_1to1} below, is that there is a one-to-one correspondence between $\cT$ and $\cT_1^{\rm ext}$:

\begin{lemma} \label{lemma_1to1}
Suppose ${\cal T}$ has a bosonic shadow ${\cal T}^{\rm ext}$ with an Abelian $\bm{\xi}\in {\cal T}^{\rm ext}_1$. Then the anyons in ${\cal T}^{\rm ext}_1$ are in one-to-one correspondence with those in $\cT^{\rm ext}_0={\cal T}$. In particular,
\begin{equation}\label{ximap01}
{\cal T}^{\rm ext}_1 =\bm{\xi}\times \cT^{\rm ext}_0= \{ \bm{\xi} \times a \mid a \in {\cal T}^{\rm ext}_0\}.
\end{equation}
\end{lemma}

Additionally if ${\cal A}_1 \subseteq {\cal T}_1^{\rm ext}$ is the subset of Abelian anyons (note this does not form a group), then ${\cal A}_1 = \bm{\xi}\times {\cal A}_0$, where ${\cal A}_0 \subseteq {\cal T}_0^{\rm ext}={\cal T}$ are the original Abelian anyons. 

\begin{proof}
Since $\bm{\xi}$ is Abelian,  the set $\bm{\xi}\times \cT^{\rm ext}_0\subseteq {\cal T}^{\rm ext}_1$ contains exactly the same number of elements as $\cT^{\rm ext}_0$. Further, since $\bm{\xi}$ is Abelian $d_{\bm{\xi}_a} = d_a$. Then the total quantum dimension of $\bm{\xi}\times \cT^{\rm ext}_0$ is equal to that of ${\cal T}^{\rm ext}_0$. Since ${\cal D}_{{\cal T}^{\rm ext}_0}={\cal D}_{{\cal T}^{\rm ext}_1}$, as proven in Appendix~\ref{app:gauging-quantum-dim}, this set saturates the total quantum dimension of $\cT^{\rm ext}_1$. Thus ${\cal T}^{\rm ext}_1 = \bm{\xi}\times {\cal T}^{\rm ext}_0$.

The same logic shows that multiplying by $\bm{\xi}$ is a one-to-one mapping between ${\cal A}_0$ and ${\cal A}_1$.
\end{proof}

Note that this fact is consistent with the statement in Appendix~\ref{HilbertspaceT} that the different spin structures are related to each other via continuously tuning boundary conditions. The continuous change in the boundary conditions is the same as the process of flux threading.  (In other contexts, it is known as spectral flow.) And the one-to-one map of the anyons is clear by following them during the continuous flux-threading process.  In the case of the $\pi$ flux threading, the map \eqref{ximap01} is with $\bm{\xi}=\bm{\gamma}$, i.e., with the vison of ${\cal T}^{\rm ext}$.

We have shown that charge assignment in the bosonic TQFT $\cal T^{\rm ext}$  such that $Q^{\rm ext}(c)=1/2$ mod 1 is parametrized by ${\cal A}_1$. When restricted to $\cT$, since anyons in $\cT$ braid trivially with $c$, different choices are classified by $\cA_1/\{1,c\}$. 
We can show that the classification is complete: that all consistent solutions of $Q(a)$ for $a\in \cT$ satisfying the conditions in Definition \ref{defn:spinc} are obtained from \eqref{QinText} for $\bm{\gamma}\in\cA_1/\{1,c\}$.

We have now accomplished our goal of classifying the symmetry enrichment of a spin$^c$ theory. Thus far, however, our classification relied on a particular choice of a modular extension ${\cal T}^{\rm ext}$, which possessed an Abelian anyon $\bm{\gamma} \in {\cal T}^{\rm ext}_1$.
We could accept this, and simply state that the modular extension we chose is the unique one chosen by gauging the fermion parity, as in  Remark~\ref{rem:all_shadows}. Instead, we show that our classification is independent of this choice. 
First, we prove two simple corollaries.

\begin{corollary}\label{corr:ab_majorana}
    Let ${\cal T}$ be a spin TQFT. Then if ${\cal T}^{\rm ext}_1$ contains an Abelian anyon $\bm{\gamma}$, it has no Majorana anyons.
\end{corollary}
\begin{proof}
    This follows directly from Lemma \ref{lemma_1to1}, as ${\cal T}^{\rm ext}_1 = \bm{\gamma} \times {\cal T}^{\rm ext}_0$ and no anyons ${\cal T}^{\rm ext}_0$ are invariant under fusion with $c$.
\end{proof}

This also proves the contrapositive:
\begin{corollary} \label{cor_abormaj}
     Let ${\cal T}$ be a spin TQFT. Then if ${\cal T}^{\rm ext}_1$ contains a Majorana anyon, it has no Abelian anyons.
\end{corollary}

One immediate consequence of these results is that since $\cT^{\rm ext}_1$ has to contain the Abelian anyon $\bm{\gamma}$, a necessary condition for a spin{}$^c$ TQFT $\cT$ is that $\cT^{\rm ext}_1$ cannot contain any Majorana anyon. This is in accord with the observation made after \eqref{numberof states}, that in a spin{}$^c$ theory $N_o=N_e$.

Now, starting from a minimal modular extension ${\cal T}^{\rm ext}$ (which contains $c$), all other minimal modular extensions can be constructed as follows:
\begin{equation} \label{eq:extension}
    {\cal T}'=\frac{{\rm Spin}(n)_1\boxtimes {\cal T}^{\rm ext}}{(\psi c)}\,,
\end{equation}
where $c\in {\cal T}^{\rm ext}$ and $\psi\in {\rm Spin}(n)_1$. 

Let us first discuss the case of odd $n$ and argue that it is inconsistent with our assumptions.  To see that, we will argue that all the $\pi$-flux anyons are non-Abelian.  Consider the anyon content of ${\cal T}'$ explicitly. We follow the steps outlined in Sec.~\ref{sec:tqft-basics} to gauge the one-form symmetry generated by $\psi c$. First, following steps one and two, the anyons that have trivial braiding with $\psi c$ are either contained in $\cT_0$ or of the form $\sigma\bm{a}$ for $\sigma \in {\rm Spin}(n)_1$, $\bm{a}\in \cT^{\rm ext}_1$. Next, following step three, we need to determine if any of these anyons are fixed points under fusion with $\psi c$. The anyons in ${\cal T}_0$ are clearly not. For the anyons of the form $\sigma \bm{a}$ we have $\sigma\bm{a} \times \psi c=\sigma (\bm{a}\times c)\neq \sigma\bm{a}$ by the results of Corrollary~\ref{corr:ab_majorana} and the fact that ${\cal T}^{\rm ext}$ was chosen such that there is an Abelian anyon in ${\cal T}_1^{\rm ext}$. Thus these anyons are not fixed points either and ${\cal T}' = {\cal T}_0'\oplus {\cal T}'_1$ with ${\cal T}_0' = {\cal T}_0$ and ${\cal T}_1' = \{\sigma \bm{\gamma}a \mid \sigma \in {\rm Spin}(n)_1, a\in \cT_0/\{1,c\}\}$. Since $\sigma$ is non-Abelian, all of the anyons in ${\cal T}_1'$ are non-Abelian. Thus, they cannot be the result of $\pi$-flux threading since these anyons must be Abelian, as discussed earlier.
Physically, since getting to ${\cal T}'$ from ${\cal T}^{\rm ext}$ corresponds to stacking $p+ip$ superconductors, the charge conservation is broken. 

Thus, we take $n=2k$. Then, one can write the anyons in Spin$(2k)_1$ as $\{1,m,\psi, m\psi\}$ with $B(m,\psi)=1/2, h(m)=k/8$. The anyons in ${\cal T}'={\cal T}'_0\oplus {\cal T}'_1$ can be described as follows: ${\cal T}'_0$ is the same as ${\cal T}^{\rm ext}_0=\cal{T}$, and ${\cal T}'_1=m\times {\cal T}_1^{\rm ext}$ (thus a one-to-one mapping). Because $m$ is Abelian, there is also a one-to-one mapping between the Abelian subset ${\cal A}_1'\subset {\cal T}_1'$, and ${\cal A}_1$. In this process $\bm{\gamma}$ is shifted to $m\times \bm{\gamma}$, but this does not affect its braiding with ${\cal T}_0 = {\cal T}_0'$. Hence, the classification is not affected. 

Note that the change of $\bm{\gamma}$ to $m\times \bm{\gamma}$ shifts the topological spin by $h(m) = k/8$. Since the chiral central charge is shifted by $k$, $c_-/8 - h(\bm{\gamma})$ is invariant under this change. Therefore $c_-/8 - h(\bm{\gamma}) = (c_- -\sigma_H)/8$ is an invariant of the choice of bosonic shadow.

Let us summarize what we have learned.  We started with a spin$^c$ theory $\cal T$ and derived its minimal bosonic extension ${\cal T}^{\mathrm{ext}}$.\footnote{Our discussions up to this point in fact provides an explicit construction of a minimal modular extension $\cT^{\rm ext}$, given a spin{}$^c$ charge assignment on $\cT$ and the Hall conductivity $\sigma_H$. More concretely, $\cT^{\rm ext}_1$ is given by $\bm{\gamma}\times \cT$, where $\bm{\gamma}$ has $h(\bm{\gamma})=\frac{\sigma_H}{8}\mod 1$ and $B(a,\bm{\gamma})=\frac12 Q(a)\mod 1$. Then the full data of the modular extension is determined.}  This can be done in eight different ways, labeled by stacking SO$(2k)_1$ with $k=0,1,\cdots,7$ before we sum over the spin structures.  The difference between these theories is in $\sigma_H\to \left(\sigma_H+k\right) \mod 8$ and consequently in $h(\bm{\gamma})\to \left(h(\bm{\gamma})+{k\over 8}\right)\mod 1$ and $Q(\bm{\gamma})\to \left(Q(\bm{\gamma})+{k\over 2}\right)\mod 2$.  The rest of the modular data is determined up to this $k$ dependence; however, the charge assignment on $\mathcal{T}$ is independent of $k$.\footnote{In Sec.~\ref{sec:tqft-gappedphase}, we argued that in spin$^c$ theories, we should identify theories with the same $\sigma_H\mod1$.  How come they appear to be different here?  The point is that, as we argued in Sec.~\ref{sec:tqft-gappedphase}, shifting $\sigma_H$ by an integer does not change $\cal T$. It affects only contact terms and the behavior of $\cal T$ in the presence of boundaries.  Indeed, the Ramond line is an anyon in ${\cal T}^\mathrm{ext}$, but it is not an anyon in $\cal T$.  One way to think about it in $\cal T$ is by removing a line from our three-manifold and specifying Ramond boundary conditions for the fermions around it, or equivalently, holonomy $-1$ for $A$ around it. This can be thought of as a boundary in $\cal T$.}

We now prove several useful facts about spin$^c$ theories. The first is a generalized version of the Gauss-Milgram sum for spin$^c$ theories that appeared in \cite{lapa_anomaly_2019}. We will not use it in our classification, but state it here for its usefulness.

\begin{definition}
Let ${\cal T}$ be a spin$^c$ theory. Then define the following quantity for anyons in $\cal T$:
\begin{equation}
\begin{aligned}
    \hat h(a)&=h(a)+\frac{Q(a)}{2}\mod 1\\
    &= h(a) + B(\bm{\gamma}, a)\\
    &= h(\bm{\gamma}\times a) - h(\bm{\gamma})\,.
\end{aligned}
\end{equation}
\end{definition}

This quantity was discussed in  \cite{Hsin:2019gvb}. It is well defined for the spin$^c$ theory even if $c$ is identified with the vacuum, as reviewed in Sec.~\ref{sec:tqft-basics}. It is a quadratic refinement of the topological spin $h$, so the ribbon identity holds.

\begin{theorem}\label{thm:Gauss-Milgram}
In a spin{}$^c$ category $\cal T$ we have a generalized Gauss-Milgram sum:
\begin{equation}
     e^{\frac{2\pi i}{8}(c_--\sigma_H)}=\frac{1}{\sqrt{2}{\cal D}_{\cal T}}\sum_{a\in {\cal T}}d_a^2e^{2\pi i \hat h(a)}.
\label{gauss2}
\end{equation}
\end{theorem}
\begin{proof}
This is established in \cite{lapa_anomaly_2019}, and we reproduce the proof here. Pick a bosonic shadow ${\cal T}^{\rm ext}$ with an Abelian anyon $\bm{\gamma}\in {\cal T}^{\rm ext}_1$. Since this is a bosonic theory, we have the following Gauss-Milgram sum:
\begin{equation}
\begin{split}
    e^{\frac{2\pi i}{8}c_-}&=
    \frac{1}{\cal D_{{\cal T}^{\rm ext}}}\left(\sum_{a\in {\cal T}}d_{a}^2e^{2\pi i h(a)}+ \sum_{\bm{\alpha}\in {\cal T}^{\rm ext}_1}d_{\bm{\alpha}}^2e^{2\pi i h(\bm{\alpha})}\right)\\
    &=\frac{1}{{\cal D}_{{\cal T}^{\rm ext}}}\sum_{\bm{\alpha}\in {\cal T}^{\rm ext}_1}d_{\bm{\alpha}}^2e^{2\pi i h(\bm{\alpha})}\\
    &=\frac{1}{\sqrt{2}{\cal D}_{\cal T}}
    \sum_{\bm{\alpha}\in {\cal T}^{\rm ext}_1}d_{\bm{\alpha}}^2e^{2\pi i h(\bm{\alpha})}\\
    &=\frac{1}{\sqrt{2}{\cal D}_{\cal T}}\sum_{a\in {\cal T}}d_a^2e^{2\pi i h(\bm{\gamma}_a)}\\
    &=\frac{1}{\sqrt{2}{\cal D}_{\cal T}}e^{2\pi ih(\bm{\gamma})}\sum_{a\in {\cal T}}d_a^2e^{2\pi i \hat h(a)}
\end{split}
\end{equation}
Recall that $c_-/8 - h(\bm{\gamma}) = (c_- - \sigma_H)/8$ is independent of the choice of bosonic shadow.
\end{proof}

We now present a few results that will prove useful in our classification.

\begin{theorem} \label{thm:triv_vison_split}
If all anyons in a spin$^c$ theory ${\cal T}$ carry integer charges, then the super-modular category must split.
\end{theorem}
\begin{proof}
If all anyons carry integer charges, then by Theorem~\ref{thm:TBone_smtc} the vison $v$ is either $1$ or $c$. We know that $\bm{\gamma}^2=v$, so $\bm{\gamma}^2=1$ or $c$. In either case the set ${\cal M} = \{1,\bm{\gamma},c,\bm{\gamma}c\}$ is closed under fusion and has non-degenerate braiding. Thus ${\cal M}$ is a MTC.\footnote{There are eight options given by Spin$(4k)_1$ and Spin$(4k+2)_1 = {\cal A}^{4,2k+1}$ for $k=0,1,2,3$. In the first $v=1$ and in the latter $v=c$.}

Now by Theorem~\ref{mtc-factorization}, ${\cal T}^{\rm ext}$ must factorize as ${\cal T}'\boxtimes \{1,\bm{\gamma},c,\bm{\gamma}c\}$ where $\cal T'$ is an MTC. Recall that ${\cal T}^{\rm ext} = {\cal T} \oplus {\cal T}\times \bm{\gamma}$, where the elements in ${\cal T}$ are precisely those that braid trivially with $c$. It is then not hard to see that since ${\cal T}^{\rm ext} = {\cal T}'\boxtimes \{1,\bm{\gamma},c,\bm{\gamma}c\}$ for some MTC ${\cal T}'$, then the elements braiding trivially with $c$ are precisely ${\cal T}'\times \{1,c\}$. So then ${\cal T} = {\cal T}' \times \{1,c\}$ and we see our super-modular category must split.
\end{proof}

\begin{corollary}
\label{cor:MTC-zero-charge}
Consider a split spin$^c$ theory, $\cT={\cal M}\boxtimes \{1,c\}$, with an MTC $\cal M$, where all the anyons have integer charge. Then, there exists an MTC ${\cal M}'$ such that $\cT={\cal M}'\boxtimes\{1,c\}$, and the anyons in ${\cal M}'$ have zero charges mod 2.
\end{corollary}
\begin{proof}
We will directly construct ${\cal M}'$. The anyon types are in one-to-one correspondence with the anyon types of ${\cal M}$, through the following map:
\begin{equation}
    x\rightarrow x'=x c^{Q(x)}\quad,\qquad x\in {\cal M}.
\end{equation}
It is clear that $\{x'\}$ follow the same fusion rules of $\cal M$, i.e., $N^{x'y'}_{z'}=N^{xy}_z$.

The F and R symbols of the theory are given by
\begin{equation}
    F^{x',y',z'}=F^{xyz}\quad, \qquad R^{x',y'}=R^{xy} (-1)^{Q(x)Q(y)}.
\end{equation}
It is straightforward to check that the pentagon and hexagon identities are still satisfied. The modular data of ${\cal M}'$ is given by 
\begin{equation}
    h(x')=h(x)+\frac12 Q(x)\quad,\qquad S_{x',y'}=S_{x,y}\,.
\end{equation}
Thus, the datum defines an MTC, where the charges satisfy $Q(x')=0\mod 2$ for all $x'$. 

We finish by using Theorem~\ref{thm:smtc-factorization} to write ${\cal T} = {\cal M}' \boxtimes {\cal T}'$, where ${\cal T}'$ is a super-modular category. Since ${\cal D}_{{\cal M}'} = {\cal D}_{{\cal M}} = {\cal D}_{{\cal T}}/\sqrt{2}$, this mean ${\cal D}_{{\cal T}'} = \sqrt{2}$. Then, we conclude ${\cal T}' = \{1,c\}$.
\end{proof}

\begin{theorem} \label{theo:notspinc}
If a super modular category $\mathcal{T}$ has a minimal modular extension ${\cal T}^{\rm ext}$ such that ${\cal T}^{\rm ext}_1$ contains both a Majorana anyon and non-Majorana anyon, then $\mathcal{T}$ is not spin$^c$.
\end{theorem}

\begin{proof}
Suppose, by contradiction, that ${\cal T}$ is spin$^c$. Then at least one of its minimal modular extensions must contain an Abelian anyon. By Corollary~\ref{cor_abormaj}, this means that at least one minimal modular extension cannot contain a Majorana anyon. 

We start with the minimal modular extension that contains both a Majorana and a non-Majorana anyon. Then all other minimal modular extensions can be constructed by
    \begin{equation}\label{TTprime}
        {\cal T}'=\frac{{\rm Spin}(n)_1\boxtimes {\cal T}^{\rm ext}}{(\psi c)}.
    \end{equation}
    
For $n$ even, we saw earlier that ${\cal T'}^{\rm ext}_1 = m \times {\cal T}^{\rm ext}_1$ for an Abelian $m\in {\rm Spin}(n)_1$, and so the anyons are in one-to-one correspondence with the anyons in ${\cal T}^{\rm ext}_1$. Thus ${\cal T'}^{\rm ext}_1$ contains a Majorana anyon as ${\cal T}^{\rm ext}_1$ contains one.

On the other hand, for $n$ odd, we can consider the anyon $\sigma a$, where $\bm{a} \in {\cal T}^{\rm ext}_1$ is non-Majorana. As $\bm{a} \neq c \bm{a}$, it is not a fixed-point of the gauging. Furthermore, $\sigma \bm{a}$ and $ \sigma c\bm{a}$ will be identified after gauging as they differ by $ c\psi$. Thus $\sigma \bm{a} \in  {\cal T'}^{\rm ext}_1$ is a Majorana anyon.

    Since all minimal modular extensions contain a Majorana anyon we arrive at a contradiction and conclude that ${\cal T}$ is not spin$^c$.
\end{proof}

\begin{corollary}
    ${\cal F}_{4m+2}$ is not spin$^c$ for $m> 0$.
\end{corollary}
Note that the same result was found in Appendix~\ref{HilbertspaceT}.  

For $m=0$, as a super-modular category, ${\cal F}_2$ is identical to the trivial theory $\{1,c\}$. Whether it is spin{}$^c$ or not now depends on additional data, namely the choice of modular extension, which is determined by the chiral central charge $c_-$. If $c_-\in \bZ+1/2$, then the theory is not spin{}$^c$.

\subsection{Examples}

Below, we consider three classes of examples.

 \begin{enumerate}
     \item  When the theory ${\cal T}$ can be factorized as ${\cal M}\boxtimes \{1,c\}$ where ${\cal M}$ is an MTC, a minimal modular extension is ${\cal M}\boxtimes \{1,c,\bm{m},\bm{m}c\}$.  
     Then there exists an Abelian $b\in {\cal M}$ such that $\bm{\gamma}=\bm{m}\times b$, and $e^{i\pi Q(a)}=e^{2\pi iB(\bm{m}\times b,a)}=e^{2\pi iB(b,a)}$  for $a\in{\cal M}$, since $B(\bm{m},a)=0$. In other words, $Q(a)=2B(b,a)$ mod 2. Note that the vison $v=\bm{\gamma}^2=\bm{m}^2b^2$. Depending on the minimal modular extension, $\bm{m}^2$ may be $1$ or $c$.

    The quadratic refinement is
    \begin{equation}
        \hat{h}(a)=h(a)+B(b,a)\mod 1.
        \label{Lorentz-sym-frac}
    \end{equation}

    In particular, if $v=1$, then we must have $b^2=1$, i.e., $b$ generates a $\bZ_2^{(1)}$ one-form symmetry. In this case, \eqref{Lorentz-sym-frac} can be viewed as a mapping from $\cal T$ to another bosonic TQFT, which is also the one discussed in Corollary~\ref{cor:MTC-zero-charge}.
    A similar map between bosonic TQFTs was studied in \cite{Hsin:2019gvb} as Lorentz symmetry fractionalization (because attaching a transparent fermion changes the spin by $1/2$).
    
     Let us further consider $\cA^{s,r}\subset {\cal M}$, which are the symmetry lines. Assume that $b\in \cA^{s,r}$. In this case, we can write $b=a_k$, then we find
\begin{equation}
    Q(a_j)=\frac{2rkj}{s}\mod {2}\,.
\end{equation}
The vison is $v=b^2=a_{2k}$. The quadratic refinement in the spin$^c$ theory is given by
\begin{equation}
    \hat{h}(a_j)=\frac{rj^2}{2s} + \frac{rkj}{s} \mod 1\,.
\end{equation}

\item  For a nonsplit example, consider the Moore-Read state $\pf_{q,1,1}^f$ (the physical realization has $\sigma_H=\frac{1}{q}$ and $c_-=\frac32$).  $\psi_{2q}$ is the transparent fermion. The unique minimal modular extension with $c_-=\frac32$ is $\cT^{\rm ext}={\rm Ising}\boxtimes \cA^{4q,1}$. In this theory the Abelian  $\pi$ flux anyons $\cA_1$ are $1_n$ and $\psi_n$  with $n=1,3,\dots, 4q-1$, where we use the notation of Appendix~\ref{PfaffianTO}. One can take $\cA_1/\{1,c\}=\{1_n|n=1,3,\dots,4q-1\}$.  Suppose $\bm{\gamma}=1_{n}$ (so the vison $v=\bm{\gamma}^2=1_{2n}$). Using the relation $h(\bm{\gamma})=\frac{\sigma_H}{8}$ mod 1, we find $\sigma_H=\frac{n^2}{q}$ mod 8. The non-Abelian anyon $\sigma_1$ then has charge $Q(\sigma_1)= \frac{n}{2q}$ mod 2. 

 \end{enumerate}

\subsection{Restriction to symmetry lines}
\label{app:rest_symm_lines}

Let us apply to ${\cal T}^{\rm ext}$ what we learned about coupling a bosonic theory TQFT to $\U$. $\cal T$ is coupled to a spin$^c$ connection $A$.  Therefore, its bosonic extension ${\cal T}^{\rm ext}$ can be thought of as coupled to an ordinary $\U$ gauge field $A^{\rm ext}=2 A$, whose fluxes obey the standard quantization.

First, it has its own vison  $\bm{\gamma}$ generating ${\cal V}^{s',r'}$ with some integers $s'$ and $r'$ to be determined.  Also, $c$ generates ${\cal V}^{2,2}$. Comparing with the discussion in Appendix~\ref{secondderiv}, the anyons in ${\cal T}^{\rm ext}_0$ are even under the $\mathbb{Z}_2^{(1)}$ one-form symmetry generated by $c$ and the anyons in ${\cal T}^{\rm ext}_1$ are odd under that symmetry.

The fact that  $Q^{\rm ext}(c)={1\over 2} \mod 1$ means that $B(\bm{\gamma},c)={1\over 2} \mod 1$, and hence $B(\bm{\gamma}^{s'},c)={s'\over 2} \mod 1={1\over 2}\mod 1$.  We conclude that $s'$ should be even.  

We should distinguish between two case depending on whether $c\in {\cal V}^{2,2}\subseteq {\cal V}^{s',r'} $ or $c\in {\cal V}^{2,2}\not \subseteq {\cal V}^{s',r'} $.

\subsubsection{$c\in {\cal V}^{s',r'}$} 
Let us first assume that $c\in {\cal V}^{s',r'}$.  Since $c$ is of order two, $c=\bm{\gamma}^{s'\over 2}$.  Therefore, it has $h(c)={r' s'\over 8}\mod 1={1\over 2}\mod 1$ and $Q^{\rm ext}(c)= {r'\over 2} \mod 1={1\over 2}\mod 1$.  Therefore, $r'$ should be odd and $s'=4s$ with an odd integer $s$. 

Now, we find the fermionic theory by gauging $\mathbb{Z}_2^{(1)} \subset \mathbb{Z}_{s'}^{(1)}$ generated by $c=\bm{\gamma}^{s'\over 2}$. This turns  ${\cal V}^{s',r'}$ into $ {\cal V}^{s,r'}$, which is generated by $v=\bm{\gamma}^2$.

Setting $r'=r$ and using the information in the bosonic theory with ${\cal V}^{4s,r}$, the elements of $ {\cal V}^{s,r}$ have
\begin{equation}\label{cin V}
    \begin{split}
        &h(v^j)={r j^2\over 2 s}\mod {1\over 2}\\
        &Q(v^j)=\left(2 Q^{\rm ext}(v^j)\right)\mod 1 ={r j\over s}\mod 1\\
         &\hat h(v^j)=\left(h(v^j) + {1\over 2} Q(v^j)\right) \mod 1=\left({r j^2\over 2 s}+{r j\over 2s}\right)\mod 1\\
         & r,s \in 2\mathbb{Z}+1\,.
    \end{split}
\end{equation}

As an example, consider $r'=1$ where ${\cal V}^{s',r'}$ is $\U_{4s}$ and the theory after gauging is $\U_s$ with odd $s$.

\subsubsection{$c\not\in {\cal V}^{s',r'}$}

Next, we consider the case where $c$ is not in ${\cal V}^{s',r'}$.  Now,  $B(\bm{\gamma},c)={1\over 2} \mod 1$ still means that $s'$ should be even.  It also means that we have a mixed anomaly between ${\cal V}^{2,2}$ of $c$ and ${\cal V}^{s',r'}$ of $\bm{\gamma}$ and therefore, even though as groups, we have $\mathbb{Z}_2^{(1)}\otimes \mathbb {Z}_{s'}^{(1)}$, the two ${\cal V}$'s do not factorize.

Now, we find the fermionic theory by gauging $\mathbb{Z}_2^{(1)}$ generated by $c$.  Since $B(\bm{\gamma},c)={1\over 2} \mod 1$, the gauging removes the lines $\bm{\gamma}^j$ with odd $j$ and leaves only powers of $v=\bm{\gamma}^2$.  The resulting one-form symmetry is $\mathbb{Z}_s^{(1)}$ with $s={s'\over 2}$.  And its anomaly is $r=2 r'$, i.e., the gauging leads to ${\cal V}^{{s'\over 2},2r'}={\cal V}^{s,r}$.  Using the information in the bosonic theory, we again find
\begin{equation}\label{cnin V}
    \begin{split}
        &h(v^j)={r j^2\over 2 s}\mod {1\over 2}\\
        &Q(v^j)=\left(2 Q^{\rm ext}(v^j)\right)\mod 1 ={r j\over s}\mod 1\\
         &\hat h(v^j)=\left(h(v^j) + {1\over 2} Q(v^j)\right) \mod 1=\left({r j^2\over 2 s}+{r j\over 2s}\right)\mod 1\\
         &r\in 2\mathbb{Z}\,.
    \end{split}
\end{equation}

For example, in the Moore-Read state, we start with a bosonic theory with $s'=4q$ and $r'=1$, i.e., ${\cal V}^{4q,1}$.  And we end up with ${\cal V}^{2q,2}$.

Comparing the two constructions based on $c\in {\cal V}$ and $c\not \in {\cal V}$ ending with \eqref{cin V} and \eqref{cnin V} respectively, we see that they are in correspondence with the various options that are consistent with the spin/charge relation $r(s+1)\in 2\mathbb{Z}$.

\section{Gauging $\bZ_\ell^{(1)}$ reduces the total quantum dimension by $\ell$}
\label{app:gauging-quantum-dim}

Suppose we gauge a $\bZ_\ell^{(1)}$ one-form symmetry generated by a bosonic anyon $b$ of order $\ell$ in a TQFT $\cT$, resulting in a new theory $\cT'$. In this appendix, we prove

\begin{theorem}\label{DTDT'}
    $\mathcal{D}_{\cT'}=\frac{1}{\ell}\mathcal{D}_{\cT}$.
\end{theorem}
This is a special case of the general relation between the quantum dimensions of theories related by ``anyon condensation,'' or gauging (possibly non-invertible) one-form symmetry. See \cite{kirillov2002q, Kong:2013aya, neupert2016boson} for the general result. Here we present an elementary proof for the case where the one-form symmetry is $\bZ_\ell^{(1)}$, which is of most interest in this work.

Before going to the proof, we need a lemma from Ref.~\cite{barkeshli_symmetry_2019}, which is of independent interest and is used in other places in the text:
\begin{lemma}\label{lemma:G-graded-D}
    Suppose $\cal T$ has  $G$-graded fusion rules where $G$ is a finite group, i.e., ${\cal T}=\oplus_{g\in G} {\cal T}_g$ and ${\cal T}_g\times {\cal T}_h\in {\cal T}_{gh}$, then ${\cal D}_g^2=\sum_{a_g\in {\cal T}_g}d_{a_g}^2={\cal T}_0^2$, where $0$ denotes the identity element of the group $G$.  In particular, ${\cal D}_{g}$ is independent of $g$.
\end{lemma}

\begin{proof}
Denote by $b_g$ an arbitrary anyon in ${\cal T}_g$. We will need the following basic properties of quantum dimensions and fusion coefficients (see \cite{Moore:1989vd,kitaev_anyons_2006,bonderson_non-Abelian_2007}): i) $d_a=d_{\bar{a}}$, where $\bar{a}$ is the dual anyon of $a$. ii) $d_ad_b=\sum_{x}N_{ab}^x d_x$. iii) $N_{a\bar{b}}^x=N_{xb}^a$. We also note that with $G$-graded fusion, $\overline{b_g}\in {\cal T}_{g^{-1}}$. 

Then
    \begin{equation}
        \begin{split}
            \mathcal{D}_g^2&=\sum_{a_g\in {\cal T}_g}d_{a_g}^2\\
            &=\sum_{a_g\in {\cal T}_g}d_{b_g}^{-1}d_{a_g}d_{a_g}d_{b_g}\\
            &=\sum_{a_g\in {\cal T}_g}d_{b_g}^{-1}d_{a_g}d_{a_g}d_{\overline{b_g}}\\
            &=\sum_{a_g\in {\cal T}_g} \sum_{x_0\in {\cal T}_0} d_{b_g}^{-1}d_{a_g}N_{a_g\overline{b_g}}^{x_0}d_{x_0}\\
            &=\sum_{a_g\in {\cal T}_g}\sum_{x_0\in {\cal T}_0} d_{b_g}^{-1}d_{x_0}N_{x_0b_g}^{a_g}d_{a_g}\\
            &=\sum_{x_0\in {\cal T}_0}d_{b_g}^{-1}d_{x_0}d_{x_0}d_{b_g}={\cal D}_0^2.
        \end{split}
    \end{equation}
\end{proof}
Now we prove the main theorem.
\begin{proof}
First, we group all anyons in $\mathcal{T}$ according to their one-form charge under $\bZ_\ell^{(1)}$ (equivalently, the braiding phase with $b$).   Denote by ${\cal T}_k$ the collection of anyons with charge $k\mod \ell$ under the one-form symmetry.  It is clear that this one-form charge gives a $\mathbb{Z}_\ell^{(1)}$-grading structure to $\mathcal{T}$, and from Lemma \ref{lemma:G-graded-D}, we know that all the quantum dimensions $\mathcal{D}_{\mathcal{T}_k}^2$ are the same. Therefore,
$
\mathcal{D}_{\mathcal{T}}^2 = \ell \, \mathcal{D}_{\mathcal{T}_0}^2.$

We now apply the three-step gauging procedure. In step 1, we remove all $\mathcal{T}_k$ with $k \neq 0\mod \ell$. For $\mathcal{T}_0$, we apply steps 2 and 3. 
We will now show that the quantum dimension is 
$
\mathcal{D}_{\mathcal{T}'}^2 = \mathcal{D}_{\mathcal{T}_0}^2/\ell
$.

To do this systematically, we need to consider what happens if an anyon $a \in \mathcal{T}_0$ is invariant under fusion with some power of $b$. Given some anyon $a$, suppose $0< r\leq\ell$ is the minimal integer such that $a=a\times b^r$.  Clearly $r$ is a factor of $\ell$ since $b^r$ has to generate a $\bZ_{\ell/r}$ subgroup of $\bZ_\ell$. We thus have an orbit
 $\{a, a \times b, \ldots, a \times b^{r-1}\}$ under fusion with $b$'s. All anyons in ${\cal T}_0$ can be grouped into such orbits.

Let us count how the quantum dimensions (squared) of these anyons change for the orbit. The anyons in the orbit $\{a, a \times b, \ldots, a \times b^{r-1}\}$ have quantum dimension squared $r d_a^2$. Next, we apply rule 2, so all the anyons in the orbit are identified and are called $a$. Applying rule 3, this anyon appears $\ell\over r$ times with quantum dimension $(r/\ell) d_a$, so the total quantum dimension squared of these anyons is
\begin{equation}
\frac{\ell}{r} \left( \frac{r}{\ell} d_a \right)^2 = \frac{r}{\ell} d_a^2\,.
\end{equation}

Therefore, for this orbit of anyons under fusion with $b$, the total quantum dimension squared shrinks by a factor of $\ell$. This applies to each orbit, therefore
$
\mathcal{D}_{\mathcal{T}'}^2 = \mathcal{D}_{\mathcal{T}_0}^2/\ell=\mathcal{D}_{\mathcal{T}}^2/\ell^2
$, so $\mathcal{D}_{\mathcal{T}'} = \mathcal{D}_{\mathcal{T}}/\ell
$.
\end{proof}

We can also derive this result from the perspective of the partition function of the 2+1d TQFT.  The three-sphere partition function before gauging is given by \eqref{ZS3D} $ {\cal Z}_{\cal T}(S^3)=S_{11}=\frac{1}{{\cal D}_{\cal T}}$.  The partition function of the theory after gauging the $\bZ_\ell^{(1)}$ one-form symmetry on a closed manifold $\tmfd$ is given by \eqref{ZsumB}
\begin{equation}
    {\cal Z}_{\cal T'}(\tmfd)={|H^0(\tmfd,\bZ_\ell)|\over |H^1(\tmfd,\bZ_\ell)|}\sum_{{\frak b}\in H^2(\tmfd, \bZ_\ell)}{\cal Z}_{\cal T}(\tmfd, {\frak b}).
\end{equation}
For $\tmfd=S^3$, both $H^1(S^3, \bZ_\ell)$ and $H^2(S^3, \bZ_\ell)$ are trivial, and $H^0(S^3, \bZ_\ell)=\bZ_\ell$.  Therefore, the sum over ${\frak b}\in H^2(\tmfd, \bZ_\ell)$ includes only one term with $\frak b =0$, which is ${\cal Z}_{\cal T}(\tmfd, 0)= {\cal Z}_{\cal T}(S^3)=S_{11}$.  So we find $ {\cal Z}_{\cal T'}(S^3)=\ell{\cal Z}_{\cal T}(S^3)=\ell S_{11}$,
which also leads to 
\begin{equation}
\mathcal{D}_{\mathcal{T}'} = {\mathcal{D}_{\mathcal{T}}\over \ell}\,.
\end{equation}

Finally, we can also prove Theorem \ref{DTDT'},  using the RCFT perspective.  We are interested in gauging the one-form symmetry $\bZ_\ell^{(1)}$ generated by $b$ to map the TQFT $\cal T$ to ${\cal T'}={\cal T}/\bZ_\ell^{(1)}$.  Since this one-form symmetry is anomaly free, its elements $b^j$ braid trivially with each other, and therefore in terms of the elements of the $S$-matrix, $S_{b^j,b^i}=S_{11}$.

Now, let us compare  $\cal T$ and  $\cal T'$.  First, rule 1 removes the anyons in $\cal T$ that do not braid trivially with $b$.  The other anyons are grouped into  orbits $\{a, a \times b, \ldots, a \times b^{r-1}\}$.  The subtleties of rule 3 do not apply to the orbit of the identity, $a=1$; hence, in this orbit, $r=\ell$. Let $\chi_x$ be the character of the anyon $x$ before gauging.  Then, using rule 2, the identity character of the RCFT after gauging is $\chi'_1=\sum_{j=0}^{\ell-1}\chi_{b^j}$.   Applying $S$ to it, $\chi'_1 \to \sum_{j=0}^{\ell-1}\sum_{x\in {\cal T}}S_{b^j,x}\chi_x$.    This can be expressed as a linear combination of characters of the anyons in ${\cal T}'$. The terms involving $\chi_{b^j}$ are $\sum_{i,j=0}^{\ell-1}S_{b^j,b^i}\chi_{b^i}=\ell S_{11}\sum_{i=0}^{\ell-1}\chi_{b^i}=\ell S_{11}\chi'_1$.  Therefore, after gauging $S'_{11}=\ell S_{11}$ and hence, 
\begin{equation}
    \mathcal{D}_{\mathcal{T}'} = {\mathcal{D}_{\mathcal{T}}\over \ell}\,.
\end{equation}

\section{Minimal modular extensions of $\cA^{q\ell, p\ell}$}
\label{app:mme}

In this appendix, we construct an explicit Abelian minimal modular extension of $\cA^{q\ell, p\ell}$, which we call ${\cal W}^{q,p,\ell}$.  Other minimal modular extensions of $\cA^{q\ell, p\ell}$ can be constructed using ${\cal W}^{q,p,\ell}$ and Eq.~\eqref{eq:mme_all}.  

For simplicity, we limit ourselves to bosonic theories.  Then, $pq\ell$ has to be even, and there are no further restrictions caused by the global $\U$ symmetry.

To start, let us factorize $\ell = \ell_1 \cdot \ell_2$  such that $\ell_2$ is the largest divisor of $\ell$ satisfying $\gcd(q,\ell_2)=1$.  Therefore, $\gcd(q\ell_1,\ell_2)=1$.  Then, using $\gcd(p,q)=1$ and $\gcd(p,\ell_1)=1$, we can factor
\begin{equation}
   \cA^{q\ell, p\ell}  = \cA^{q\ell_1 \ell_2, p\ell_1 \ell_2} = \cA^{q\ell_1, p\ell_1 \ell_2^2} \boxtimes \cA^{\ell_2, pq\ell_1^2\ell_2}\,.
\end{equation}

Next, we present a minimal modular extension of the two factors. For the first factor, we use the fact that $\gcd(q\ell_1^2,p\ell_2^2)=1$ to find that $\cA^{q\ell_1^2, p\ell_2^2} $ is a minimal modular extension of $\cA^{q\ell_1, p\ell_1 \ell_2^2}$. For the second factor, we use the fact that for even $pq$, $\cA^{\ell_2, pq\ell_1^2\ell_2}\cong \cA^{\ell_2,0}$ and therefore, it is minimally contained in $D(\mathbb{Z}_{\ell_2})$.  And if $pq$ is odd, $\cA^{\ell_2, pq\ell_1^2\ell_2}\cong \cA^{\ell_2,\ell_2}$ is minimally contained in $\cA^{\ell_2^2,1}$. (Here we used the fact that for a bosonic theory with odd $pq$, $\ell_2$ should be even.) Thus
\begin{equation}
{\cal W}^{q,p,\ell} = \left\{ \begin{array}{cc} \cA^{q\ell_1^2, p\ell_2^2} \boxtimes D(\mathbb{Z}_{\ell_2}), \qquad & pq\in 2 \mathbb{Z}, \\ \cA^{q\ell_1^2, p\ell_2^2} \boxtimes \cA^{\ell_2^2,1}, \qquad & pq\in 2 \mathbb{Z} +1\,.
\end{array}\right.
\end{equation}

\bibliographystyle{apsrev4-2_custom}
\bibliography{Hall.bib}

\end{document}